\documentclass[final,3p,times]{elsarticle}
\usepackage{lineno}
\usepackage{amssymb}
\usepackage{amsmath}
\usepackage{hyperref}
\usepackage{mathtools}

 \usepackage{relsize}
\usepackage{thmtools}
\usepackage{amsfonts}
\usepackage{float}
\usepackage{amsthm}

\usepackage{mathabx}
\usepackage{graphicx}
\usepackage{cleveref}
\usepackage{tikz}
\usetikzlibrary{shapes}
\usetikzlibrary{decorations.pathreplacing,calligraphy}
\usepackage{apxproof}[forwardlinking=yes]
\usepackage{tikz-cd} 
\usepackage{enumitem}
\usepackage{apxproof}
\usepackage{comment}
\usepackage[small,caps]{complexity}
\usepackage{todonotes}
\usepackage[normalem]{ulem}
\usepackage{soul}

\setdescription{leftmargin=0pt}
\usepackage{wrapfig}
\usepackage{pgf}
\usetikzlibrary{arrows,automata}
\usepackage{xcolor}

\makeatletter
\newcommand{\subsetsim}{\mathrel{\mathpalette\subset@sim\relax}}
\newcommand{\subset@sim}[2]{%
  \vtop{\offinterlineskip\m@th
    \ialign{\hfil$#1##$\hfil\cr
      \subset\cr\noalign{\kern1pt}\sim\cr
    }%
  }%
}
\makeatother

\makeatletter
\newcommand{\supsetsim}{\mathrel{\mathpalette\supset@sim\relax}}
\newcommand{\supset@sim}[2]{%
  \vtop{\offinterlineskip\m@th
    \ialign{\hfil$#1##$\hfil\cr
      \supset\cr\noalign{\kern1pt}\sim\cr
    }%
  }%
}
\makeatother

\newcommand\spans[1]{\ensuremath{#1^*}}
\newcommand\roots[1]{\ensuremath{R(#1)}}
\newcommand\pref[1]{\ensuremath{#1^\prime}}
\newcommand\suff[1]{\ensuremath{#1^{\prime\prime}}}
\renewcommand\implies{\ensuremath{\Rightarrow}}
\newcommand\as{\ensuremath{\alpha}}
\newcommand\bs{\ensuremath{\beta}}
\newcommand\cs{\ensuremath{\gamma}}

\newcommand{\ic}[2]{\ensuremath{\mathit{lshift}_{#2}(#1)}}

\usetikzlibrary{automata,positioning}

\newcommand\calC{\mathcal{C}}
\newcommand\calT{\mathcal{T}}
\newcommand\calF{\mathcal{F}}
\newcommand\fst{\mathit{fst}}
\newcommand\snd{\mathit{snd}}

\newcommand{\calK}{\mathcal{K}}
\newcommand{\calM}{\mathcal{M}}

\newcommand{\calW}{\mathcal{W}}

\newcommand\bfM{\mathbf{M}}

\newcommand\N{\mathbb{N}}

\newcommand\pair{\ensuremath{P}}

\newcommand\dom{\mathit{dom}}

\newcommand\delay{\mathit{delay}}

\renewcommand\({\left(}
\renewcommand\){\right)}

\newcommand{\saina}[2]{}

\newtheorem{theorem}{Theorem}	
\newtheorem{example}{Example}
\newtheorem{definition}{Definition}
\newtheorem{question}{Question}
\newtheorem{remark}{Remark}
\newtheorem{claim}{Claim}

\newtheoremrep{theorem}{Theorem}
\newtheoremrep{corollary}{Corollary}%[section]
\newtheoremrep{proposition}{Proposition}%[section]
\newtheoremrep{lemma}{Lemma}

\journal{Information and Computation}

\begin{document}

\begin{frontmatter}

\title{Deciding Conjugacy of a Rational Relation}

\author[inst1,inst2]{C. Aiswarya}
\ead{aiswarya@cmi.ac.in}

\author[inst3]{Amaldev Manuel \fnref{fund}}
\ead{amal@iitgoa.ac.in}
\fntext[fund]{Supported by the DST SERB MATRICS grant for the project \textit{Deciding closeness of finite state transducers} [MTR/2022/000628]}		

\author[inst3]{Saina Sunny}
\ead{saina19231102@iitgoa.ac.in}

\affiliation[inst1]{organization={Chennai Mathematical Institute},
            country={India}}
            
\affiliation[inst2]{organization={CNRS, ReLaX, IRL 2000},
            country={India}}
            
\affiliation[inst3]{organization={Indian Institute of Technology Goa},
            country={India}}

\begin{abstract}

A rational relation is conjugate if every pair of words in the relation are conjugates, i.e., cyclic shifts of each other. We show that checking whether a rational relation is conjugate is decidable.

We assume that the rational relation is given as a rational expression over pairs of words. Every rational expression is effectively equivalent to a sum of sumfree expressions, possibly with an exponential size blow-up. The general problem reduces to determining the conjugacy of sumfree rational expressions.
To solve this specific case, we generalise the classical Lyndon-Sch\"utzenberger theorem from word combinatorics that equates conjugacy of a pair of words $(u,v)$ and the existence of a word $z$ (called a \emph{witness}) such that $uz=zv$ (or $zu=vz$). We give two generalisations of this result. We say that a set of conjugate pairs has a \emph{common witness} if there is a word that is a witness for every pair in the set. The generalisations are as follows:
\begin{enumerate}
\item If $G$ is an arbitrary set of conjugate pairs, then $G^*$ is conjugate if and only if there is a common witness for $G$.  Moreover, a word is a common witness for $G^*$ if and only if it is a common witness for $G$.  %(\Cref{conjugacy}).
\item If $G_1^*, \ldots, G_k^*$, $k > 0$ are arbitrary sets of conjugate pairs and $(\alpha_0,\beta_0), \ldots, (\alpha_k,\beta_k)$ are arbitrary pairs of words, then the set of words \[G = (\alpha_0,\beta_0) G_1^* (\alpha_1,\beta_1) \cdots (\alpha_{k-1},\beta_{k-1})G_k^*(\alpha_k,\beta_k)\] is conjugate if, and only if, 
\begin{itemize}
    \item either it has a common witness,
    \item or there exists words $w \sim \as_0\ldots \as_k$ and $\rho$ such that  
    each pair generated by $M$ is conjugate to a word in $w\rho^*$.  
    \end{itemize}

Moreover, the common witnesses of $G$ are computable in polynomial time from the common witnesses of $G_1^*, \ldots, G_k^*$.% (\Cref{generalmonomialwitness}).
\end{enumerate}
Further, the set of common witnesses, if it exists, is computable inductively. 
Using the above results,
we give a polynomial-time algorithm for checking the conjugacy of a sumfree expression and an exponential-time procedure for the general problem.

\end{abstract}

%% Keywords

\end{frontmatter}

\tableofcontents
%\linenumbers

% !TEX root = main.tex
\section{Introduction
\label{sec:introduction}}
Conjugacy of two elements $u$ and $v$ in a group can be defined as any of the following equivalent cases:
\begin{enumerate}
\item $uz = zv$ for some $z$,
\item $u = xy$ and $v = yx$ for some $x,y$. 
\end{enumerate}
The conjugacy problem asks if a  given pair of elements in a finitely presented group (typically infinite) is conjugate.
It along with the word and isomorphism problems constitute the classical triad of decision problems on groups identified by Dehn in 1912 \cite{dehn}. Dehn's prescient choice turned out to be instrumental not only in mathematics, but also to the theory of  semigroups/monoids and automata in computer science. It turns out that the above conditions are equivalent for free monoids (i.e., when $u,v,z,x,y$ are taken to be words over some finite alphabet). This is the well-known second theorem of Lyndon-Sch\"utzenberger. But unlike in the case of groups where condition (1) is taken to be the definition of conjugacy, in the case of monoids condition (2) is taken as the definition of conjugacy. Hence the statement reads the following way. 

\begin{theorem} [Proposition 1.3.4 of \cite{Lot83}]\label{uz=zv}
A  pair of nonempty words $(u,v)$ is conjugate if and only if there exists a word $z$ such that $uz=zv$. Moreover, $z \in (xy)^*x$ where $x$ and $y$ are such that $u = xy$ and $v = yx$.
\end{theorem}

\smallskip

Conjugacy problem is solvable in polynomial time over free monoids and free groups. We consider a generalisation of the problem to a finitely-presented possibly infinite set of pairs. Let $A$ be a finite alphabet. A relation $R \subseteq A^* \times A^*$ over the free monoid $A^*$ is conjugate if each pair $(u,v)\in R$ is conjugate. Consider the following decision version: 
\begin{enumerate}
\item \emph{Given a relation $R$ over $A^*$, is it conjugate?}
\end{enumerate}
Of particular interest is when $R$ is automata-definable because of motivations detailed later. First we recall the class of rational relations. The family of \emph{rational subsets} of a monoid $M$ is the smallest class containing $\emptyset$, all singleton subsets of $M$ and closed under union, product and Kleene closure. A natural way to present a rational subsets of $M$ is as a rational expression: $\emptyset, m\in M$ are rational expressions, and if $E_1, E_2$ are rational expressions then $E_1\cdot E_2$, $E_1 + E_2$, and $E_1^*$ are also rational expressions. A rational relation over $A^*$ is a rational subset of the product monoid $A^*\times A^*$. Coincidentally, rational relations are precisely those that are defined by nondeterministic finite state transducers. Checking conjugacy of rational relation is equivalent to checking if the input output pairs of a relation defined by a nondeterministic transducer are conjugates. A rational expression over $A^\ast \times A^\ast$ is conjugate if the rational relation it represents is conjugate.

\begin{example}
The rational expression $E_1= (\epsilon,a)(ab,ba)^*(a,\epsilon)$ denotes the relation $\{((ab)^na,a(ba)^n) \mid n \geq 0\}$. The expression $ E_2=((a,aa) + (b,\epsilon))^*$ represents $ \{(u,v) \mid $  $v$ is obtained  from $u$ by duplicating  $a$'s and discarding  $b$'s$\}$.
The expression $E_1$ is conjugate,  but $E_2$ is not.
\end{example}

A strong justification for the above problem comes from the theory of word transducers. Checking a number of properties of word transducers, for instance sequentiality (can the given transducer be input-determinised?) 
or  finite sequentiality (is the given transducer equivalent to a disjoint union of input-deterministic transducers?), bounded edit distance \cite{editdistance} (is the edit distance between the respective outputs of the given transducers bounded?) etc.~amounts to checking conjugacy of the rational relations defined by the strongly connected components of the transducer and certain specific properties of the underlying acyclic graph of strongly connected components. Loosely speaking, conjugacy of the relations defined by the strongly connected components imply that the loops of the transducer are pumpable. Likewise, given two functional transducers with identical domain, checking if the output pairs of these transducers are conjugates on all input, reduces to checking conjugacy of the rational relation of output pairs defined by their product automata. A converse result also holds: checking conjugacy of rational relation is equivalent to checking conjugacy of output pairs of two sequential transducer with identical domain by virtue of Nivat's theorem \cite{nivat1968transductions,editdistance}.

Our main result is summarised by the following theorem.
\begin{theorem}\label{maintheorem}
Conjugacy of rational relations is decidable.
\end{theorem}

The union operation of rational relations preserves conjugacy, i.e., if $R_1$ and $R_2$ are conjugate, then $R_1 \cup R_2$ is also conjugate. However, conjugacy is not preserved under the product and Kleene closure.
\begin{example}\label{eg:closure}
	Consider $E_1 = (\mathit{ab},\mathit{ba})$,  $E_2 = (\mathit{ca},\mathit{ac})$ and $E_3 = (\mathit{ac},\mathit{ca})$, all of which are  conjugates. $E_1 \cdot E_3$ is conjugate but  $E_1 \cdot E_2 $ is not.  Also $(E_1 + E_3)^\ast $ is conjugate, but $(E_1 + E_2)^\ast$ is not. 
\end{example}
Assume that the rational relation is given by a rational expression. Every rational expression is effectively equivalent to a sum of sumfree expressions (called sumfree normal form (SNF)), by inductively rewriting the  expression using the identities $(a+b)^*=(a^*b^*)^*$ and $(a+b)\cdot (c+d) = ac+ad+bc+bd$ (see \Cref{sumfreere}). This fact is the rational-expression analogue of the factorisation forest theorem of Simon, a deep result from the theory of finite semigroups \cite{simon1990factorization}. Rewriting a rational expression in SNF may result in an exponential blow-up, both in the number of summands and the size of each summand. 
\begin{example}
For the expression $E= (a + b)^{n}$ for some $n >0$, it can be shown that any equivalent expression in SNF will have at least $2^n$ summands. For $E' = \$(E\#)^*$, any equivalent SNF expression will have at least one summand of exponential size, and the expression $E\cdot E'$ in SNF will have exponentially many summands of exponential size.
\end{example}

For proving \Cref{maintheorem} we only need to decide the conjugacy of a rational relation given by a sumfree expression. The decidability of conjugacy of a sumfree expression hinges on the notion of a common witness of a relation, inherited from Lyndon-Sch\"utzenberger's theorem. A \emph{witness} of a conjugate pair $(u,v)$ is a word $z$ such that either $uz = zv$ (\emph{inner witness}) or $zu = vz$ (\emph{outer witness}). A word $z$ is a \emph{common inner (resp.~outer) witness} of a relation, if for every pair $(u,v)$ in the relation, $z$ is an inner witness (\textit{resp.}~outer witness) of $(u,v)$. By \Cref{uz=zv}, if a relation has a common inner witness or common outer witness then it is conjugate. However, the converse is easily shown to be false.  

We show that a sumfree rational relation is conjugate if and only if either it has a common witness --- either a common inner witness or a common outer witness, but not necessarily both --- or there exist words $u$ and $v$ such that every pair generated by the relation is a cyclic shift of some word in $uv^*$. This characterisation of conjugacy is a main contribution of our paper.  It is in fact a generalisation of Lyndon-Sch\"utzenberger theorem characterising conjugacy of two words.

There are two interesting questions regarding common witnesses:
\begin{enumerate}
\item[I.] \emph{Is there a common witness for the relation $R$?}
\item[II.] \emph{Given a word $z$, is it a common witness of $R$?}	
\end{enumerate}
Question II proves to be comparatively more tractable, as it can be reduced to verifying whether the rational relation $R^\prime = \{(uz,zv) \mid (u,v) \in R\}$ (or, $R^\prime = \{(zu,vz) \mid (u,v) \in R\}$) consists of only identical pairs \cite{sakarovitch2009elements}. In fact, the decidability of the twinning property of a transducer is connected to Question II.  It is further elaborated in \Cref{Subsec:relatedwork}.

Question I, on the other hand, is more difficult \emph{a priori} as we do not have a bound on the size of a possible common witness. We provide a decision procedure for Question I. This is another main contribution of the article. Our characterisation of conjugacy via common witness, together with this procedure, yields an algorithm for deciding conjugacy.

In the rest of this section, we give an overview of determining the conjugacy of sumfree expressions. Subsequently, it is argued that decidability follows from two specific questions (\Cref{qn:kleene} and \Cref{qn:monoid}). Finally, we discuss related works. 
                                           
\subsection{Conjugacy of a Sumfree Expression}
\label{Subsec:decidability}
A rational expression is \emph{sumfree} if it does not use the union (i.e., $+$). 
The set of sumfree expressions is formally defined as a hierarchy. Fix a monoid $\bfM=(M, \cdot, 1)$. Given a class $\calC$ of expressions over $\bfM$, the \emph{Kleene closure} of $\calC$, denoted as $\calK \calC$, is the class of expressions 
$$\calK \calC = \calC \cup \{ E^* \mid E \in \calC\}.$$ 
The \emph{monoid closure} of $\calC$, denoted as $\calM \calC$, is the class of expressions 
$$\calM \calC = \calC \cup \{ E_1 \cdots E_k \mid E_i \in \calC \text{ for each $1\leq i \leq k$ and $k\in\mathbb{N}$}\}.$$

\begin{definition}[Sumfree Expression]
The family $\calF$ of sumfree expressions is defined inductively. Let $\calF_0 = M \cup \{\emptyset\}$ and
$\calF_{i+1} = \calM\calK\calF_i \ $  for each $i\geq 0$. We define
$$\calF = \bigcup_{i\geq 0} \calF_i.$$
The \emph{star height} of an expression $E$ is the smallest $k \in \mathbb{N}$ such that $E$ belongs to $\calF_k$.
\end{definition}

Over the free monoid $A^*$, the set of expressions $\calF_0$ is  $A^* \cup \{\emptyset\}$ and $\calK \calF_0$ is the set of expressions $\calF_0 \cup \{w^* \mid w \in A^* \}$ (for convenience we assume that $\emptyset$ is not used in any other expression other than $\emptyset$ itself).  
It is not difficult to see that $\calM\calK\calF_0$ is the set of expressions $\calK \calF_0 \cup \{u_1v_1^*u_2v_2^* \cdots u_{k}v_k^*u_{k+1} \mid u_i,v_i \in A^*, k\in \mathbb{N}\}$.

We now proceed to solve the conjugacy of sumfree expressions. We use pairs of lowercase Greek letters $(\alpha,\beta)$ with suitable modifications to denote pairs of words over $A^* \times A^*$. Clearly $\emptyset$ and $(\epsilon,\epsilon)$ are conjugates. For an expression of the form $(\alpha,\beta)$, it is straightforward to check conjugacy (see \Cref{isconjugate}). Thus, the conjugacy problem is decidable for the class of expressions $\calF_0$. 

To show the decidability of the conjugacy problem for the whole family $\calF$, it suffices to show that if the problem is decidable for $\calF_i$, $i\geq 0$, then it is also decidable for $\calK \calF_i$ and $\calF_{i+1}=\calM \calK \calF_i$. Then by induction on $i$ the decidability extends to the whole family $\calF$.

Assume that conjugacy is decidable for $\calF_i$. Let $E$ be an expression in $\calF_i$ and hence $E^* \in \calK\calF_i$. Since $L(E) \subseteq L(E^*)$,
\begin{proposition}\label{Econj}
If the expression $E^*$ is conjugate, then $E$ is conjugate. 
\end{proposition}

Because conjugacy is decidable for $\calF_i$, we can check whether $E$ is conjugate. Therefore, to show the decidability of conjugacy for $\calK\calF_i$, it suffices to show the decidability of the following question.

\begin{question}[Conjugacy of Kleene Closures]
\label{qn:kleene}
Given a conjugate sumfree expression $E$, is $E^*$ conjugate?
\end{question}

Next, assume that conjugacy is decidable for $\calK \calF_i$. Let $E = (\alpha_0,\beta_0) E_1^* (\alpha_1,\beta_1) \cdots E_k^*(\alpha_k,\beta_k)$ be an expression in $\calM\calK\calF_i$ where $E_1^*,\ldots,E_k^*$ are from $\calK \calF_i$.
Analogous to the case of Kleene closures, $E$ is conjugate only if $E_1^*,\ldots,E_k^*$ are conjugate, as the next lemma shows.

\begin{lemma}\label{boundedkleenestar}
If the expression $E=(\alpha_0,\beta_0) F^* (\alpha_1,\beta_1)$ is conjugate, then $F^*$ is conjugate.
\end{lemma}

\begin{proof}

If $F^*$ is an empty set, then it is conjugate. Otherwise, assume that $(u,v)$ is a nonempty pair that belongs to $F^*$. Therefore, $(u^\ell,v^\ell)$ for each $\ell \geq 0$ also belongs to $F^*$. We can safely assume that $|u| = |v|$, otherwise each iteration will increase the difference in length between $u^\ell$ and $v^\ell$, leading to nonconjugacy of $E$.

Let $k = |\alpha_0| + |\beta_0| + |\alpha_1| + |\beta_1|$. Consider the pair $(\as_0,\bs_0)(u^\ell,v^\ell)(\as_1,\bs_1)$ where $\ell$ is some value much larger than $k$, say $2^k$. Since $\ell$ is much larger than $k$ and $(\as_0u^\ell \as_1,\bs_0v^\ell \bs_1)$ is conjugate, there exist large factors of $u^\ell$ and $v^\ell$ that match as shown in Figure \ref{uvlarge}. Since $|u| = |v|$, we can infer that $u$ is a factor of $vv$, and $v$ is a factor of $uu$. 
\begin{figure}[t]

\centering
\begin{tikzpicture}

\draw[red!50,line width=.55mm,loosely dotted] [-](-1,2) -- (0,2)[anchor=south, xshift=-1.5cm, yshift=0cm]node{};
\draw[red!50,line width=.55mm,densely dotted] [|-](0,2) -- (2.5,2)[anchor=south, xshift=-1.5cm, yshift=0cm]node{$u$};
\draw[red!50,line width=.55mm,densely dotted] [|-](2.5,2) -- (5,2)[anchor=south, xshift=-1.5cm, yshift=0cm]node{$u$};
\draw[red!50,line width=.55mm,densely dotted] [|-](5,2) -- (7.5,2)[anchor=south, xshift=-1.5cm, yshift=0cm]node{$u$};
\draw[red!50,line width=.55mm,densely dotted] [|-](7.5,2) -- (10,2)[anchor=south, xshift=-1.5cm, yshift=0cm]node{$u$};
\draw[red!50,line width=.55mm,loosely dotted] [|-](10,2) -- (11,2)[anchor=south, xshift=-1.5cm, yshift=0cm]node{};

\draw[blue!50,line width=.55mm,loosely dotted] [-](.5,0) -- (1.5,0)[anchor=north, xshift=-1.5cm, yshift=0cm]node{};
\draw[blue!50,line width=.55mm,densely dotted] [|-](1.5,0) -- (4,0)[anchor=north, xshift=-1.5cm, yshift=0cm]node{$v$};
\draw[blue!50,line width=.55mm,densely dotted] [|-](4,0) -- (6.5,0)[anchor=north, xshift=-1.5cm, yshift=0cm]node{$v$};
\draw[blue!50,line width=.55mm,densely dotted] [|-](6.5,0) -- (9,0)[anchor=north, xshift=-1.5cm, yshift=0cm]node{$v$};
\draw[blue!50,line width=.55mm,densely dotted] [|-](9,0) -- (11.5,0)[anchor=north, xshift=-1.5cm, yshift=0cm]node{$v$};
\draw[blue!50,line width=.55mm,loosely dotted] [|-](11.5,0) -- (12.5,0)[anchor=north, xshift=-1.5cm, yshift=0cm]node{};

\draw [decorate,
decoration = {calligraphic brace,
raise=3pt,
amplitude= 5pt,mirror
}] (2.5,2) -- (4,2);
\draw [decorate,
decoration = {calligraphic brace,
raise=3pt,
amplitude= 5pt,mirror
}] (4,2) -- (5,2);

\draw [decorate,
decoration = {calligraphic brace,
raise=3pt,
amplitude= 5pt,mirror
}] (5,2) -- (6.5,2);

\draw [decorate,
decoration = {calligraphic brace,
raise=3pt,
amplitude= 5pt,mirror
}] (6.6,2) -- (7.5,2);

\draw [decorate,
decoration = {calligraphic brace,
raise=3pt,
amplitude= 5pt
}] (2.5,0) -- (4,0);
\draw [decorate,
decoration = {calligraphic brace,
raise=3pt,
amplitude= 5pt
}] (4,0) -- (5,0);

\draw [decorate,
decoration = {calligraphic brace,
raise=3pt,
amplitude= 5pt
}] (5,0) -- (6.5,0);

\draw [decorate,
decoration = {calligraphic brace,
raise=3pt,
amplitude= 5pt
}] (6.6,0) -- (7.5,0);

\draw[draw=none] [-](2.5,2) -- (4,2)[anchor=south, xshift=-0.7cm, yshift=-.7cm]node{$p$};
\draw[draw=none] [-](4,2) -- (5,2)[anchor=south, xshift=-0.45cm, yshift=-.7cm]node{$x$};
\draw[draw=none] [-](5,2) -- (6.5,2)[anchor=south, xshift=-0.7cm, yshift=-.7cm]node{$y$};
\draw[draw=none] [-](6.5,2) -- (7.5,2)[anchor=south, xshift=-0.45cm, yshift=-.7cm]node{$q$};

\draw[draw=none] [-](2.5,0) -- (4,0)[anchor=south, xshift=-0.7cm, yshift=.25cm]node{$p$};
\draw[draw=none] [-](4,0) -- (5,0)[anchor=south, xshift=-0.45cm, yshift=.25cm]node{$x$};
\draw[draw=none] [-](5,0) -- (6.5,0)[anchor=south, xshift=-0.7cm, yshift=.25cm]node{$y$};
\draw[draw=none] [-](6.5,0) -- (7.5,0)[anchor=south, xshift=-0.45cm, yshift=.25cm]node{$q$};

\end{tikzpicture}
\caption{$v$ as infix of $uu$.}

\label{uvlarge}
\end{figure}

Since $v$ is an infix of $uu$, the following holds as shown in Figure \ref{uvlarge}. There exist words $x,y,p,$ and $q$ such that $v = xy$ and $u = px = yq$.
Since $|u| = |v|$, length of $p$ and length of $y$ are the same, that implies $p=y$ (since $u = px = yq)$. Therefore, $u = yx$. Hence $u$ and $v$ are conjugate words.
Since the pair $(u,v)$ was arbitrary, $F^*$ is conjugate.
\end{proof}

We can generalize the above lemma to the general form of sumfree expressions.
\begin{corollary}\label{boundedconj}
If the expression $E = (\alpha_0,\beta_0) E_1^* (\alpha_1,\beta_1) \cdots (\alpha_{k-1},\beta_{k-1})E_k^*(\alpha_k,\beta_k)$, for $k > 0$, is conjugate, then each of $E_1^*, \ldots, E_k^*$ is conjugate.
\end{corollary}

\begin{proof}
If $E$ is conjugate, then for each $i \in \{1, \ldots, k\}$, $$(\as_0 \cdots \as_{i-1},\bs_0 \cdots \bs_{i-1}) E_i^* (\as_i \cdots \as_k,\bs_i \cdots \bs_k) \subseteq E$$ is conjugate. Therefore, from \Cref{boundedkleenestar} we get that each of $E_1^*, \ldots, E_k^*$ is conjugate.
\end{proof}

Since the conjugacy of $\calK\calF_i$ is decidable, we can check whether $E_1^*, \ldots, E_k^*$ are conjugate expressions. Thus, to show the decidability of $\calM\calK\calF_i$, it suffices to show the decidability of the following question.

\begin{question}[Conjugacy of Monoid Closures]
\label{qn:monoid}
Given conjugate sumfree expressions $E_1^*, \ldots, E_k^*$, is the expression $E = (\alpha_0,\beta_0) E_1^* (\alpha_1,\beta_1) \cdots (\alpha_{k-1},\beta_{k-1})E_k^*(\alpha_k,\beta_k)$ conjugate?
\end{question}
We show that \Cref{qn:kleene} and \Cref{qn:monoid} can be effectively answered. The idea is to use the notion of common witness that we mentioned in the beginning  (further elaborated in \Cref{def:commonwitness}).

We present two common witness theorems that address the above questions:
\begin{enumerate}
\item \label{Mj:1} Let $G$ be an arbitrary set of conjugate pairs. The set $G^*$ is conjugate if and only if $G$ has a common witness (\Cref{conjugacy}). 
\item\label{Mj:2} Let $G_1^*, \ldots, G_k^*$, $k > 0$, be arbitrary sets of conjugate pairs. The set \[G=(\alpha_0,\beta_0) G_1^* (\alpha_1,\beta_1) \cdots (\alpha_{k-1},\beta_{k-1})G_k^*(\alpha_k,\beta_k),\] called a \emph{sumfree set}, is conjugate if and only if either it has a common witness, or each $G_i^*$ has a common witness and there exist words $u,v$ such that each pair of words in $G$ is a cyclic shift of a word in $uv^*$. (\Cref{thm:conjugateMonoidClosure}).

\end{enumerate}

\begin{remark}
Note that the assumption of conjugacy of the sets $G, G_1^*, \ldots, G_k^*$ is not necessary. However, if they are not conjugate then the corresponding sets will neither have a common witness nor be conjugate, and the statements will be vacuously true (Since \Cref{Econj} and \Cref{boundedconj} also hold for arbitrary sets).
\end{remark}
\Cref{Mj:2} is a generalisation of \Cref{Mj:1}, and its proof relies on \Cref{Mj:1}. Both theorems are generalisations of the Lyndon-Sch\"utzenberger theorem. 

In the case of rational sumfree expressions of pairs the above theorems are \emph{effective}. In \Cref{Subsec:computingwitness} we show that a common witness of a rational sumfree expression, if exists, is computable in polynomial time in the length of the expression. Finally, we give a decision procedure proving our main theorem (\Cref{maintheorem}) in \Cref{sec:decidingconj}.

\subsection{Related Work}
\label{Subsec:relatedwork}
\subparagraph*{Conjugate Post Correspondence Problem:} A problem much related to \Cref{conjugacy} is the \emph{Conjugate Post Correspondence problem}: given a finite set of pairs $G$, does there exist of a pair $(u,v) \in G^*$ such that $u$ and $v$ are conjugate? This problem is shown to be undecidable by reduction to the word problem of a special type of semi-Thue systems \cite{ConjugatePCP}. In \Cref{Sec:cwtheorems}, we show that the universal version of this problem --- checking if all the pairs in $G^*$ are conjugate --- is decidable.

\subparagraph*{Edit-Distance between Transducers:}
Distance between two functional transducers are studied in \cite{editdistance}. 
For a metric $d$ on words,  the distance between two functional transducers $\calT_1,\calT_2$ from $A^*$ to $B^*$ is given by  
\[d(\calT_1,\calT_2) = \begin{cases} \sup \left \{\,  d(\calT_1(w),\calT_2(w)) \,\mid\, w \in \dom(\calT_1) \,\right \} & \text{ if $\dom(\calT_1) = \dom(\calT_2)$} \\
 \infty & \text{ otherwise }
 \end{cases}\]
Two transducers are said to be \emph{close} if the distance between them is finite ($<$$\infty$). This allows to compare two transducers beyond their equivalence.
One important class of metrics over words is edit distance, in particular Levenshtein edit distance. Levenshtein edit distance between two words is the minimum number of insertions, deletions, and substitutions required to convert one word to another. In the case of edit distances, closeness guarantees that the output of one transducer can be converted to the output of another by doing a bounded number of edits. 
It is shown that two functional transducers with identical domain are close with respect to Levenshtein edit distance if and only if all the loops in the cartesian product of the transducers are conjugates \cite{editdistance}. Our result gives an algorithm for checking the latter, and hence decidability of closeness of transducers follows.

\subparagraph*{Twinning and Sequentiality:}\label{Twinning}
Rational functions are functions defined by a finite state transducer. A rational function is \emph{sequential} if it can be realised by a sequential transducer, i.e., those that are deterministic in the input. These were originally called subsequential in the literature by Sch\"utzenberger \cite{schuetzenberger1977variante}. 
 Sequentiality of rational functions is a decidable property due to a topological characterisation called the \emph{twinning property} by Choffrut \cite{choffrut1977caracterisation}.

A transducer $\calT$ from $A^*$ to $B^*$ is an automaton over $A^* \times B^*$. A transition of $\calT$ from state $p$ to state $q$ is of the form $(p,(u,v),q)$ where the word $u \in A^*$ is called the {input} and the word $v \in B^*$ is called the {output}. A path from state $p$ to $q$ on an input word $w$ producing an output word $x$ is represented as $p \xrightarrow{w \mid x} q$. The transducer $\calT$ realises the rational relation $\{ (w,x) \mid q_0 \xrightarrow{w \mid x} q_f\}$ over $A^* \times B^*$ where $q_0,q_f$ is an initial and a final state respectively.

The \emph{delay} between two words $u$ and $v$, denoted by $\delay(u,v)$, is the pair $(u',v')$ such that $u = \ell u'$ and $v=\ell v'$ where $\ell$ is the longest common prefix of $u$ and $v$.

A  transducer with the initial state $q_0$ is \emph{twinning} if for all states $p,q$ and for all words $w_1,w_2 \in A^*$ and $x,y,u,v \in B^*$, if $q_0 \xrightarrow{w_1 \mid x} p \xrightarrow{w_2 \mid u} p$ and $q_0 \xrightarrow{w_1 \mid y} q \xrightarrow{w_2 \mid v} q$, then $\delay(x,y) = \delay(xu,yv)$.  
 This is equivalent to either $u = v = \epsilon$, or $u \neq \epsilon \neq v$ and $u$ and $v$ are conjugates with $\delay(x,y)$ being a witness of $(u,v)$ (Proposition 6.2 of \cite{Berstel}).

Since the twinning property compares paths with the same input label, an equivalent definition for twinning can be defined on the square of the transducer \cite{beal2003squaring,lombardy2006sequential}. The \emph{square} of a transducer $\calT$, denoted by $\calT^2$, is a cartesian product of $\calT$ by itself, equivalent to the transducer from $A^*$ into $B^* \times B^*$.
The original definition of twinning has the following equivalent form.
Let $\calT$ be a trim transducer. Two states $p$ and $q$ of $\calT$ are \emph{twinned} if whenever $(u,v)$ is a nonempty output pair of a loop in $\calT^2$ rooted at a reachable state $(p,q)$, and for any path from initial state to $(p,q)$ in $\calT^2$ with output $(x,y)$, 
the following holds: $\delay(x,y) = \delay(xu,yv)$, or equivalently, $\delay(x,y)$ is a witness of $(u,v)$.
A transducer $\calT$ is twinning if, for any pair of states $p$ and $q$ that are reachable in $\calT^2$, the states $p$ and $q$ are twinned.

Thus checking twinning property reduces to computing a number of instances of the following problem:  decide if a given rational relation $R$ is conjugate with a given witness $z$. This is straightforward to check: verify whether the rational relation $R^\prime = \{(uz,zv) \mid (u,v) \in R\}$ (or, $R^\prime = \{(zu,vz) \mid (u,v) \in R\}$) consists of only identical pairs.

Generalisation of the twinning property called \emph{weak twinning} is used to characterise \emph{multi-sequential} (also called \emph{plurisubsequential} or \emph{finitely sequential}) functions \cite{choffrut1986decomposition} and relations \cite{jecker2015multi}. This property also can be stated as conjugacy of suitably defined rational relations with respect to certain predetermined witnesses, and hence decidable albeit higher complexity.

The results in this paper holds the potential to define more general properties of transducers where the witnesses are not known a priori, for instance determinising approximately.

\subparagraph*{Others:} 

A generalisation of Lyndon-Sch\"utzenberger's theorem to infinite sets, though with no comparison to ours, is considered in \cite{OnConjugacyOfLanguages, karhumaki2001combinatorial}, where solutions to the language equation $XZ=ZY$, where $X,Y,Z$ are sets of words, are given for special cases. The general solution is still open.

\subsection{Organisation of the Paper}
\label{Subsec:organisation}

In \Cref{Sec:prelim}, we revisit the standard tools from combinatorics of words required to state and prove our main theorems. We present the common witness theorems for addressing  \Cref{qn:kleene} and \Cref{qn:monoid}, along with the proofs of the easier directions in \Cref{Sec:cwtheorems}. However, the difficult directions require a detailed case analysis, and we complete the proof of common witness theorem for Kleene closure in \Cref{sec:proofcw}. We prove the common witness and conjugacy theorem of monoid closures in \Cref{sec:conjugacy_monomial} and \Cref{sec:conjugateMonoid} respectively. We outline the decision procedure for computing the witness in \Cref{Subsec:computingwitness}, and deciding conjugacy in \Cref{sec:decidingconj}.
In \Cref{Sec:conclusion}, we state some future directions and conclude.

% !TEX root = main.tex
\section{Combinatorial Tools on Words}
\label{Sec:prelim}
In this section, we introduce some preliminary notions and results that we need for the rest of the paper. 

Let $\mathbb N$ denote the set $\{1, 2, \dots\}$. For $i,j \in \mathbb{N} \cup \{0\}$, define $[i,j]$ as the set $\{i, i+1, \dots, j\}$ if $i \le j$, and as $\{i, i-1, \dots, j\}$ otherwise. Given a word $u$ and $k \in \mathbb N \cup \{\omega\}$, we denote by $u^k$ the concatenation of  $u$ with itself $k$ times. A word $u$ is called a \emph{factor} (respectively \emph{prefix}, \emph{suffix}) of a word $v$, if there exist words $x,y \in A^*$ such that $v = xuy$ (respectively $v = uy$, $v = xu$). Let $u[i \cdot\cdot j]$ denote the factor of $u$ from index $i$ to $j$ where $1 \leq i \leq j \leq |u|$. A \emph{cyclic shift} of a word $w=uv$ is a word $vu$ for some words $u,v \in A^*$.  Let $\ic{u}{i}$, denote the word obtained after $i$ left cyclic shifts of a word $u$ for $i \geq 0$. 
If $u$ and $v$ are words such that $u$ is a prefix of $v$, the \emph{left quotient} of $v$ by $u$, denoted by $u^{-1}v$, is the word $x$ such that $v = ux$. Similarly, the \emph{right quotient} of $v = xu$ by $u$, denoted as $vu^{-1}$, is the word $x$.

Given a set of pairs of words $G$, let $\fst(G)$ denote the projection of $G$ onto the first component, i.e., $\fst(G) = \{u \mid (u,v) \in G\}$.   Similarly,  $\snd(G)$ denote the projection of $G$ onto the second component, i.e., $\snd(G) = \{v \mid (u,v) \in G\}$.

Two words $u$ and $v$ are conjugate if they are cyclic shifts of each other, i.e., there exist words $x,y$ such that $u=xy$ and $v=yx$. We use $u \sim v$ to denote that $u$ and $v$ are conjugate words. Let $X,Y \subseteq A^*$. We say $X$ is a \emph{conjugate subset} of $Y$, denoted by $X \subsetsim Y$, if for every word $x \in X$ there exists a word $y \in Y$ such that $x \sim y$.  Also, $X$ and $Y$ are \emph{conjugate}, denoted by $X \sim Y$, if $X \subsetsim Y$ and $Y \subsetsim X$.

\subsection{Primitive words, Roots and Conjugacy}
A word is \emph{primitive} if it cannot be expressed as a power of any strictly smaller word. For example, $\mathit{aba}$ is primitive but $\mathit{abab}$ is not. 

A word $\rho$ is called a \emph{primitive root} of a word $u$ if $u = \rho^n$ for $n \geq 1$ and $\rho$ is a primitive word. Every word $u$ has a unique primitive root, denoted by $\rho_u$ (\cite{Lot83}, Proposition 1.3.1). 

\begin{proposition}\label{computeroot}
The primitive root of a word can be computed in time linear in the length of the word.
\end{proposition}

\begin{proof}
For a word $w$, find the first nontrivial occurrence of $w$ in $ww$, and let the offset be the prefix of $ww$ before this occurrence. This can be done in linear time using the classical pattern matching algorithm by Knuth-Morris-Pratt \cite{knuth1977fast}. If the length of the offset divides the length of $w$, the offset is the primitive root of the word $w$.
\end{proof}

\begin{proposition}\label{prop:comb}
Let $u, v$ be nonempty words such that their primitive roots are distinct, i.e., $\rho_u \neq \rho_v$. Let $w = uv$. Then $\rho_w \neq \rho_u$ and $\rho_w \neq \rho_v$.

\end{proposition}
\begin{proof}
Consider the word $w = uv$. We show that $\rho_w \neq \rho_u$ and $\rho_w \neq \rho_v$. Assume for contradiction that $\rho_w = \rho_u$. Then $uv \in (\rho_u)^*$. Since $u \in (\rho_u)^*$, it follows that $v \in (\rho_u)^*$ as well. Since primitive root of a word is unique, we get $\rho_v = \rho_u$, which contradicts the assumption.
%Similarly, assume for contradiction that $\rho_w = \rho_v$. Then $uv \in (\rho_v)^*$ and $v \in (\rho_v)^*$, it follows that $u \in (\rho_v)^*$ as well. Since primitive root of a word is unique, we get $\rho_v = \rho_u$, which contradicts the assumption.
By a similar argument, we get $\rho_w \neq \rho_v$. Therefore, $\rho_w \neq \rho_u$ and $\rho_w \neq \rho_v$, which completes the proof.
\end{proof}
\begin{proposition}
\label{prop:nontriv-concat2}
    Let $u,v\in A^*$ be words such that $\rho_u \sim \rho_v$ and $\rho_u \neq \rho_v$. Then any word \begin{align} w \in (u+v)^*\end{align}
    in which both $u$ and $v$ appear at least once, satisfies $\rho_{w} \not \sim \rho_u$. 
\end{proposition}
\begin{proof}

Suppose, for contradiction, that $\rho_{w} \sim \rho_u$. Then, we have $|\rho_{w}| = |\rho_u| = |\rho_{v}|$.
Since $w \in (\rho_{w})^*$ and both $u$ and $v$ are factors of $w$ with lengths divisible by $|\rho_{w}|$, it follows that $u, v \in (\rho_{w})^*$. By uniqueness of primitive roots, this implies $\rho_{u} = \rho_{w} = \rho_{v}$ contradicting $\rho_{u} \neq \rho_{v}$.
\end{proof}

First Theorem of Lyndon-Sch\"utzenberger relates primitivity and commutativity.

\begin{theorem} [\cite{lyndonSch}, Lemma 3]\label{commute}
Two words $u,v \in A^*$ commute, i.e., $uv =vu$, if and only if they have the same primitive root.
\end{theorem}

We lift the notion of primitive root to a pair and a set of pairs as follows.
\begin{definition}[Primitive Root of a Set of Pairs]
The primitive root of a pair $(u,v)$ is the pair $(\rho_u,\rho_v)$.
The primitive root of a set of pairs $P$, denoted by \roots{P}, is the set of all primitive roots of each pair in $P$, i.e., $\roots{P} = \{ (\rho_{u},\rho_{v}) \mid (u,v) \in P\} \ . $

For example, $\{ (\mathit{ab},\mathit{ba}),(\mathit{bab},\mathit{abb})\}$ is the primitive root of the set $\{(\mathit{abab},\mathit{baba}),(\mathit{bab},\mathit{abb})\}$.
\end{definition}

Recall that a pair of words $(u,v)$ is conjugate if $v$ can be obtained from $u$ by cyclic shifts, i.e., there exist words $x,y$ such that $u=xy$ and $v=yx$. For example, $(\mathit{ aaab}, \mathit{aaba})$ is a conjugate pair with $x=a$ and $y=aab$. It is not difficult to see that conjugacy relation is an equivalence relation on the set of words.

\begin{proposition}\label{isconjugate}
Deciding if a pair of words is conjugate can be done in linear time. 
\end{proposition}
\begin{proof}
Two words $u$ and $v$ are conjugate if and only if either $u=v$ or $|u| = |v|$ and there is exactly one nontrivial occurence of $v$ in $uu$.

If $u$ and $v$ are distinct conjugates, then $u=xy$ and $v=yx$ for some nonempty words $x,y$, implying $uu = xyxy = xvy$, where $v$ occurs nontrivially in $uu$ exactly once. Conversely, if $v$ occurs nontrivially in $uu$ exactly once and $|u| = |v|$, then there exist words $x,y,p$ and $q$ such that $v = xy$ and $u = px = yq$. Since $|u| = |v|$, we get $|p|=|y|$. This implies $p=y$ since $u = px = yq$. Therefore, $u = yx$. Hence $u$ and $v$ are conjugate words.

Therefore, deciding whether a pair of words $(u,v)$ is conjugate can be done by checking if $u =v$. If not, verify $|u| = |v|$ and whether $v$ occurs in $uu$ nontrivially exactly once. This can be done in linear time using the Knuth-Morris-Pratt algorithm \cite{knuth1977fast}.
\end{proof}

A conjugate pair of words and its primitive root are connected.
\begin{proposition}[\cite{ChoffrutKarhumaki97}, Lemma 1]\label{primitiveconjugates}
If a pair $(u,v)$ is conjugate, then its primitive root $(\rho_u,\rho_v)$ is also conjugate. Moreover, $(u,v) = (\rho_u, \rho_v)^n$ for some $n \geq 1$.
\end{proposition}
Below is a fundamental result on words by Fine and Wilf.
\begin{theorem}[Theorem 5 of \cite{ChoffrutKarhumaki97}]\label{finewilf}
Two nonempty words $u$ and $v$  have the same primitive root if and only if the words $u^\omega$ and $v^\omega$ have a common prefix of length $|u| + |v| - \mathit{gcd}(|u|,|v|)$. 
\end{theorem}
The above theorem can be adapted to yield conjugate primitive roots.
\begin{theorem}[\cite{ChoffrutKarhumaki97}]\label{commonfactor}
For any two words $u,v \in A^+$, if $u^\omega$ and $v^\omega$ have a common factor of length at least $|u| + |v| - \mathit{gcd}(|u|,|v|)$, then the primitive roots of $u$ and $v$ are conjugates. 
\end{theorem}
The bound $|u| + |v| - \mathit{gcd}(|u|,|v|)$ is  known as the \emph{Fine and Wilf index} of $u$ and $v$.

\subsection{Cuts and Uniqueness of Cuts of Primitive Pairs}
\begin{definition}[Cut]\label{cut}
A \emph{cut} of a conjugate pair $(u,v)$ is a pair of words $(x,y)$ such that $u=xy$ and $v=yx$. Alternatively, we say that $u$ has a cut at position $|x|$, or equivalently, $v$ has a cut at position $|y|$.

If either $x$ or $y$ is the empty word, then we say the cut is \emph{empty}. Otherwise the cut is \emph{nonempty}.
\end{definition}
For example, the pair $(\mathit{ aabb}, \mathit{bbaa})$ has a cut $(\mathit{aa},\mathit{bb})$. There can be several cuts for a conjugate pair. For instance, the pair $(\mathit{ abab}, \mathit{baba})$ has cuts $(\mathit{a},\mathit{bab})$ and $(\mathit{aba},\mathit{b})$. Empty cuts are possible only for pair of identical words.

\begin{definition}[Matching of a Cut]\label{def:matchcut}
    Let $u$ and $v$ be conjugate words, and let $(x,y)$ be a cut of $(u,v)$ with $|x|=k$. It corresponds to a matching between  positions of $u$ and $v$, namely the circular permutation given by:
    
    \[\begin{array}{rcl@{\hspace{1cm}}rcl}
        |x|+1 &\mapsto & 1& 1 & \mapsto &|y|+1\\
        |x|+2 &\mapsto &2 & 2 &\mapsto &|y|+2\\
        &\vdots&&&\vdots&\\
        |x|+|y|&\mapsto &|y| &|x|&\mapsto &|x|+|y|
    \end{array}\]
This terminology is extended to factors of $u$ an $v$, namely given a cut $(x,y)$ of $u$ and $v$, and factors $u'$ and $v'$ of $u$ and $v$ respectively, we say $u'$ and $v'$ match if the positions of $u'$ and $v'$ match according to the above bijection. 
\end{definition}

If $u$ and $v$ are conjugates and one of them is primitive, by \Cref{primitiveconjugates}, the other is also primitive. A pair $(u,v)$ is primitive if both $u$ and $v$ are primitive words.
For such pairs, their cuts are also special.

\begin{proposition}[Uniqueness of Cuts of Primitive Pairs]\label{primitive}
Let $(u,v)$ be a conjugate primitive pair. If $(u,v)$ is  \emph{distinct}, then $(u,v)$ has a unique cut $(x,y)$. If $(u,v)$ is not distinct (i.e, $u=v$), the only two possible cuts of $(u,v)$ are $(u,\epsilon)$ and $(\epsilon,v)$.
\end{proposition}

\begin{proof}
By definition, if pair $(u,v)$ is conjugate, then there exist a cut $(x,y)$ such that $u=xy$ and $v=yx$. Assume $u$ and $v$ are distinct, and hence $x$ and $y$ have to be nonempty. It suffices to show that $x$ and $y$ are unique if $u$ and $v$ are primitive. 

For the sake of contradiction, assume that $(x,y)$ is not unique, i.e., there exists a different cut $(x^\prime,y^\prime)$ for $(u,v)$, i.e., $u = x^\prime y^\prime$, $v= y^\prime x^\prime$ and $x^\prime \neq x, y^\prime \neq y$. WLOG, assume that $|x| > |x^\prime|$. Therefore there exists a nonempty word $p$ such that $x = x^\prime p$ and $y^\prime=py$. Substituting for $x$ in $v$, we get $$v = yx = yx^\prime p$$ and substituting for $y^\prime$ in $v$, we obtain $$v = y^\prime x^\prime = pyx^\prime \ .$$ Therefore $yx^\prime$ and $p$ commutes. By \Cref{commute}, they have the same primitive root. Since $p$ and $yx'$ are nonempty words, $v$ is a power of some smaller word. Hence $v$ is not primitive and it is a contradiction.

In the case where $u=v$ and $u$ being primitive, two possible cuts are $(u,\epsilon)$ and $(\epsilon,v)$, i.e., the empty cuts. Imagine there is a nonempty cut $(x,y)$. Since $u=v$, it follows that $xy=yx$. Using \Cref{commute}, $u$ and $v$ are powers of a smaller word and hence not primitive. Therefore, when $u$ is primitive and $u=v$, the only possible cuts are the empty cuts.
\end{proof}
It is known that a primitive word cannot be equal to any of its nontrivial cyclic shifts, see for instance \cite{smyth2003computing,bai2016new}. We give a reformulation of this fact below. 
 The statement of the lemma is given in a fashion that is suitable for some of the proofs in the rest of the paper.

\begin{lemma}[Cut Lemma]\label{cutcorr}
Assume $(u,v)$ is a conjugate primitive pair.
\begin{enumerate}[label=\Roman*.]
\item \label{distinct} If $(u,v)$ is a \emph{distinct} pair with the unique cut $(x,y)$, then the following equalities \emph{cannot} hold for any nonempty words $\pref{x}, \suff{x}, \pref{y}$, $\suff{y}$ such that $x=\pref{x}\suff{x}$ and $y=\pref{y}\suff{y}$.
\begin{enumerate}[label=(\alph*), ref= \theenumi{} (\alph*)]
\item\label{cuta} $xy = \suff{x}y\pref{x}$
\item\label{cutb} $xy = \suff{y}x\pref{y}$
\item\label{cutc} $yx = \suff{y}x\pref{y}$
\item\label{cutd} $yx = \suff{x}y\pref{x}$
\item\label{cute} $xy = yx$
\end{enumerate}
\item \label{same} In the special case when $u=v$, there are two empty cuts $(u,\epsilon)$ and $(\epsilon,u)$. In both cases, the equality $u = \suff{u}\pref{u}$ \emph{cannot} hold for any nonempty words $\pref{u}, \suff{u}$ such that $u=\pref{u}\suff{u}$.
\end{enumerate}
\end{lemma}

\begin{proof}
Consider the case when $(u,v)$ is a distinct pair with the unique nonempty cut $(x,y)$. It suffices to show that if any of the equalities hold, there exists a different nonempty cut of the primitive pair $(u,v)$ contradicting \Cref{primitive}.
\begin{enumerate}
\item In the case of \ref{cuta}, the other nonempty cut is $(\suff{x},y\pref{x})$ since $(xy,yx) = (\suff{x}y\pref{x},y\pref{x}\suff{x})$.
\item In the case of \ref{cutb}, the other nonempty cut is $(\suff{y}x,\pref{y})$ since $(xy,yx) = (\suff{y}x\pref{y},\pref{y}\suff{y}x)$.
\item When \ref{cutc} is true, we obtain a different nonempty cut $(x\pref{y},\suff{y})$ because $(xy,yx) = (x\pref{y}\suff{y},\suff{y}x\pref{y})$.
\item If \ref{cutd} holds, the other nonempty cut is $(\pref{x},\suff{x}y)$ since $(xy,yx) = (\pref{x}\suff{x}y,\suff{x}y\pref{x})$.
\item If \ref{cute} holds, the other nonempty cut is $(y,x)$ since $(xy,yx) = (yx,xy)$ and $x \neq y$ (since $u \neq v$).
\end{enumerate}
Consider the special case when $u=v$. If the equality $u = \suff{u}\pref{u}$ holds, then we obtain $u = \pref{u}\suff{u} = \suff{u}\pref{u}$. Therefore, $\pref{u}$ and $\suff{u}$ commutes. Since $\pref{u}$ and $\suff{u}$ are nonempty words, $u$ is a power of some smaller word using \Cref{commute}. Hence $u$ is not primitive and it is a contradiction.
\end{proof}

A consequences of the above lemma is that the cut of the primitive root decides the cuts of its powers. 
\begin{proposition}\label{cutlemmaprop}
Let $(u,v)$ is a \emph{distinct} conjugate primitive pair with the unique cut $(x,y)$. Any cut of the pair $(u^n,v^n)$ for $n \geq 1$ is of the form $((xy)^*x,(yx)^*y)$.
\end{proposition}

\begin{proof}
Let $(u^\prime,v^\prime)= (u^n,v^n)$ for some $n \geq 1$. The lemma is trivially true for $n = 1$ by the uniqueness of cut of primitive pairs by \Cref{primitive}.

Consider the case when $n \geq 2$. Substituting for $u=xy$ and $v=yx$ in $u^\prime$ and $v^\prime$,
\begin{align*}
u^\prime &= \overbrace{u \cdots u}^{n \text{ times}} = xy \cdots xy\\
v^\prime &= v \cdots v = yx \cdots yx
\end{align*}
We show that cut in $u^\prime$ will always be at the end of some $x$ and all other cases leads to one of Cases \ref{cuta} to \ref{cute} of \nameref{cutcorr}.

\paragraph*{Case 1: When the cut is at the end of $y$} I.e., there exists a cut $(p,q)$ for $(u^\prime, v^\prime)$ such that $p \in (xy)^+$. Then
\begin{align}
u^\prime &= \overbrace{ xy\cdots xy}^{p}\overbrace{ xy\cdots xy}^{q}\\
v^\prime &= yx \cdots yxyx \cdots yx = qp = \overbrace{ xy\cdots xy}^{q}\overbrace{ xy\cdots xy}^{p} \label{cl:eq1}
\end{align}
Equating the suffixes of $v^\prime$ of length $|xy|$ in both side of the \Cref{cl:eq1}, we deduce $xy = yx$, i.e., $u = v$. It satisfies Case \ref{cute} of \nameref{cutcorr}. Hence a contradiction.

\paragraph*{Case 2: When the cut is strictly within some $x$ or $y$}
We will make a further case analysis: when there is an $xy$ present before the cut, and when there is an $xy$ present after the cut (since $n \geq 2$).

Suppose the cut in $u^\prime$ is in the $i^\text{th}$ $xy$ for $i > 1$, i.e., there is an $xy$ present before the cut.
\begin{enumerate}
\item When the cut is within $x$, i.e., there exists a cut $(p,q)$ of $(u^\prime, v^\prime)$ such that $p \in (xy)^+\pref{x}$ where \pref{x} is a nonempty proper prefix of $x$ and $x=\pref{x}\suff{x}$ for some word $\suff{x}$. Now,
\begin{align}
u^\prime &=\overbrace{ \cdots \pref{x}\suff{x}y\pref{x}}^{p}\overbrace{\suff{x}y \cdots}^{q}\\
v^\prime &= y\pref{x}\suff{x} \cdots y\pref{x}\suff{x} = qp = \overbrace{\suff{x}y \cdots}^{q}\overbrace{ \cdots \pref{x}\suff{x}y\pref{x}}^{p} \label{cl:eq2}
\end{align}
As before, equating the suffixes of $v^\prime$ of length $|xy|$ on both sides of \Cref{cl:eq2}, we obtain
$$yx = y\pref{x}\suff{x} = \suff{x}y\pref{x}$$
Here $\pref{x}$ and $\suff{x}$ satisfies Case \ref{cutd} of \nameref{cutcorr}. Hence a contradiction.

\item When the cut is within $y$, i.e., there exists a cut $(p,q)$ of $(u^\prime, v^\prime)$ such that $p \in (xy)^+x\pref{y}$ where \pref{y} is a nonempty prefix of $y$ and $y=\pref{y}\suff{y}$ for some word $\suff{y}$. Then,
\begin{align}
u^\prime &=\overbrace{ \cdots x\pref{y}\suff{y}x\pref{y}}^{p}\overbrace{\suff{y} \cdots}^{q}\\
v^\prime &= \pref{y}\suff{y}x \cdots \pref{y}\suff{y}x = qp = \overbrace{\suff{y} \cdots}^{q}\overbrace{ \cdots x\pref{y}\suff{y}x\pref{y}}^{p} \label{cl:eq3}
\end{align}
On both sides of the \Cref{cl:eq3}, equating the suffixes of $v^\prime$ of length $|xy|$, we get
\[yx = \pref{y}\suff{y}x = \suff{y}x\pref{y}\]
that is Case \ref{cutc} of \nameref{cutcorr}. Hence a contradiction.
\end{enumerate}
The case when there is an $xy$ after the cut is symmetric and leads to Cases \ref{cuta} and \ref{cutb} of \nameref{cutcorr}.

Since we have eliminated all of other cases, the only possible cuts of the pair $(u^\prime,v^\prime)$ are of the form $((xy)^*x,(yx)^*y)$.
\end{proof}

The below lemma gives the relationship between the cuts of two conjugate primitive pairs, particularly when there is a specific combination of these pairs which is conjugate.
\begin{lemma}[Equal Length Lemma]\label{eqlengthprimitives}
Let $(u_1,v_1),(u_2,v_2)$ be two conjugate primitive pairs of equal length (i.e., $|u_1| = |u_2|$). If pair $(u_1,v_1)^{\ell_1}(u_2,v_2)^{\ell_2}$ where $\ell_1 > 2, \ell_2 > \ell_1 + 2$ is conjugate, then there exist cuts $(x_1,y_1)$ of $(u_1,v_1)$ and $(x_2,y_2)$ of $(u_2,v_2)$ such that either $x_1 = x_2$ or $y_1 = y_2$.
\end{lemma}

\begin{proof}
Let $(u,v) = (u_1,v_1)^{\ell_1}(u_2,v_2)^{\ell_2}$ such that $\ell_1 > 2$ and $\ell_2 > \ell_1 + 2$ suffices.
\begin{align*}
u &= \overbrace{u_1 \cdots u_1}^{\ell_1 \text{ times}}\overbrace{u_2u_2 \cdots u_2u_2}^{\ell_2 \text{ times}}\\
v &= v_1 \cdots v_1v_2v_2 \cdots v_2v_2
\end{align*}

If $(u,v)$ is conjugate, then they have a cut say $(p,q)$. If either $p$ or $q$ is empty then both the pairs are identical. In this case, $u_1=x_1y_1=y_1x_1=v_1$ and $u_2=x_2y_2=y_2x_2=v_2$. By \Cref{primitive}, either $x_1=x_2 = \varepsilon$ or $y_1=y_2=\varepsilon$. Thus, assume that both $p$ and $q$ are nonempty. There are two possibilities for a cut of $(u,v)$: when the cut in $u$ is within ${u_1}^{\ell_1}u_2$ or it is after ${u_1}^{\ell_1}u_2$.

In both the cases we show that either $x_1 = x_2$, or $y_1=y_2$ or both.

\medskip
\paragraph*{Case 1: When the cut in $u$ is within ${u_1}^{\ell_1}u_2$}
In this case, the cut in $v$ is within the suffix $v_2^{\ell_1+1}$ since the $|u_1| = |u_2| = |v_2|$ and $\ell_2 > \ell_1 + 2$. Substituting $(u_1,v_1)$ and $(u_2,v_2)$ with $(x_1y_1,y_1x_1)$ and $(x_2y_2,y_2x_2)$,
\begin{align*}
u &= x_1y_1 \cdots x_1y_1\,x_2y_2 \cdots x_2y_2 = pq\\
v &= y_1x_1 \cdots y_1x_1 \cdots y_2x_2 \underbrace{y_2x_2}_{\text{cut region}} \cdots = qp
\end{align*}
Since $\ell_2 > \ell_1 + 2$, there exist at least one $y_2x_2$ before the cut in $v$. We compare the suffixes of $q$ in both $u$ and $v$. Since $q$ ends with $x_2y_2$ in $u$, the cut in $v$ should be at the end of a $y_2$ by Case \ref{distinct} of \nameref{cutcorr}. Note that by Case~\ref{same} of \nameref{cutcorr}, the cut can end after $x_2$ but then $u_2=x_2y_2=y_2x_2 = v_2$ and hence either $x_2 = \varepsilon$ or $y_2 = \varepsilon$. In both case, we can say that the cut is at the end of a $y_2$.

Hence $p$ can be of the form $x_2$ or $(x_2y_2)^+x_2$ depending upon if the cut in $v$ is within the last $y_2x_2$ or not.
\begin{align*}
u &= x_1y_1 \cdots x_1y_1x_2y_2 \cdots x_2y_2 = pq\\
v &= \underbrace{y_1x_1 \cdots y_1x_1 \cdots y_2x_2y_2}_{q} \underbrace{x_2 \cdots}_{p}
\end{align*}

Suppose $p \in (x_2y_2)^+x_2$, then equating the prefixes of $p$ in $u$ and $v$ of length $|x_2y_2| = |x_1y_1|$ (Since $|u_1| = |u_2|$), we obtain $x_2y_2 = x_1y_1$. Substituting this in $u$,
\begin{align*}
u &= x_1y_1 \cdots x_1y_1x_1y_1 \cdots x_1y_1 = pq\\
v &= \underbrace{y_1x_1 \cdots y_1x_1 \cdots y_2x_2y_2}_{q} \underbrace{x_2y_2 \cdots x_2}_{p}
\end{align*}

Now we compare the prefixes of $q$ in $u$ and $v$. Since $q$ starts with $y_1x_1$ in $v$, from  \nameref{cutcorr}, the cut in $u$ should be at the end of $x_1$. Therefore, $p \in (x_1y_1)^+x_1$ in $u$. Also, $p \in (x_2y_2)^+x_2$ in $v$. Since $|x_1y_1| = |x_2y_2|$ and $|x_1|,|x_2| < |x_1y_1|$, we can deduce $p = (x_1y_1)^{i}x_1 = (x_2y_2)^{i}x_2$ for some $i$. Hence, $x_1 = x_2$. Therefore, it implies $y_1 = y_2$ since $x_1y_1 = x_2y_2$ and hence, $(u_1,v_1)$ and $(u_2,v_2)$ are identical.

Suppose $p =x_2$ in $v$. Here, $p$ in $u$ is within the first $x_1y_1$ since $|x_1y_1| = |x_2y_2|$. Moreover $p = x_1$ since the only possible cut in $u$ will be at the end of $x_1$ by %Cases \ref{cutc}, \ref{cutd}, \ref{cute} and \ref{same} of 
\nameref{cutcorr} (comparing the prefixes of $q$ in $u$ and $v$). Hence $x_1 = x_2$.

\medskip
\paragraph*{Case 2: Cut in $u$ is after ${u_1}^{\ell_1}u_2$}
This case is symmetric. For the sake of completeness we prove it.

Substituting $(u_1,v_1)$ and $(u_2,v_2)$ with $(x_1y_1,y_1x_1)$ and $(x_2y_2,y_2x_2)$,
\begin{align*}
u &= x_1y_1 \cdots x_1y_1 \cdots x_2y_2 \overbrace{x_2y_2}^{\text{cut region}}\cdots = pq\\
v &= y_1x_1 \cdots y_1x_1y_2x_2 \cdots y_2x_2 = qp
\end{align*}

Note that there is at least one $x_2y_2$ before the cut. We compare the suffixes of $p$ in $u$ and $v$. Since $p$ ends with $y_2x_2$ in $v$, the cut in $u$ should be at the end of $x_2$ by Case~\ref{distinct} of \nameref{cutcorr}. Note that by Case~\ref{same} of \nameref{cutcorr}, the cut can end after $y_2$ but then $u_2=x_2y_2=y_2x_2 = v_2$ and hence either $x_2 = \varepsilon$ or $y_2 = \varepsilon$. Hence $q$ is of the form $y_2$ or $(y_2x_2)^+y_2$ depending upon if the cut in $u$ is within the last $x_2y_2$ or not.
\begin{align*}
u &= \overbrace{x_1y_1 \cdots x_1y_1 \cdots x_2y_2x_2}^{p}\overbrace{\cdots y_2}^{q}\\
v &= y_1x_1 \cdots y_1x_1y_2x_2 \cdots y_2x_2 = qp
\end{align*}

If $q \in (y_2x_2)^+y_2$, then comparing the prefixes of $q$ of length $|y_2x_2| = |y_1x_1|$ (Since $|v_1| = |v_2|$) in $u$ and $v$, we obtain $y_2x_2 = y_1x_1$. Substituting this in $v$,
\begin{align*}
u &= \overbrace{x_1y_1 \cdots x_1y_1 \cdots x_2y_2x_2}^{p}\overbrace{\cdots y_2}^{q}\\
v &= y_1x_1 \cdots y_1x_1y_1x_1 \cdots y_1x_1 = qp
\end{align*}

We compare the prefixes of $p$ in $u$ and $v$. Since $p$ starts with $x_1y_1$ in $u$, the cut in $v$ should be at the end of $y_1$ using  
\nameref{cutcorr}. Therefore, $q \in (y_1x_1)^+y_1$ in $v$ and $q \in (y_2x_2)^+y_2$ in $u$. Hence, as before we can deduce that $y_1 = y_2$. It also implies $x_1 = x_2$ since $y_1x_1 = y_2x_2$ and thus $(u_1,v_1)$ and $(u_2,v_2)$ are identical.

If $q = y_2$. The cut in $v$ is within first $y_1x_1$. In fact, $q=y_1$ since the only possible cut in $u$ will be at the end of $y_1$ by  
\nameref{cutcorr} (comparing the prefixes of $p$ in $u$ and $v$). Hence $y_1 = y_2$.
\end{proof}

% !TEX root = main.tex

\section{Common Witness Theorems}
\label{Sec:cwtheorems}
In this section, it is shown that an infinite set of pairs that is generated by a sumfree set is conjugate if and only if there is a word witnessing its conjugacy.
\subsection{Common Witness and its Characterisations}
Recall from \Cref{uz=zv} that a pair of words $(u,v)$ is conjugate, then there exists a word $z$ such that $uz=zv$ where $u=xy$,$v=yx$ and $z \in (xy)^*x$. By symmetry of conjugacy, there also exists a word $z'$ such that $z'u=vz'$ where $z^\prime \in (yx)^*y$. We call $z$  (resp. $z'$) in the above characterisation as an \emph{inner witness} (resp.  \emph{outer witness}) of the pair $(u,v)$ (since $z$ is appended to the inner ends).  

Given a conjugate pair $(u,v)$, the set of all inner witnesses of $(u,v)$ is $\{z \mid uz = zv\} = \cup_{\{(x, y )\mid  u = xy, v= yx \}} (xy)^\ast x$. Similarly, the set of all outer witnesses of $(u,v)$ is $\{z \mid zu = vz\} = \cup_{\{(x, y )\mid  u = xy, v= yx \}} (yx)^\ast y$. 

For example, the pair $(\mathit{aba},\mathit{baa})$ has inner witnesses $(aba)^*a$ and outer witnesses $(baa)^*ba$.

An inner witness of a pair $(u,v)$ is an outer witness of the pair $(v,u)$. We say that a pair of words has a \emph{witness} if it has either an inner witness or an outer witness.

The following proposition shows that a pair and its primitive root shares the same set of witnesses.
\begin{proposition}\label{samewitness}
Let $(u,v)$ be a conjugate pair with the primitive root $(\rho_{u}, \rho_{v})$. The following are equivalent for a word $z$.
\begin{enumerate}
\item $z$ is an inner (\textit{resp.}~outer) witness of $(u,v)$.
\item $z$ is an inner (\textit{resp.}~outer) witness of $(\rho_{u},\rho_{v})$.
\end{enumerate}
\end{proposition}

\begin{proof}
We prove $(1) \iff (2)$ for inner witness. The outer witness case is symmetric.
\begin{description}
\item[$(2) \implies (1)$:]
Assume that $z$ is an inner witness of $(\rho_u,\rho_v)$. By induction on $n$, we prove that $z$ is also an inner witness of $(\rho_u, \rho_v)^n$ for all $n \geq 1$, i.e., $\rho_u^nz = z\rho_v^n$. It is true when $n=1$. For all $n > 1$,
\begin{align*}
\rho_u^nz &= \rho_u^{n-1}\rho_uz\\
&= \rho_u^{n-1}z\rho_v && \text{(Since $\rho_uz=z\rho_v$)}\\
&= z\rho_v^{n-1}\rho_v && \text{(Inductive Hypothesis)}\\
&= z\rho_v^n
\end{align*}
Since $(u,v)$ is a conjugate pair, there exists an integer $n \geq 1$ such that $(u,v) = (\rho_u, \rho_v)^n$ by \Cref{primitiveconjugates}. Hence, $z$ is also an inner witness for $(u,v)$.
\item[$(1) \implies (2)$:] Since $(u,v)$ is conjugate, its primitive root $(\rho_u,\rho_v)$ is conjugate as well by \Cref{primitiveconjugates}.
Consider the case when $u \neq v$. It follows that $\rho_u \neq \rho_v$. According to \Cref{primitive}, the pair $(\rho_u,\rho_v)$ has a unique cut, denoted as $(x,y)$. From \Cref{cutlemmaprop}, all cuts of $(u,v)$ are of the form $((xy)^*x,(yx)^*y)$. From \Cref{uz=zv}, an inner witness of $(u,v)$ belongs to \[((xy)^*x(yx)^*y)^*(xy)^*x = (xy)^*x\] and hence is an inner witness of $(\rho_u,\rho_v)$. 

In the case where $u=v$, it follows that $\rho_u = \rho_v$. Consequently, any witness $z$ for the pair $(u,u)$ belongs to the set $u^*$ that is a subset of $\rho_u^*$. Thus, $z$ is also a witness for the pair $(\rho_u,\rho_v)$, since $\rho_u^*$ consists of witnesses of $(\rho_u,\rho_v)$.
\end{description}
\end{proof}

\begin{proposition}\label{prop:iwredux}
Let $\rho, \rho_1, \rho_2$ be nonempty primitive words that are pair-wise conjugate, and let $\as$ be an arbitrary word. If $\rho_1\as\rho_2$ is an infix  of a word in $\rho^*$, then $\as$ is an inner witness of $(\rho_1,\rho_2)$, i.e. $\rho_1\as = \as\rho_2$.
\end{proposition} 

\begin{proof}

Since $(\rho_1,\rho_2)$ are conjugate, there exists a cut $(x,y)$ such that 
$\rho_1 = xy$ and $\rho_2 = yx$. Moreover, since $\rho \sim \rho_1 \sim \rho_2$, all three words have the same length, say $p > 0$. Further, since $\rho \sim \rho_1 = xy$, we have $\rho^+ = s \rho_1^*t =  s(xy)^* t$ for some words $s,t$ such that $|st| = p$.

Assume that $\rho_1\alpha\rho_2$ is an infix of a word $w \in \rho^+$. By the previous observation $w=s(xy)^it$ for some sufficiently large $i>0$. Firstly,
assume $\rho_1 \neq \rho_2$. Then, the cut $(x,y)$ is non-empty and unique by \Cref{primitive}. Since $\rho_1$ and $\rho_2$ are primitive, by \Cref{same} of \nameref{cutcorr}, neither $xy$ nor $yx$ is equal to any of its non-trivial cyclic shifts. Given that $\rho_1 \as \rho_2 = xy \as yx$ occurs as an infix of $w$. We do a case analysis. If $xy \as yx$  lies entirely within the factor $(xy)^i$ of $w$, then by \nameref{cutcorr}, the factor $xy$ of $xy \as yx$  must align exactly with a factor $xy$ of $w$, and likewise $yx$ should match exactly with a factor $yx$ of $w$. As a consequence, $\as \in (xy)^*x$ which is an inner witness of $(\rho_1,\rho_2)$. 

Next, assume for the sake of contradiction that $xy \as yx$ is a factor of $s(xy)^i$, but not a factor of $(xy)^i$. Hence a suffix of $s$ is a prefix of $xy$ (recall that $|s| \le p = |xy|$). Since $\rho \sim xy$ and $|s| \le p$, $s$ is an infix of $xy$. Hence, the prefix $xy$ of $\rho_1\as\rho_2$ matches with a nontrivial cyclic shift of $xy$ in $w$. This contradicts \nameref{cutcorr}. A symmetric argument rules out the case when $xy \as yx$ is a factor of $(xy)^it$, but not a factor of $(xy)^i$. 

When $\rho_1 = \rho_2$, the only cuts possible are empty cuts (i.e., either $x$ or $y$ is empty) by \Cref{primitive}. Thus $\rho_1 = \rho_2 = x$, and the previous argument applies directly, and we get $\as \in x^*$, and hence is an inner witness of $(\rho_1,\rho_2)$.     
\end{proof}

We generalise the notion of a witness of a pair to a set of pairs.

\begin{definition}[Common Witness]\label{def:commonwitness}
A word is a \emph{common inner witness} of a set of pairs $\pair$ if it is an inner witness of each pair in $\pair$. Similarly, a word is a \emph{common outer witness} of $\pair$ if it is an outer witness of each pair in $\pair$.

A set of pairs has a \emph{common witness} if it has either a common inner witness or a common outer witness.
We denote by $W(P)$ the set of all common witnesses of $P$.
\end{definition}

The structure of a common witness of a set of pairs is obtained from \Cref{uz=zv}.

\begin{proposition}\label{cwpattern}
Let $\pair$ be a set of pairs of words. The following are equivalent.
\begin{enumerate}
\item $z$ is a common inner witness of $\pair$.
\item There exists a cut $(x,y)$ of each pair $(u,v) \in P$ such that $z \in \bigcap_{(u,v) \in P}(xy)^*x$. 
\item $z \in \bigcap_{(u,v) \in P} \bigcup_{\{(x, y)\mid  u = xy, v= yx \}}(xy)^*x$
\end{enumerate}
The statement for common outer witness is analogous.
\end{proposition}

\begin{proof}
Follows from the definition of common witness and \Cref{uz=zv}. 
\end{proof}

\begin{example}
Consider the set $P = \{(\mathit{ab},\mathit{ba}),(\mathit{abab},\mathit{baba})\}$. The pair $(ab,ba)$ has a unique cut $(a,b)$, and the pair $(\mathit{abab},\mathit{baba})$ has two cuts: $(a,bab)$ and $(aba,b)$. The word $a$ is a common inner witness of $P$ since $a$ belongs to both $(ab)^*a$ and $(abab)^*a$ (using the first cut). Similarly, $aba$ is also a common inner witness of $P$ since $aba$ belongs to both $(ab)^*a$ and $(abab)^*aba$ (using the second cut). Notice that $aba$ is not in the intersection of $(ab)^*a$ and $(abab)^*a$.
\end{example}
\Cref{samewitness} connecting witness of a conjugate pair and its root can be lifted to a set of conjugate pairs and its root as follows. 
\begin{proposition}\label{samewitnessext} The common witnesses of a set of conjugate  pairs $G$ and its root $\roots{G}$ are the same, i.e., a word $z$ is a common inner (\textit{resp.}~outer) witness of $G$ iff $z$ is a common inner (\textit{resp.}~outer) witness of $\roots{G}$. 
\end{proposition}
\begin{proof}
We prove for common inner witness. The common outer witness case is symmetric.
\begin{align*}
z &\text{ is a common inner witness of } G\\
&\iff z \text{ is an inner witness of } (u,v) \text{ for each $(u,v) \in G$}\\
&\iff z \text{ is an inner witness of } (\rho_u,\rho_v) \text{ for each $(u,v) \in G$}  && \text{ (By \Cref{samewitness})}\\
&\iff z \text{ is a common inner witness of } \roots{G}
\end{align*}
\end{proof}

When a set is not conjugate, clearly it has no common witness. However, even when a set is conjugate, it may have both common inner and outer witnesses, or only common inner witness, or only common outer witness, or neither of them as shown below.

\begin{example}
Consider the set $P = \{(\mathit{ab},\mathit{ba}),(\mathit{ac},\mathit{ca}))\}$. The pair $(\mathit{ab},\mathit{ba})$ has inner witnesses $(ab)^*a$ and outer witnesses $(ba)^*b$. Similarly, the pair $(\mathit{ac},\mathit{ca})$ has inner witnesses $(ac)^*a$ and outer witnesses $(ca)^*c$. According to \Cref{cwpattern}, the set $P$ has a unique common inner witness $a = (ab)^*a \cap (ac)^*a$, but it does not have any common outer witness since $(ba)^*b \cap (ca)^*c = \emptyset$.

The set $\{(\mathit{ab},\mathit{ba}),(\mathit{abab},\mathit{baba})\}$ has both common inner witnesses $$(ab)^*a = (ab)^*a \cap ({(abab)^*aba \cup (abab)^*a})$$ and common outer witnesses $$(ba)^*b = {(ba)^*b \cap ({(baba)^*b \cup (baba)^*bab})} \ .$$

However, the set $\{(\mathit{ab},\mathit{ba}),(\mathit{ba},\mathit{ab})\}$ has no common witnesses since $(ab)^*a \cap (ba)^*b = \emptyset$.
\end{example}

Next we analyse the number of common witnesses a set of pairs can have.
\begin{proposition}\label{infmanycwcorr}
Let $G$ be a set of conjugate pairs of words. The following are equivalent.
\begin{enumerate}
\item $G$ has more than one common witness.
\item $G$ has infinitely many common witnesses.
\item $G$ has infinitely many common inner witnesses.
\item $G$ has infinitely many common outer witnesses.
\item All the pairs in $G$ have the same primitive root.
\end{enumerate}
\end{proposition}

\begin{proof}
We prove $(1) \implies (2) \implies (1)$ and $(1) \implies (5) \implies (3),(4) \implies (1)$.

 $(2) \implies (1)$ and $(3),(4) \implies (1)$ is obvious. We first show $(5) \implies (3),(4)$. Since all the pairs in $G$ have the same primitive root, $\roots{G}$ is singleton. Thus, $\roots{G}$ has infinitely many common inner as well as common outer witnesses by \Cref{uz=zv}. Since $G$ and $\roots{G}$ have the same witnesses (\Cref{samewitnessext}), $G$ also has infinitely many common inner and common outer witnesses.

Now it remains to show that $(1) \implies (2)$ and $(1) \implies (5)$. If the set $G$ has two common witnesses, namely $z_1$ and $z_2$, then according to \Cref{samewitnessext}, $z_1$ and $z_2$ are also common witnesses of $\roots{G}$.  Let $\roots{G} = \{(u_i,v_i) \mid i \in I\}$ where $I$ is an Index set.

Suppose $z_1$ and $z_2$ are two common \emph{inner} witnesses of $\roots{G}$. We choose cuts $(x_i,y_i)$ for each $(u_i,v_i) \in \roots{G}$ such that both $z_1$ and $z_2$ belongs to $\bigcap_{i \in I}(x_iy_i)^*x_i$. Since $\roots{G}$ is a set of primitive pairs, $(x_i,y_i)$ is the unique cut of $(u_i,v_i)$ when $u_i \neq v_i$ by \Cref{primitive}. When $u_i=v_i$, it has only two empty cuts, namely $(\epsilon, u_i)$ and $(u_i,\epsilon)$ (if $u_i = v_i$). The inner witnesses obtained using cut $(\epsilon, u_i)$ is a superset of inner witnesses obtained using $(u_i,\epsilon)$. Hence, we can choose cut $ (x_i,y_i) = (\epsilon, u_i)$ for pair $(u_i,v_i)$ when $u_i = v_i$. Therefore, as stated in \Cref{cwpattern}, both $z_1$ and $z_2$ belongs to $\bigcap_{i \in I}(x_iy_i)^*x_i$. 

Without loss of generality, assume that $|z_1| < |z_2|$. As depicted in \Cref{infmany}, a common factor $w \in \bigcap_{i \in I}(y_ix_i)^{\geq 1}$ exists for each $u_i^\omega$ that can be repeated one after another in $u_i^\omega$ to get longer and longer common inner witnesses.
\begin{figure}[t]
\centering

\begin{tikzpicture}

\draw [thick,decorate,
decoration = {calligraphic brace,
raise=5pt,
amplitude= 5pt,
}] (0,4.5) -- (2,4.5);
\draw[blue] (1,5) node{$z_1$};
\draw [thick,decorate,
decoration = {calligraphic brace,
raise=5pt,
amplitude= 5pt,
}] (2,4.5) -- (5,4.5);
\draw (3.5,5) node{$w$};
\draw [thick,decorate,
decoration = {calligraphic brace,
raise=5pt,
amplitude= 5pt,
}] (0,5) -- (5,5);
\draw[magenta] (2.5,5.5) node{$z_2$};
\draw (-0.5,3) node {$u_2^\omega$};
\draw[gray, line width=3pt, line cap=round, dash pattern=on 0pt off 2\pgflinewidth][-](0,3) -- (2,3);
\draw [fill = cyan, opacity = 0.2, draw = none] (-0.10,2.90) rectangle (2.05,3.10);
\draw [fill = magenta!70, opacity = 0.1, draw = none] (-0.15,2.85) rectangle (5,3.15);
\draw[blue] (1,3.3) node{$\scriptstyle (x_2y_2)^{m_2}x_2$};
\draw [gray, line width=3pt, line cap=round, dash pattern=on 0pt off 2\pgflinewidth][-](2.15,3) -- (5,3);
\draw (3.5,3.3) node{$\scriptstyle (y_2x_2)^{\geq 1}$};
\draw [gray, line width=3pt, line cap=round, dash pattern=on 0pt off 2\pgflinewidth][-](5.1,3) -- (6,3);
\draw [magenta] (2.5,2.7) node{$\scriptstyle (x_2y_2)^{n_2}x_2$};

\draw (-0.5,4) node {$u_1^\omega$};
\draw[gray, line width=3pt, line cap=round, dash pattern=on 0pt off 2\pgflinewidth][-](0,4) -- (2,4);
\draw [fill = cyan, opacity = 0.2, draw = none] (-0.10,3.90) rectangle (2.05,4.10);
\draw [fill = magenta!70, opacity = 0.1, draw = none] (-0.15,3.85) rectangle (5,4.15);
\draw[blue] (1,4.3) node{$\scriptstyle (x_1y_1)^{m_1}x_1$};
\draw [gray, line width=3pt, line cap=round, dash pattern=on 0pt off 2\pgflinewidth][-](2.15,4) -- (5,4);
\draw (3.5,4.3) node{$\scriptstyle (y_1x_1)^{\geq 1}$};
\draw [gray, line width=3pt, line cap=round, dash pattern=on 0pt off 2\pgflinewidth][-](5.1,4) -- (6,4)[anchor=south, xshift=-1.5cm, yshift=.4cm];
\draw [magenta] (2.5,3.7) node{$\scriptstyle (x_1y_1)^{n_1}x_1$};

\draw (3,2.25) node{$\vdots$};

\end{tikzpicture}

\caption{When there are at least two common inner witnesses $z_1, z_2$.}
\label{infmany}
\end{figure}
By symmetry, when $z_1$ and $z_2$ are both common \emph{outer} witnesses of  $\roots{G}$, we get infinitely many common outer witnesses.

Now, WLOG assume that $z_1$ is a common \emph{inner} witness and $z_2$ is a common \emph{outer} witness of $\roots{G}$, where $z_1 \neq z_2$. Here also we choose cuts $(x_i,y_i)$ for each $(u_i,v_i) \in \roots{G}$ such that $z_1$ belongs to $\bigcap_{i \in I}(x_iy_i)^*x_i$ and $z_2$ belongs to $\bigcap_{i \in I}(y_ix_i)^*y_i$.
For each distinct primitive pair in $\roots{G}$, there exists a unique cut. However, for identical primitive pairs, we fix a cut based on the values of $z_1$ and $z_2$. We consider two cases: either $z_1, z_2 \neq \epsilon$, or exactly one of $z_1$ and $z_2$ is equal to $\epsilon$.
In the case where $z_1, z_2 \neq \epsilon$, for primitive pairs $(u_i,v_i)$ such that $u_i = v_i$, we can choose either of the two empty cuts as $(x_i,y_i)$, resulting in $z_1 \in (x_iy_i)^*x_i$ and $z_2 \in (y_ix_i)^*y_i$.
In the second case, if $z_1 = \epsilon$, we select the cut $(x_i,y_i) = (\epsilon, u_i)$. This choice ensures that $z_1 \in (x_iy_i)^*x_i$ and $z_2 \in (y_ix_i)^*y_i$ (since $z_2 \neq \epsilon$). Similarly, if $z_2 = \epsilon$, we choose the cut $(x_i,y_i) = (u_i, \epsilon)$.
Consequently, we can conclude that $z_1$ belongs to $\bigcap_{i \in I}(x_iy_i)^*x_i$ and $z_2$ belongs to $\bigcap_{i \in I}(y_ix_i)^*y_i$.

As shown in \Cref{infmany1}, concatenating $z_1 \cdot z_2 \cdot z_1$ in $u_i^\omega$, we get one more common inner witness $z_3$ for $\roots{G}$. By the above argument, $\roots{G}$ has infinitely many common witnesses. Since witnesses of $\roots{G}$ are also witnesses of $G$ (\Cref{samewitnessext}), it implies that $G$ itself has infinitely many common witnesses. This completes the proof of $(1) \implies (2)$.

\begin{figure}[htbp]
\centering

\begin{tikzpicture}
\draw [thick,decorate,
decoration = {calligraphic brace,
raise=5pt,
amplitude= 5pt,
}] (0,4.5) -- (2,4.5);
\draw[blue] (1,5) node{$z_1$};
\draw [thick,decorate,
decoration = {calligraphic brace,
raise=5pt,
amplitude= 5pt,
}] (2,4.5) -- (5,4.5);
\draw[magenta] (3.5,5) node{$z_2$};
\draw [thick,decorate,
decoration = {calligraphic brace,
raise=5pt,
amplitude= 5pt,
}] (5,4.5) -- (7,4.5);
\draw[blue] (6,5) node{$z_1$};

\draw (-0.5,4) node {$u_1^\omega$};
\draw[gray, line width=3pt, line cap=round, dash pattern=on 0pt off 2\pgflinewidth][-](0,4) -- (2,4);
\draw [fill = blue, opacity = 0.2, draw = none] (-0.10,3.90) rectangle (2,4.10);
\draw [fill = magenta!70, opacity = 0.1, draw = none] (2,3.90) rectangle (5,4.10);
\draw [fill = blue, opacity = 0.2, draw = none] (5,3.90) rectangle (7.07,4.10);
\draw[blue] (1,4.3) node{$\scriptstyle (x_1y_1)^{m_1}x_1$};
\draw [gray, line width=3pt, line cap=round, dash pattern=on 0pt off 2\pgflinewidth][-](2.15,4) -- (5,4);
\draw[magenta] (3.5,4.3) node{$\scriptstyle (y_1x_1)^{n_1}y_1$};
\draw [gray, line width=3pt, line cap=round, dash pattern=on 0pt off 2\pgflinewidth][-](5.1,4) -- (8,4);
\draw [blue] (6,4.3) node{$\scriptstyle (x_1y_1)^{m_1}x_1$};

\draw (3,2.25) node{$\vdots$};

\draw (-0.5,3) node {$u_2^\omega$};
\draw[gray, line width=3pt, line cap=round, dash pattern=on 0pt off 2\pgflinewidth][-](0,3) -- (2,3);
\draw [fill = blue, opacity = 0.2, draw = none] (-0.10,2.90) rectangle (2,3.10);
\draw [fill = magenta!70, opacity = 0.1, draw = none] (2,2.90) rectangle (5,3.10);
\draw [fill = blue, opacity = 0.2, draw = none] (5,2.90) rectangle (7.07,3.10);
\draw[blue] (1,3.3) node{$\scriptstyle (x_2y_2)^{m_2}x_2$};
\draw [gray, line width=3pt, line cap=round, dash pattern=on 0pt off 2\pgflinewidth][-](2.15,3) -- (5,3);
\draw[magenta] (3.5,3.3) node{$\scriptstyle (y_2x_2)^{n_2}y_2$};
\draw [gray, line width=3pt, line cap=round, dash pattern=on 0pt off 2\pgflinewidth][-](5.1,3) -- (8,3);
\draw [blue] (6,3.3) node{$\scriptstyle (x_2y_2)^{m_2}x_2$};

\draw (3,2.25) node{$\vdots$};

\end{tikzpicture}

\caption{When there are 1 common inner witness $z_1$ and 1 common outer witness $z_2$.}
\label{infmany1}
\end{figure}

In both the cases, we get that $(x_1y_1)^\omega = (x_2y_2)^\omega = \cdots$ and $(y_1x_1)^\omega = (y_2x_2)^\omega = \cdots$. Hence from Fine and Wilf \Cref{finewilf}, all $u_i$'s has the same primitive root. Similarly, all $v_i$'s has the same primitive root. This proves $(1) \implies (5)$.
\end{proof}
 
Therefore, a set of pairs can have no common witness, a unique common witness, or infinitely many common witnesses.
\begin{corollary}\label{1-infinity}
    For all set of pairs $G$, if $W(G) \neq \emptyset$, then one of the following is true.
    \begin{enumerate}
        \item $|W(G)|=\infty$ and $\roots{G}$ is singleton.
        \item $|W(G)|=1$ and $\roots{G}$ is not singleton.
    \end{enumerate}
\end{corollary}

\subsection{Common Witness Theorem for Kleene Closure}
Lyndon-Sch\"utzenberger theorem characterises conjugacy of a pair of words. We generalise the notion in \Cref{uz=zv} to an infinite set of pairs closed under concatenation. The question we ask is:

``Given an arbitrary set of pairs $G$, is $G^*$ conjugate?''

If $\spans{G}$ has a common witness, then each pair in $\spans{G}$ has a witness and is conjugate. Hence $\spans{G}$ is conjugate. We prove the converse, namely, if $\spans{G}$ is conjugate, then it has a common witness. 
The below theorem characterises conjugacy of a freely generated set of pairs of words.

\begin{theorem}[Common Witness Theorem for Kleene Closure]\label{conjugacy}
Let $G$ be an arbitrary conjugate set of pairs of words. The following are equivalent.
\begin{enumerate}
\item $\spans{G}$ is conjugate.
\item $\spans{G}$ has a common witness $z$.
\item $G$ has a common witness $z$.
\item $\roots{G}$ has a common witness $z$.
\end{enumerate}
\end{theorem}

\begin{proof}
We prove $(4) \implies (3) \implies (2) \implies (1) \implies (4)$. 
The only nontrivial direction is $(1) \implies (4)$ that is proved in \Cref{sec:proofcw}.
\begin{description}
\item[$(4) \implies (3)$]
Follows from \Cref{samewitnessext}.
\item[$(3) \implies (2)$]
Suppose there exists a common inner witness $z$ of the set $G$. Hence $\forall (u,v) \in G$, $u z=z v$. Let $(u',v')$ be any arbitrary element from $\spans{G}$, i.e., $(u',v') = (u_1\cdots u_n,v_1\cdots v_n)$ for some $n \geq 1$ and $(u_i,v_i) \in G$ for $1 \leq i \leq n$. By induction on $n$, we equate $u'z= zv'$ as follows. Thus, $z$ is a common inner witness of $G^*$. The proof for common outer witness is symmetric. 
\begin{align*}
u'z &= u_1\cdots u_{n-1}u_n z \\
&=  u_1\cdots u_{n-1} z v_n && \text{(Since $u_n z = z v_n$)}\\
&= zv_1\cdots v_{n-1}v_n  && \text{(Inductive Hypothesis)}\\
&= zv'
\end{align*}
\item[$(2) \implies (1)$] Follows from \Cref{uz=zv}.
\end{description}
\end{proof}

\begin{remark}
Note that whenever $G^\ast$ is conjugate even though $G$ and $\roots{G}$ have the same common witnesses, $G^\ast$ need not be equal to $\roots{G}^\ast$. Indeed, $  G^\ast \subseteq \roots{G}^\ast $.
\end{remark}

As a corollary, we get that for any rational expression $E$, the expression $E^*$ is conjugate if, and only if, $E$ has a common witness. Below is an illustration of the common witness theorem for a set of pairs that is not rational.

\begin{example}
Let $G= \{(ab^p,b^pa) \mid p \text{ is a prime number}\}$. The word $a \in \bigcap_{p \in \mathbb{N} \text{, $p$ is a prime}} (ab^p)^*a$ is a common inner witness of the set $G$. It is also easy to verify that $G^*$ is conjugate and $a$ is a common inner witness of $G^*$.
\end{example}

\subsection{Conjugacy and Common Witness Theorem for Monoid Closure}
Next we state the conjugacy theorem for monoid closures, i.e., sumfree sets of the form $$M = (\as_0,\bs_0){G_1}^*(\as_1,\bs_1){G_2}^*\cdots (\as_{k-1},\bs_{k-1}){G_k}^*(\as_k,\bs_k), k > 0 \ .$$ where $G_1^*,G_2^*, \ldots, G_k^*$ are arbitrary nonempty sets of conjugate pairs.

\begin{theorem}\footnote{We thank the participants of Autoboz 2025 (Grenaa, Denmark) for identifying a missing case in the conjugacy theorem for monoid closure (Item 2 of Theorem 7).}[Conjugacy Theorem for Monoid Closure]\label{thm:conjugateMonoidClosure}
    Let $M = (\alpha_0,\beta_0)(\prod_{i=1}^{k}G_i^*(\as_i,\bs_i))$, $k \geq 1$, be a sumfree set. $M$ is conjugate if, and only if, one of the following holds.
\begin{enumerate}
    \item  $M$ has a common witness.
    \item For each $i \in [1,k]$, $G_i^*$ is conjugate and there exist a primitive word $\rho$ such that $\fst(G_i^*)\sim \snd(G_i^*) \subsetsim \rho^*$,  and there exists a word $v \sim \as_0 \cdots \as_k \sim \bs_0 \cdots \bs_k$ such that $\fst(M) \sim \snd(M) \subsetsim v\rho^*$.
\end{enumerate}
    \end{theorem}
     
We give the proof of the above theorem in \Cref{sec:conjugateMonoid}. Note that the above result does not generalise to arbitrary sets of pairs, in particular, rational sets using sum.
\begin{example}
$(ab,ba)^* + (baa,aab)^*$ is an infinite conjugate set with \emph{no} common witness. Also, it does not satisfy Item $(2)$ of \Cref{thm:conjugateMonoidClosure}.
\end{example}

\begin{definition}[Redux, Singleton Redux]
Let $M$ be the sumfree set $$(\as_0,\bs_0){G_1}^*(\as_1,\bs_1){G_2}^*\cdots (\as_{k-1},\bs_{k-1}){G_k}^*(\as_k,\bs_k) \ .$$ The \emph{redux} of $M$ is the pair $(\as_0\as_1\cdots \as_k,\bs_0\bs_1\cdots \bs_k)$ obtained by substituting each $G_i^*$ by the empty pair $(\epsilon, \epsilon)$.
A \emph{singleton redux} of $M$ is a set obtained by substituting all but one of the $G_i^*$'s by the empty pair $(\epsilon, \epsilon)$. They are of the form $(\as_0 \cdots \as_{i-1},\bs_0 \cdots \bs_{i-1}){G_i}^*(\as_i \cdots \as_k,\bs_i \cdots \bs_k)$ where $1 \leq i \leq k$.
\end{definition}

\begin{example}
Consider the set $M = (a,a)(\mathit{baa},\mathit{aba})^*(b,a)(\mathit{aab},\mathit{baa})^*(a,b)$. The redux of $M$ is $(\mathit{aba},\mathit{aab})$, and its singleton reduxes are $(a,a)(\mathit{baa},\mathit{aba})^*(\mathit{ba},\mathit{ab})$ and $(\mathit{ab},\mathit{aa})(\mathit{aab},\mathit{baa})^*(a,b)$.
\end{example}

If a sumfree set has a common witness, it is conjugate. The converse need not be true as illustrated in the below example.

\begin{example}
    Consider the sumfree set \[M = \left( \begin{array}{c}\varepsilon\\
c\end{array}\right) 
\left( \begin{array}{c}ab\\
ba\end{array}\right)^*
\left( \begin{array}{c}
c\\
b\end{array}\right)
\left( \begin{array}{c}ba\\
ab\end{array}\right)^*
\left( \begin{array}{c}b\\
\varepsilon\end{array}\right) \]
The set $M = \{((ab)^nc(ba)^mb,c(ba)^nb(ab)^m) \mid n,m  \geq 0\}$ is conjugate and has no common witness. However, $G_1^* = (ab,ba)^*$ and $G_2^* = (ba,ab)^*$ is conjugate and there exist a primitive word $\rho = ab$ such that $\fst(G_i^*)\sim \snd(G_i^*) \subsetsim \rho^*=(ab)^*$,  and there exists a word $v = cb$ (which is a cyclic shift of the redux of $M$) such that $\fst(M) \sim \snd(M) \subsetsim v\rho^* = (cb)(ab)^*$.

Note that the $v$ and $\rho$ need not be unique. For $\rho= ba$ and $v=bc$, it is also true that $\fst(M) \sim \snd(M) \subsetsim v\rho^* = (bc)(ba)^*$.
\end{example}

For sumfree sets of the form $M = (\as_0,\bs_0)G^*(\as_1,\bs_1)$ with only one Kleene closure, we show that $M$ is conjugate iff $M$ has a common witness. Later, we show that if a general sumfree set has a common witness then that is in the intersection of the common witnesses of the singleton reduxes of the set. Towards this, we need the following definition.

\begin{definition}[Prefix Delay, Suffix Delay]\label{mutualquotient}
If $u$ and $v$ are words such that one of them is a prefix of another, we define the \emph{prefix delay} between $u$ and $v$ as
$$[u,v]_L = \begin{cases}u^{-1}v & \text{ if $u$ is a prefix of $v$}\\
v^{-1}u & \text{ if $v$ is a prefix of $u$}
\end{cases}$$
Similarly, the \emph{suffix delay} of two words $u$ and $v$ such that one of them is suffix of another, denoted by $[u,v]_R$, is $vu^{-1}$ if $u$ is a suffix of $v$ and $uv^{-1}$ if $v$ is a suffix of $u$.

For example, $[abaa,ab]_L = aa = [ab,abaa]_L$.
\end{definition}

Following is the common witness theorem for a sumfree set with only one Kleene star, i.e., sets of the form $(\as_0,\bs_0)G^*(\as_1,\bs_1)$. In short it states that such a set is conjugate if and only if it has a common witness that is determined by the common witnesses of $G \cup \{(\as_1\as_0,\bs_1\bs_0)\}$.
\begin{theorem}\label{singlemonomial}
Let $M= (\as_0,\bs_0)G^*(\as_1,\bs_1)$ be a sumfree set with nonempty redux. The following are equivalent.
\begin{enumerate}
\item $M$ is conjugate.
\item There exist a common witness  of $G \cup \{(\as_1\as_0,\bs_1\bs_0)\}$.
\item $M$ has a common witness. 
\end{enumerate}

Furthermore,  
\begin{enumerate}[label=(\alph*)]
\item\label{Thm7:1} If the set $G \cup \{(\as_1\as_0,\bs_1\bs_0)\}$ has a unique common inner witness, say $z'$,  then $M$ has a unique common witness $z = [\as_0z',{\bs_0}]_R = [\as_1,z'\bs_1]_L$. Moreover, if $|\as_0z'| \geq |\bs_0|$ or equivalently $|\as_1| \leq |z'\bs_1|$, then $z$ is a common inner witness, otherwise it is a common outer witness.

\item\label{Thm7:2} If   the set $G \cup \{(\as_1\as_0,\bs_1\bs_0)\}$ has a  unique common outer witness, say $z'$,  then $M$ has a unique common witness $z = [\as_0,\bs_0z']_R = [z'\as_1,\bs_1]_L$. Moreover, if $ |z'\as_1| \geq |\bs_1|$ or equivalently $|\as_0| \leq |\bs_0z'|$, then $z$ is a common outer witness, otherwise it is a common inner witness.

\item\label{Thm7:3} If $G \cup \{(\as_1\as_0,\bs_1\bs_0)\}$ has infinitely many common witnesses, then $M$ is a subset of powers of the primitive root of its redux. Thus, $M$ has infinitely many common witnesses.
\end{enumerate}

\end{theorem}

\begin{example}
Let $M = (\as_0,\bs_0)G^*(\as_1,\bs_1)$ be a sumfree set with one Kleene star where $$\(\begin{matrix}\as_0 \\\bs_0 \end{matrix}\) = \(\begin{matrix}ab \\b \end{matrix}\), G =\left\{ \(\begin{matrix}bab \\ abb \end{matrix}\)\right\}, \(\begin{matrix}\as_1 \\\bs_1 \end{matrix}\) = \(\begin{matrix}b \\ ab \end{matrix}\). $$ The redux of $M$ is $(\as_0\as_1,\bs_0\bs_1) = (abb,bab)$. The set $G \cup \{(\as_1\as_0,\bs_1\bs_0)\} = \{(bab,abb)\} \cup \{(bab,abb)\} =  \{(bab,abb)\}$ and, hence it has infinitely many common witnesses. By \Cref{singlemonomial} \ref{Thm7:3}, $M$ is a subset of powers of the primitive root of the redux, i.e., $M = (abb,bab)^{+}$. Therefore, $M$ has infinitely many witnesses same as those of $(abb,bab)$. 
\end{example}
A singleton redux of a sumfree set is nothing but a sumfree set with only one Kleene star. Given any sumfree set $M$, if $M$ is conjugate, each of its singleton reduxes are conjugate. From \Cref{singlemonomial}, a singleton redux of $M$ has a common witness. Further, we prove that $M$ has a common witness $z$ iff $z$ is the common witness of each of its singleton reduxes. The below theorem characterises the common witness of a general sumfree set.
\begin{theorem}[Common Witness Theorem for Monoid Closure]\label{generalmonomialwitness}
Let $M$ be a sumfree set. The following are equivalent.
\begin{enumerate}
\item $z$ is a common witness of each singleton redux of $M$.
\item $z$ is a common witness of $M$.
\end{enumerate}
\end{theorem}

\begin{example}
Let $M = (\as_0,\bs_0)G_1^*(\as_1,\bs_1)G_2^*(\as_2,\bs_2)$ be a sumfree set with two Kleene star where 
$$\(\begin{matrix}\as_0 \\\bs_0 \end{matrix}\) = \(\begin{matrix}b \\a \end{matrix}\), G_1 =\left\{ \(\begin{matrix}ac \\ ca \end{matrix}\)\right\}, \(\begin{matrix}\as_1 \\\bs_1 \end{matrix}\) = \(\begin{matrix}ab \\ b \end{matrix}\), G_2 =\left\{ \(\begin{matrix}bab \\ bab \end{matrix}\)\right\}, \(\begin{matrix}\as_2 \\\bs_2 \end{matrix}\) = \(\begin{matrix}\epsilon \\ b \end{matrix}\). $$
The redux of $M$ is $(\as_0\as_1\as_2,\bs_0\bs_1\bs_2) = (bab,abb)$. The set $M$ has two singleton reduxes, $$M_1 = \(\begin{matrix}\as_0 \\ \bs_0 \end{matrix}\) G_1^*\(\begin{matrix}\as_1\as_2 \\\bs_1\bs_2 \end{matrix}\) = \(\begin{matrix}b \\a \end{matrix}\){\(\begin{matrix}ac \\ ca \end{matrix}\)}^*\(\begin{matrix}ab \\bb \end{matrix}\)$$

and,
$$M_2 = \(\begin{matrix}\as_0\as_1 \\ \bs_0\bs_1 \end{matrix}\) G_2^*\(\begin{matrix}\as_2 \\ \bs_2 \end{matrix}\) = \(\begin{matrix}bab \\ab \end{matrix}\){\(\begin{matrix}bab \\ bab \end{matrix}\)}^*\(\begin{matrix}\epsilon \\b \end{matrix}\) .$$

The set $G_1 \cup \{(\as_1\as_2\as_0,\bs_1\bs_2\bs_0)\} = \{(ac,ca)\} \cup \{(abb,bba)\} =  \{(ac,ca), (abb,bba)\}$ has a unique common inner witness, say $z_1 = a = (ac)^*a \cap (abb)^*a$ and no common outer witness since $(ca)^*c \cap (bba)^*bb = \emptyset$. By \Cref{singlemonomial} \ref{Thm7:1}, the unique common inner witness of the singleton redux $M_1$ of $M$ is $[\as_0z_1,\bs_0]_R = [ba,a]_R = b$. 

The set $G_2 \cup \{(\as_2\as_0\as_1,\bs_2\bs_0\bs_1)\} =  \{(bab,bab)\}$ has infinitely many common witnesses. Thus the singleton redux $M_2$ is a subset of powers of the primitive root of the redux using \Cref{singlemonomial} \ref{Thm7:3}, i.e., $M_2 = (bab,abb)^+$.  Thus $M_2$ have infinitely many common inner witnesses $(bab)^*b$ and common outer witnesses $(abb)^*ab$.

By \Cref{generalmonomialwitness}, $M$ has a unique common inner witness $b \cap (bab)^*b = b$, that equals to the intersection of the common inner witness of its singleton reduxes $M_1$ and $M_2$. 
\end{example}

The proof of the common witness theorems of singleton redux as well as the monoid closure are given in \Cref{sec:conjugacy_monomial}. Finally, using those results we prove the conjugacy theorem for monoid closure (\Cref{thm:conjugateMonoidClosure}) in \Cref{sec:conjugateMonoid}

% !TEX root = main.tex

\section{Existence of Common Witness for Kleene Closure}\label{sec:proofcw}

In this section, we prove the direction $(1) \implies (4)$ of \Cref{conjugacy} recalled in the following lemma.

\begin{lemma}\label{1to4cw}
For a set of pairs $G$, if $\spans{G}$ is conjugate then $\roots{G}$ has a common witness.
\end{lemma}
We prove the lemma when $G$ is finite by case analysis and then extend it for a countably infinite set of pairs using a compactness argument. Towards proving this, we need the following auxiliary result proven using \nameref{eqlengthprimitives}. A corollary of the proposition below is that if $G^*$ is conjugate and $\roots{G}$ only consists of pairs of identical length, then $\roots{G}$ has a common witness. 
\begin{proposition}\label{eqlengthprimitivescorr}
Let $G = \{(u_i,v_i) \mid i \in \{1, \ldots, k\}\}$ be a set of $k \geq 2$ conjugate pairs such that all pairs in $\roots{G}$ are of equal length. If $(u_1,v_1)^*(u_2,v_2)^* \cdots (u_k,v_k)^*$ is conjugate then either $x_1 = x_2 = \cdots = x_k $ or $y_1 = y_2 = \cdots = y_k$ where $(x_i,y_i)$ is a cut of the primitive pair $(\rho_{u_i},\rho_{v_i})$ for $i \in [k]$.
\end{proposition}

\begin{proof}
Proof is by induction on the number of pairs in $\roots{G}$.
\begin{description}
    \item[Base Case:] When $\roots{G}$ has only 2 pairs, i.e., let $\roots{G} = \{(\rho_{u_1},\rho_{v_1}),(\rho_{u_2},\rho_{v_2})\}$. There exist $\ell_1,\ell_2$ such that $\ell_2 > \ell_1 + 2, \ell_1 > 2$ and $(\rho_{u_1},\rho_{v_1})^{l_1} (\rho_{u_2},\rho_{v_2})^{l_2} \in (u_1,v_1)^*(u_2,v_2)^*$ and hence it is conjugate. From \nameref{eqlengthprimitives} we get either $x_1 = x_2$ or $y_1 = y_2$.

We also consider the case when $\roots{G}$ has exactly 3 pairs, i.e., let $\roots{G} = \{(\rho_{u_1},\rho_{v_1}),(\rho_{u_2},\rho_{v_2}), (\rho_{u_3},\rho_{v_3})\}$. For  $i,j\in [3]$, there exists  set of pairs $\{(\rho_{u_i},\rho_{v_i})^{\ell_i}(\rho_{u_{j}},\rho_{v_{j}})^{\ell_{j}} \mid \ell_{j} > \ell_i + 2, \ell_i > 2 \} \subset (u_1,v_1)^*(u_2,v_3)^*(u_{3},v_{3})^*$ that is conjugate. Thus, the pairs $(\rho_{u_i},\rho_{v_i})$ and $(\rho_{u_{j}},\rho_{v_{j}})$ satisfy either $x_i = x_{j}$ or $y_i = y_{j}$ by \nameref{eqlengthprimitives} where $(x_i,y_i)$ and $(x_j,y_j)$ are cuts of $(\rho_{u_1},\rho_{v_1})$ $(\rho_{u_j},\rho_{v_j})$ respectively. There are 3 cases.
\begin{description}
    \item[$1.$] All pairs in $\roots{G}$ are identical. Then there exists cuts $(x_i,y_i)$ for each pair $(\rho_{u_i},\rho_{v_j}) \in \roots{G}$ for $i \in [3]$ such that $x_1 = x_2=x_3 = \varepsilon$. Hence, the proposition holds. 

    \item[$2.$] Exactly one pair in $\roots{G}$ is identical. Let $(\rho_{u_1},\rho_{v_1})$ and $(\rho_{u_2},\rho_{v_2})$ are distinct pairs with unique nonempty cuts $(x_1,y_1)$ and $(x_2,y_2)$ respectively. By \Cref{eqlengthprimitives}, either $x_1=x_2$ or $y_1=y_2$. Assume that $x_1=x_2$ Also, since $(\rho_{u_3},\rho_{v_3})$ is identical pair, either $\rho_{u_3}=x_i$ or $\rho_{u_3}=y_i$ for $i\in [1,2]$. If there exists an $i \in [1,2]$ such that $\rho_{u_3}=x_i$, then $x_1=x_2 = \rho_{u_3} = x_3$ where $(x_3,y_3) = (\rho_{u_3},\varepsilon)$ is a cut of $(\rho_{u_3},\rho_{v_3})$. Suppose for all $i \in [1,2]$, $\rho_{u_3}=y_i$, then $y_1=y_2 = \rho_{u_3} = y_3$ where $(x_3,y_3) = (\varepsilon,\rho_{u_3})$ is a cut of $(\rho_{u_3},\rho_{v_3})$. Hence, the proposition holds.

    \item[$3.$] Exactly one pair in $\roots{G}$ is non-identical.  WLOG assume that $(\rho_{u_1},\rho_{v_1})$ and $(\rho_{u_3},\rho_{v_3})$ are the identical pairs,  $(\rho_{u_2},\rho_{v_2})$ is non-identical with unique nonempty cut $(x_2,y_2)$. 
   
    \begin{claim}\label{claim:samerho}
        $\rho_{u_1}=\rho_{u_3}$.
    \end{claim}
    \begin{proof}
        Consider the pair $(u,v)=(u_1,v_1)^j(u_2,v_2)^{l}(u_3,v_3)^j$ where $j$ and $l$ are set such that $|u_2^l| > 4 \cdot |u_1^j| = 4 \cdot |u_3^j|$. Substituting with their primitive roots, we get
        \begin{align*}
            u &= \overbrace{u_1 \cdots u_1}^{B_1}\overbrace{u_2 \cdots u_2}^{B_2}\overbrace{u_3 \cdots u_3}^{B_3} = \rho_{u_1} \cdots \rho_{u_1}x_2y_2 \cdots x_2y_2 \cdots x_2y_2\rho_{u_3} \cdots \rho_{u_3}\\
            v &= v_1 \cdots v_1v_2 \cdots v_2v_3 \cdots v_3 =\rho_{u_1} \cdots \rho_{u_1}y_2x_2 \cdots y_2x_2 \cdots y_2x_2\rho_{u_3} \cdots \rho_{u_3}
        \end{align*}
        Let $B_1,B_2,B_3$ denote the block of $u_1$'s, $u_2$'s and $u_3$'s. Similarly, $B_1',B_2',B_3'$ denote the block of $v_1$'s, $v_2$'s and $v_3$'s.
        Since $(u,v)$ is conjuagte, it has a cut $(p,q)$. We do a case analysis on $p$.
        \begin{enumerate}
            \item $p$ ends within $B_1(\rho_{u_1})^{-1}$. In this case, $p\in \rho_{u_1}^*$ by Case~\ref{same} of \nameref{cutcorr} since there exist atleast a factor $\rho_{u_1}$ of $q$ in $u$ to match with the prefix $\rho_{u_1}$ of $q$ in $v$.  On further matching $q$ in $u$ and $v$, since $|\rho_{u_1}|=|x_2y_2| = |y_2x_2| = |\rho_{u_3}|$, we get $\rho_{u_1}=x_2y_2 = y_2x_2 = \rho_{u_3}$.
             \item $p$ ends after block $B_2\rho_{u_3}$. This case is symmetric to the above case, where we get $q \in \rho_{u_3}^*$ by Case~\ref{same} of \nameref{cutcorr}. On further matching $p$ in $u$ and $v$, since $|\rho_{u_3}|=|x_2y_2| = |y_2x_2| = |\rho_{u_1}|$, we get $\rho_{u_3}=x_2y_2 = y_2x_2 = \rho_{u_1}$.
            \item  $p$ ends within $\rho_{u_1}B_2\rho_{u_3}$. WLOG assume that $p$ ends within the first half of $\rho_{u_1}B_2\rho_{u_3}$. Since $|\rho_{u_1}B_2\rho_{u_3}|/2 > |B_2|/2 \geq 2. |u_1^j|$, on matching $q$ in $u$ and $v$, the prefix $u_1^j$ of $v$ matches entirely with block $B_2$ and there exist at least one factor of $x_2y_2$ following that. Using \nameref{cutcorr}, the prefix $p (u_1^j)$ in $u$ ends in $x_2$. Since $|\rho_{u_1}| = |x_2y_2| = |y_2x_2|$, we get $\rho_{u_1} = y_2x_2 = x_2y_2$, and this contradicts that $(u_2,v_2)$ is a non-identical pair. Similarly, we get a contradiction when $p$ ends withing second half of $\rho_{u_1}B_2\rho_{u_3}$. \qedhere
           
        \end{enumerate}
        
    \end{proof}
    Thus, we get $\rho_{u_1}=\rho_{u_3}$ when exactly one pair in $\roots{G}$ is non-identical. Since $(x_2,y_2)$ is the unique nonempty cut of $(u_2,v_2)$, by \Cref{eqlengthprimitives}, either $x_2=\rho_{u_1} = \rho_{u_3}$ or $y_2 = \rho_{u_1}= \rho_{u_3}$. In either case the proposition holds. 
     \end{description}

     \item[Inductive Case:]
Assume that the statement holds for $k$ pairs in $\roots{G}$, i.e.,
let
\[
\roots{G}= \{(\rho_{u_1},\rho_{v_1}), \ldots, (\rho_{u_k},\rho_{v_k})\}.
\]
By the induction hypothesis, if $(u_1,v_1)^* \cdots (u_k,v_k)^*$
is conjugate, then either $x_1=\cdots=x_k \quad \text{or} \quad y_1=\cdots=y_k$. Without loss of generality assume that $x_1=\cdots=x_k$.
We prove that the statement also holds for $k+1$ pairs.

Let $G' \supseteq G$ such that
$\roots{G'}=\roots{G}\cup \{(\rho_{u_{k+1}},\rho_{v_{k+1}})\},
$ where $(\rho_{u_{k+1}},\rho_{v_{k+1}})$ is the conjugate primitive root of
$(u_{k+1},v_{k+1})\in G'$ and has the same length as the pairs in $\roots{G}$.
Assume that $(u_1,v_1)^*\cdots (u_k,v_k)^*(u_{k+1},v_{k+1})^*$ is conjugate. Then for every $i\in[k]$ there exists a set of pairs
\[
\{(\rho_{u_i},\rho_{v_i})^{\ell_i}(\rho_{u_{k+1}},\rho_{v_{k+1}})^{\ell_{k+1}}
\mid
\ell_{k+1}>\ell_i+2,\ \ell_i>2
\}
\subseteq
(u_1,v_1)^*\cdots (u_k,v_k)^*(u_{k+1},v_{k+1})^*
\]
that is conjugate. Hence, by \nameref{eqlengthprimitives},
for every $i\in[k]$ the pairs
$(\rho_{u_i},\rho_{v_i})$ and $(\rho_{u_{k+1}},\rho_{v_{k+1}})$ satisfy
either
\[
x_i=x_{k+1}
\quad\text{or}\quad
y_i=y_{k+1},
\]
where $(x_i,y_i)$ and $(x_{k+1},y_{k+1})$ are cuts of
$(\rho_{u_i},\rho_{v_i})$ and $(\rho_{u_{k+1}},\rho_{v_{k+1}})$ respectively.

We partition $\roots{G}$ into two sets: $I=\{(\rho,\rho')\in\roots{G}\mid \rho=\rho'\}$
containing all identical pairs, and $J=\roots{G}\setminus I$ containing all non-identical pairs.
We distinguish the following cases.

\begin{enumerate}

\item If $J=\emptyset$, then all pairs in $\roots{G}$ are identical.
Hence there exist cuts $(x_i,y_i)$ such that
\[
x_1=\cdots=x_k=\varepsilon .
\]
If $(\rho_{u_{k+1}},\rho_{v_{k+1}})$ is also identical, we can take
$x_{k+1}=\varepsilon$ and the statement follows.

If $(\rho_{u_{k+1}},\rho_{v_{k+1}})$ is non-identical, it has a unique
nonempty cut $(x_{k+1},y_{k+1})$.
By \Cref{eqlengthprimitives} together with \Cref{claim:samerho},
for all $i,j\in[k]$ either
\[
x_{k+1}=\rho_{u_i}=\rho_{u_j}
\quad\text{or}\quad
y_{k+1}=\rho_{u_i}=\rho_{u_j}.
\]
Thus
\[
\rho_{u_1}=\cdots=\rho_{u_k}=x_{k+1}
\quad\text{or}\quad
\rho_{u_1}=\cdots=\rho_{u_k}=y_{k+1},
\]
and the proposition holds.

\item Suppose $J\neq\emptyset$.

\begin{enumerate}

\item If $I=\emptyset$, then every pair in $\roots{G}$ is non-identical and
has a unique nonempty cut $(x_i,y_i)$ for $i\in[k]$.
If $(\rho_{u_{k+1}},\rho_{v_{k+1}})$ is also non-identical with unique cut
$(x_{k+1},y_{k+1})$, then by \Cref{eqlengthprimitives}
either $x_i=x_{k+1}$ for some $i\in[k]$ or
$y_i=y_{k+1}$ for all $i\in[k]$.
In the former case we obtain $x_1=\cdots=x_k=x_{k+1}$ 
by the induction hypothesis, and in the latter case $y_1=\cdots=y_k=y_{k+1}$.

If $(\rho_{u_{k+1}},\rho_{v_{k+1}})$ is identical, then
\Cref{eqlengthprimitives} implies that either
$x_i=\rho_{u_{k+1}}$ for some $i\in[k]$ or
$y_i=\rho_{u_{k+1}}$ for all $i\in[k]$.
Using the induction hypothesis we again obtain the desired equality.

\item Suppose both $I$ and $J$ are nonempty.
Without loss of generality let
\[
I=\{(\rho_{u_1},\rho_{v_1}),\ldots,(\rho_{u_{j-1}},\rho_{v_{j-1}})\}  \text{ and }
\]
\[
J=\{(\rho_{u_j},\rho_{v_j}),\ldots,(\rho_{u_k},\rho_{v_k})\}
\]
for some $1<j\le k$.
Each pair in $J$ is non-identical and hence has a unique nonempty cut.

By the induction hypothesis we obtain
\begin{equation}\label{eq:IJ}
\rho_{u_1}=\cdots=\rho_{u_{j-1}}=x_j=\cdots=x_k .
\end{equation}

If $(\rho_{u_{k+1}},\rho_{v_{k+1}})$ is non-identical with cut
$(x_{k+1},y_{k+1})$, then applying the base case for three pairs yields
either
\[
\rho_{u_1}=x_i=x_{k+1}
\quad\text{for some } i\in\{j,\ldots,k\}, \text{ or}
\]
\[
\rho_{u_1}=y_i=y_{k+1}
\quad\text{for all } i\in\{j,\ldots,k\}.
\]
Using \eqref{eq:IJ} we obtain the desired equality of all $x$'s
or all $y$'s.

If $(\rho_{u_{k+1}},\rho_{v_{k+1}})$ is identical, then
\Cref{eqlengthprimitives} implies that either
$x_i=\rho_{u_{k+1}}$ for some $i\in\{j,\ldots,k\}$ or
$y_i=\rho_{u_{k+1}}$ for all such $i$.
Combining this with \eqref{eq:IJ} and
\Cref{claim:samerho} again yields the required equality.

\end{enumerate}

\end{enumerate}
\end{description}
\end{proof}

\subsection{For a Finite Set of Pairs}
The common witness theorem for a finite set is a straightforward corollary of the below lemma.
\begin{lemma}\label{commonwitnessfinite}
Let $G = \{(u_i,v_i) \mid 1\leq i \leq k\}$ be a finite set of $k \in \N$ pairs. If the set $(u_1,v_1)^*(u_2,v_2)^* \cdots (u_k,v_k)^*$ is conjugate then $\roots{G}$ as well $G$ has a common witness.
\end{lemma}

\begin{proof}
When $k= 1$, $G$ has only one pair $(u,v)$ and by assumption it is conjugate. By \Cref{uz=zv}, $(u,v)$ has a witness. From \Cref{samewitness}, we obtain that the witnesses of $\roots{G} = \{(\rho_{u},\rho_{v})\}$ are same as that of $(u,v)$.

Next we assume that $k >1$. Let $\approx$ be the equivalence relation on $G$ whereby $(u,v) \approx (u^\prime,v^\prime)$ if $\rho_{u} \sim \rho_{u^\prime}$, i.e., the primitive roots of the pairs are conjugates. Assume that $\approx$ has $d$ equivalence classes. Clearly $1 \leq d \leq k$. We do a cases analysis on whether $d=1$ or otherwise.

If $\approx$ has only one equivalence class, then the primitive roots of all the pairs in $G$ are conjugates. Consequently, their lengths are identical. Since $(u_1,v_1)^* \cdots (u_k,v_k)^*$ is conjugate and all the pairs in $\roots{G}$ have identical lengths, by \Cref{eqlengthprimitivescorr}, $\roots{G}$ has a common witness.

Now we assume that $d > 1$. Choose $d$ pairs $(u_1, v_1), (u_2, v_2), \ldots, (u_d, v_d)$ from each equivalence class. We construct a pair $(u,v) \in (u_1,v_1)^*(u_2,v_2)^* \cdots (u_d,v_d)^* \subseteq (u_1,v_1)^* \cdots (u_k,v_k)^*$ and show that $(u,v)$ is conjugate only if $\roots{G}$ has a common witness.

Let $m$ be the least common multiple of $|u_1|, \ldots, |u_d|$. Let $\ell_{ij} = |u_i| + |u_j| - \mathit{gcd}(|u_i|,|u_j|) > 0$ for $1 \leq i,j \leq d$ and $i \neq j$. Let $\ell = \max \ \{\ell_{ij} \mid 1 \leq i,j \leq d, i \neq j\}$. Let $N$ be a multiple of $m$ that is $ > 2\ell$.

Let $(u,v) = (u_1,v_1)^{j_1}(u_2,v_2)^{j_2} \cdots (u_d,v_d)^{j_d}$ such that $j_1, \ldots, j_d > 2$ and $|u_i^{j_i}| = N$ for each $1 \leq i \leq d$.
\begin{align*}
u &= \overbrace{u_1 \cdots u_1}^{N}\overbrace{u_2 \cdots u_2}^{N} \cdots \overbrace{u_d \cdots u_d}^{N}\\
v &= v_1 \cdots v_1\,v_2 \cdots v_2\, \cdots v_d \cdots v_d
\end{align*}

Since $(u,v)$ is conjugate, it has a cut, say $(p,q)$. Substituting each pair in $(u,v)$ with their primitive roots, we get
\begin{align*}
u &= \rho_{u_1} \cdots \rho_{u_1}\rho_{u_2} \cdots \rho_{u_2} \cdots \rho_{u_d} \cdots \rho_{u_d} = pq\\
v &= \rho_{v_1} \cdots \rho_{v_1}\rho_{v_2} \cdots \rho_{v_2} \cdots \rho_{v_d} \cdots \rho_{v_d} = qp
\end{align*}
Let $(x_i,y_i)$ be a cut of $(\rho_{u_i},\rho_{v_i})$ for $1 \leq i \leq d$. 

Let $B_1, B_2, \ldots, B_d$ represent the blocks in $u$, and let $B_1^\prime, B_2^\prime, \ldots, B_d^\prime$ represent the blocks in $v$.

The cut in $u$ can be either within the first block $B_1$, or the last block $B_d$, or anywhere between the first and the last blocks in $u$. We do a case analysis on all the possible cuts $(p,q)$ of $(u,v)$ and show that there exists a common witness of $\roots{G}$ in each of the cases. We can assume that both $p$ and $q$ are nonempty. When either $p$ or $q$ is empty, then all pairs $(u_i,v_i)$ in $(u,v)$ are identical, and hence either $x_1 = \cdots = x_d = \varepsilon$ or $y_1 = \cdots =y_d= \varepsilon$. 

\medskip
\paragraph*{Case 1: When the cut in $u$ is in the first block $B_1$}

We make a further case analysis depending upon if the cut is within the first half or the second half of the first block.

Suppose the cut in $u$ is within the first half of the block, i.e., $p$ is of length at most $N/2$. In this case, since the length of each block are equal, the cut in $v$ is within the second half of the last block $B_d^\prime$, i.e., $p$ is a suffix of $v$ of length at most $N/2$.
\begin{align*}
u &= \overbrace{\rho_{u_1} \cdots}^{p}\overbrace{\cdots \rho_{u_1}\rho_{u_2} \cdots \rho_{u_2} \cdots \rho_{u_d} \cdots \rho_{u_d}}^{q}\\
v &= \underbrace{\rho_{v_1} \cdots \cdots \rho_{v_1}\rho_{v_2} \cdots \rho_{v_2} \cdots \rho_{v_d} \cdots}_{q} \underbrace{\cdots \rho_{v_d}}_{p}
\end{align*}
We obtain $q = p^{-1}(B_1B_2 \cdots B_d) = (B_1^\prime B_2^\prime \cdots B_d^\prime)p^{-1} \ .$
\begin{claim}
The following holds for each $1 \leq i \leq d$.
\begin{enumerate}
\item $B_i$ is of the form $pq_i$ and $B_i^\prime$ is of the form $q_ip$, and
\item $p$ is of the form $(x_iy_i)^{m_i}x_i$ and $q_i$ is of the form $(y_ix_i)^{m_i}y_i$ for some $m_i \geq 0$.
\end{enumerate}
\end{claim}

\begin{proof}
Proof is by induction on $i$.
\begin{enumerate}
\item \emph{Base Case:} when $i=1$. We compare the prefixes of $q$ in $u$ and $v$. Since $|p| \leq N/2$, the prefix of $q$ in $u$ must begin within the first block $B_1$. Also, there must be at least one occurrence of the factor $\rho_{u_1} = x_1y_1$ following the cut. Since $q$ in $v$ starts with $\rho_{v_1} = y_1x_1$, the cut should be at the end of $x_1$ by Case~\ref{distinct} of \nameref{cutcorr}. Note that by Case~\ref{same} of \nameref{cutcorr}, the cut can end after $y_1$ but then $\rho_{u_1}=x_1y_1=y_1x_1 = \rho_{v_1}$ and hence either $x_1 = \varepsilon$ or $y_1 = \varepsilon$. Hence, in both case, we get $p=(x_1y_1)^{m_1}x_1$ for some integer $m_1 \geq 0$. Consequently, the prefix of $q$ in the block $B_1$, denoted as $q_1$, is of the form $y_1(x_1y_1)^{n_1}$, for some $n_1 > 0$. After matching $q_1$ in $v$, we observe that a factor equal to $p$ appears in the suffix of the block $B_1^\prime$, as shown below.
\begin{align*}
u &= \overbrace{(x_1y_1)^{m_1}x_1}^{p}\overbrace{y_1(x_1y_1)^{n_1}}^{q_1}\overbrace{\rho_{u_2} \cdots \rho_{u_2} \cdots \rho_{u_d} \cdots \rho_{u_d}}^{{q_1}^{-1}q}\\
v &= \underbrace{y_1(x_1y_1)^{n_1}}_{q_1}\underbrace{(x_1y_1)^{m_1}x_1}_{=p}\rho_{v_2} \cdots \rho_{v_2} \cdots \rho_{v_d} \cdots \rho_{v_d} = qp
\end{align*}

\item \emph{Inductive Case:} Assume the claim is true for first $i$ blocks where $1 \leq i < k$.
\begin{align*}
u &= \overbrace{(x_1y_1)^{m_1}x_1}^{p}\overbrace{y_1(x_1y_1)^{n_1}}^{q_1}\cdots \overbrace{(x_iy_i)^{m_i}x_i}^{p}\overbrace{y_i(x_iy_i)^{n_i}}^{q_i}\overbrace{\rho_{u_{i+1}} \cdots \rho_{u_{i+1}} \cdots \rho_{u_d} \cdots \rho_{u_d}}^{\suff{q} = (q_1p \cdots q_{i-1}pq_i)^{-1}q} = pq\\
v &= \underbrace{y_1(x_1y_1)^{n_1}}_{q_1}\underbrace{(x_1y_1)^{m_1}x_1}_{p}\cdots \underbrace{y_i(x_iy_i)^{n_i}}_{q_i}\underbrace{(x_iy_i)^{m_i}x_i}_{p}\rho_{v_{i+1}} \cdots \rho_{v_{i+1}} \cdots \rho_{v_d} \cdots \rho_{v_d} = qp
\end{align*}
Let $\suff{q}$ denote the remaining suffix of $q$ in $u$ after the $i$-th block $B_i$. We obtain $\suff{q} = (q_1p \cdots q_{i-1}pq_i)^{-1}q$. By comparing the prefixes of $\suff{q}$ in $u$ and $v$, we get that the suffix of the block $B_i^\prime$ in $v$, that is equal to $p$, matches within the first half of the block $B_{i+1}$ in $u$ since $|p| < N/2$. Moreover, the matching should end at $x_{i+1}$ by Case~\ref{distinct} %\ref{cutc}, \ref{cutd}, \ref{cute} and \ref{same}
of \nameref{cutcorr}. Note that by Case~\ref{same} of \nameref{cutcorr}, the matching can end after $y_{i+1}$ but then $\rho_{v_{i+1}}=y_{i+1}x_{i+1}=x_{i+1}y_{i+1} = \rho_{u_{i+1}}$ and hence either $x_{i+1} = \varepsilon$ or $y_{i+1} = \varepsilon$. Hence, $p = (x_{i+1}y_{i+1})^{m_{i+1}}x_{i+1}$ for some integer $m_{i+1} \geq 0$. Consequently, the remaining suffix of the block $B_{i+1}$, denoted by $q_{i+1}$, is of the form $(y_{i+1}x_{i+1})^{n_{i+1}}y_{i+1}$ for some integer $n_{i+1} > 0$. After matching $q_{i+1}$ in $B_{i+1}^\prime$, we observe that a factor equal to $p$ appears in the suffix of the block $B_{i+1}^\prime$. Hence, $B_{i+1} = pq_{i+1}$ and $B_{i+1}^\prime = q_{i+1}p$.

\end{enumerate}

\end{proof}

From the above claim, it follows that $q = q_1pq_2p \cdots pq_d$ and \[p = (x_1y_1)^{m_1}x_1 = (x_2y_2)^{m_2}x_2 = \cdots = (x_dy_d)^{m_d}x_d\]
for $m_1, \ldots, m_d \geq 0$. 

From \Cref{Klthm:claim2}, the above equation holds for any two pairs between the equivalence classes. By \Cref{cwpattern}, we obtain $p$ is a common inner witness of $\roots{G}$.

Next we assume the cut in $u$ is in the second half of $B_1$. We compare the prefixes of $q$ in $u$ and $v$ and deduce that there exist a common factor of length at least $N/2 > \ell > \ell_{12}$ between block $B_1^\prime$ and block $B_2$. From \Cref{commonfactor}, $\rho_{v_1}$ is conjugate to $\rho_{u_{2}}$, that is in turn is conjugate to $\rho_{v_{2}}$. From transitivity of conjugacy, $\rho_{v_1}$ is conjugate to $\rho_{v_2}$, which contradicts the fact that $(u_1,v_1)$ and $(u_2,v_2)$ belong to different equivalence classes. Hence cut in second half of $B_1$ is not possible.

\medskip
\paragraph*{Case 2: When the cut in $u$ is in the last block $B_d$}

We make a further case analysis depending upon if the cut in $u$ is in the first half or second half of the last block $B_d$. The proof is symmetric to that of the previous case.  In particular, when the cut is in the second half of the last block, \[q = (y_1x_1)^{m_1}y_1 = (y_2x_2)^{m_2}y_2 = \cdots = (y_dx_d)^{m_d}y_d \ .\]
By  \Cref{Klthm:claim2} and \Cref{cwpattern}, $q$ is a common outer witness of $\roots{G}$. 

\medskip
\paragraph*{Case 3: When the cut in $u$ is within the block $B_j$ for $1 < j < d$}

WLOG, assume that the cut in $u$ is within the first half of the block $B_j$. In this case the cut in $v$ will be within the second half of the block $B_{d-j+1}$.
\begin{align*}
u &= \overbrace{\rho_{u_1} \cdots \rho_{u_1} \cdots \rho_{u_j} \cdots}^{p}\overbrace{\cdots \rho_{u_j} \cdots \cdots \cdots \cdots \rho_{u_d} \cdots \rho_{u_d}}^{q}\\
v &= \underbrace{\rho_{v_1} \cdots \rho_{v_1} \cdots \cdots \cdots \rho_{v_{d-j+1}} \cdots}_{q}\underbrace{\cdots \rho_{v_{d-j+1}} \cdots \rho_{u_d} \cdots \rho_{u_d}}_{p}
\end{align*}
By matching $q$ in $u$ and $v$, we get $\rho_{v_1}$ and $\rho_{u_j}$ shares a common factor of length at least $N/2 > \ell > \ell_{1j}$ and hence they are conjugates to each other by \Cref{commonfactor}. Since $\rho_{u_j}$ is conjugate to $\rho_{v_j}$, by transitivity of conjugacy we obtain $\rho_{v_1}$ and $\rho_{v_j}$ are conjugates. This contradicts the fact that $(u_1,v_1)$ and $(u_j,v_j)$ belongs to different equivalence classes. Hence cut in $B_j$ where $1 < j < d$ is not possible.

\begin{claim}\label{Klthm:claim2}
Let $(u_i,v_i)$ and $(u_j,v_j)$ be two pairs in different equivalence class  where $i,j \in [d]$. Let $(x_i,y_i)$ and $(x_j,y_j)$ be a cut of the primitive root of $(u_i,v_i)$ and $(u_j,v_j)$ respecively. The following holds.
\begin{enumerate}
\item If $(x_iy_i)^{m_i}x_i = (x_jy_j)^{m_j}x_j$, then for any pair $(u_{j'},v_{j'})$ that belongs to the same equivalence class as $(u_j,v_j)$, $(x_iy_i)^{m_i}x_i = (x_{j'}y_{j'})^{m_j}x_{j'}$ with $(x_{j'},y_{j'})$ being a cut of the primitive root of $(u_{j'},v_{j'})$.
\item  If $(y_ix_i)^{m_i}y_i = (y_jx_j)^{m_j}y_j$, then for any pair $(u_{j'},v_{j'})$ that belongs to the same equivalence class as $(u_j,v_j)$, $(y_ix_i)^{m_i}y_i = (y_{j'}x_{j'})^{m_j}y_{j'}$ with $(x_{j'},y_{j'})$ being a cut of the primitive root of $(u_{j'},v_{j'})$.
\end{enumerate}
\end{claim}

\begin{proof}
We prove Item 1. The proof of Item 2 is symmetric. 

For contradiction, assume that $(x_iy_i)^{m_i}x_i = (x_jy_j)^{m_j}x_j \neq  (x_{j'}y_{j'})^{m_j}x_{j'}$. 
Construct a pair $(\bar{u},\bar{v})$ similar to $(u,v)$ by keeping only the pairs $(u_i,v_i), (u_j,v_j)$ and $(u_{j'},v_{j'})$ such that the total length of each pairs are equal to $N$. WLOG, assume, $i < j < j'$
$$\bar{u} = \overbrace{u_{i} \cdots u_{i}}^{N} \overbrace{u_{j} \cdots u_{j}}^{N} \overbrace{u_{j'} \cdots u_{j'}}^{N}$$
$$\bar{v} = v_{i} \cdots v_{i} v_{j} \cdots v_{j} v_{j'} \cdots v_{j'}$$

Since $(\bar{u},\bar{v}) \subseteq (u_1,v_1)^* \cdots (u_k,v_k)^*$, it is conjugate with a cut. After doing a case analysis on the cut, we get three possible cases.
\begin{enumerate}
\item There exist integers $n_{i}, n_{j}, n_{j'} \geq 0$ such that $(x_{i}y_{i})^{n_{i}}x_{i} = (x_{j}y_{j})^{n_{j}}x_{j} = (x_{j'}y_{j'})^{n_{j'}}x_{j'}$ .

Both $n_{i} = m_i$ and $n_{j} = m_j$. Otherwise if either $n_{i} \neq m_i$ or $n_{j} \neq m_j$, then the set of pairs $\{(u_{i},v_{i}), (u_{j},v_{j})\}$ has more than one common witness. By \Cref{infmanycwcorr}, they have infinitely many common witnessess and their primitive roots are the same. This condradicts that $(u_i,v_i)$ and $(u_j,v_j)$ are in different equivalence class. 
Since $(u_j,v_j)$ and $(u_{j'},v_{j'})$ belongs to the same equivalence class, $|x_{j}y_{j}| = |x_{j'}y_{j'}|$. Hence , we get $n_{j'} = n_j = m_j$. Thus, $(x_iy_i)^{m_i}x_i = (x_jy_j)^{m_j}x_j = (x_{j'}y_{j'})^{m_j}x_{j'}$

\item There exist integers $n_{i}, n_{j}, n_{j'} \geq 0$ such that $(y_{i}x_{i})^{n_{i}}y_{i} = (y_{j}x_{j})^{n_{j}}y_{j} = (y_{j'}x_{j'})^{n_{j'}}y_{j'}$

Since $(x_{i}y_{i})^{m_i}x_{i} = (x_{j}y_{j})^{m_j}x_{j}$, we get there is more than one common witness for the set of pairs $\{(u_{i},v_{i}), (u_{j},v_{j})\}$ and from \Cref{infmanycwcorr}, we get they have infinitely many common witnesses. Thus the primitive roots are the same, and condradicts that $(u_i,v_i)$ and $(u_j,v_j)$ are in different equivalence class. 

\item The block of $u_{i}$'s, $u_{j}$'s and $u_{j'}$'s share a common factor of length atleast at least $N/2 \geq \ell$ and via \Cref{commonfactor}, we get their primitive roots are conjugates to each other. This contradicts that $(u_i,v_i)$ and $(u_j,v_j)$ are in different equivalence class. \qedhere
\end{enumerate} 
\end{proof}

Hence, for a finite set of pairs $G$, $(u_1,v_1)^* \cdots (u_k,v_k)^*$ is conjugate only if $\roots{G}$ has a common witness. By \Cref{samewitnessext}, we also conclude that $G$ has a common witness.
\end{proof}

Hence, we proved \Cref{1to4cw} for the finite case.

\subsection{For an Infinite Set of Pairs}
We now extend \Cref{1to4cw} from a finite set to an infinite set of pairs.

\begin{lemma}[Compactness Theorem]\label{compactness}
Let $G$ be an infinite set of pairs. If every finite subset of $G$ has a common witness, then $G$ has a common witness.
\end{lemma}

\begin{proof}
From \Cref{infmanycwcorr}, if a set has a witness, it has exactly one common witness or infinitely many common witnesses. Given that every finite subset of $G$ has a common witness, there are two possible cases: a finite subset of $G$ with a unique witness exists, or every finite subset of $G$ has infinitely many witnesses.
\begin{enumerate}
\item Assume that there exists a finite subset $G_f$ of $G$ with exactly one common witness, say $z$. We claim that $z$ is a common witness of $G$ as well. By assumption, the finite set $G_f \cup \{(u,v)\}$ has a common witness, for any pair $(u,v) \in G$. Moreover, the witness for this set must be $z$; otherwise, it contradicts the uniqueness of the witness of $G_f$. This implies that $z$ is a witness for any pair in $G$. Hence $z$ is a common witness of $G$.

\item Next we assume that every finite subset of $G$ has infinitely many common witnesses. Take any pair $(u_i,v_i)$ and $(u_j,v_j)$ from $G$. The set $\{(u_i,v_i),(u_j,v_j)\}$ is a finite set with infinitely many witnesses by assumption. Therefore, from \Cref{infmanycwcorr}, both $(u_i,v_i)$ and $(u_j,v_j)$ have the same primitive root. Since primitive roots are unique, the primitive root of every pair in $G$ is the same. From \Cref{samewitness}, the witnesses of the primitive root is same as that of the witnesses of each pair in $G$. Hence, $G$ has a common witness.
\end{enumerate}
\end{proof}
The proof of \Cref{1to4cw} is a straightforward corollary of the \nameref{compactness}. If $\spans{G}$ is conjugate, then the closure of every finite subset of $G$ is also conjugate. Using \Cref{commonwitnessfinite}, every finite subset of $G$ has a common witness. Using \nameref{compactness}, $G$ has a common witness. From \Cref{samewitnessext}, we conclude that $\roots{G}$ has a common witness.

This concludes the proof of Common Witness Theorem (\Cref{conjugacy}).

% !TEX root = main.tex

\section{Existence of Common Witness for Monoid Closure}\label{sec:conjugacy_monomial}
In this section, we prove the equivalence between conjugacy and the presence of a common witness in sumfree sets. We begin by proving \Cref{singlemonomial} for sumfree sets that contain only one Kleene star. Subsequently, we establish \Cref{generalmonomialwitness} that extends the result to general sumfree sets.

\subsection{Common Witness of a Singleton Redux}

We prove \Cref{singlemonomial} by showing $(1) \implies (2) \implies (3) \implies (1)$.  In this subsection, we assume that the sumfree set has a nonempty redux as stated in \Cref{singlemonomial}. It is trivial that $(3) \implies (1)$.

Now we proceed to prove $(1) \implies (2)$, namely, if a sumfree set $M =(\as_0,\bs_0)G^*(\as_1,\bs_1)$ is conjugate, then there exists a common witness of $G \cup \{(\as_1\as_0,\bs_1\bs_0)\}$. We first prove this direction when $G$ is just a singleton set and later generalise it to any arbitrary set of pairs $G$.

\begin{proposition}\label{monomialcut}
If a set $M = (\as_0,\bs_0)(u,v)^*(\as_1,\bs_1)$ is conjugate then the set $\{(u,v), (\as_1\as_0,\bs_1\bs_0)\}$ has a common witness.
\end{proposition}

\begin{proof}
Since $M$ is conjugate, $(u,v)$ is also conjugate by \Cref{boundedkleenestar}. 
Consider the pair $(u^\prime,v^\prime) = (\as_0,\bs_0)(u,v)^{4n}(\as_1,\bs_1)$ where $n|u| \geq 2 \cdot \textit{(length of the redux of $M$)}$. Let $(x,y)$ denote a cut of the primitive root of the conjugate pair $(u,v)$. 

We now examine the possible cuts, say $(p,q)$, of $(u^\prime,v^\prime)$ and show that in each case, a common witness of $\{(u,v), (\as_1\as_0,\bs_1\bs_0)\}$ exists.

\medskip
\paragraph*{Case 1: When the cut in $u^\prime$ is within $\as_0$} Then $p = \pref{\as_0}$ is a prefix of $\as_0$ and $\as_0=\pref{\as_0}\suff{\as_0}$ for some word $\suff{\as_0}$. Substituting $(u,v)$ with powers of $(xy,yx)$,
\begin{align*}
u^\prime &= \overbrace{\pref{\as_0}}^{p}\overbrace{\suff{\as_0}xy\cdots xy \as_1}^{q}\\
v^\prime &= \bs_0 yx \cdots yx \bs_1
\end{align*}
Comparing prefixes of $q$ in $u^\prime$ and $v^\prime$, we obtain three possible cases for $\bs_0$.
\begin{enumerate}[label=(\alph*)]
\item \label{pcase1a} $\bs_0$ is a proper prefix of $\suff{\as_0}$:
After matching $\bs_0$ with the prefix of $q$ in $u^\prime$, we find that the remaining suffix of $\suff{\as_0}$ matches with the prefix of the block $yx \cdots yx$ in $v^\prime$. Since the total length of the block $yx \cdots yx$ is greater than $4|\as_0|$, it follows that $\suff{\as_0}$ must end within the first half of the block. Furthermore, using Cases \ref{cuta}, \ref{cutb}, \ref{cute} and \ref{same} of \nameref{cutcorr}, it should end after a $y$ since there exists an $xy$ after $\suff{\as_0}$ in $u^\prime$. Thus, we can express $\suff{\as_0}$ as
\begin{equation}
\suff{\as_0} = \bs_0 (yx)^m y \label{sing:eq1}
\end{equation}
for some integer $m \geq 0$. Continuing to match $q$ in $u^\prime$ and $v^\prime$, we obtain
\begin{align}
u^\prime &= \pref{\as_0}\suff{\as_0}\overbrace{xy \cdots x}^{=}(yx)^{m}y{\as_1} = pq\\
v^\prime &= \underbrace{\bs_0 (yx)^{m}y}_{\suff{\as_0}}\underbrace{xy \cdots x}_{=}\bs_1 = qp = \suff{\as_0}\overbrace{xy \cdots x}^{=}(yx)^{m}y{\as_1} \pref{\as_0} \label{singleton:eq1}
\end{align}
By equating the sets for $v^\prime$ on both sides of \Cref{singleton:eq1}, we get
\begin{equation}
\bs_1 = (yx)^{m}y\as_1\pref{\as_0} \ . \label{sing:eq2}
\end{equation}
Concatenating \Cref{sing:eq1} and \Cref{sing:eq2}, we obtain
\[(yx)^{m}y\as_1\as_0 =\bs_1\bs_0(yx)^{m}y \ .\]
From \Cref{uz=zv}, we get $(yx)^m y$ is a outer witness for $(\as_1\as_0, \bs_1\bs_0)$, and it is also a outer witness of $(u,v)$ using \Cref{samewitness}. Therefore, $(yx)^m y$ is a common outer witness of $\{(u,v), (\as_1\as_0, \bs_1\bs_0)\}$.
\item \label{pcase1b} The case when $\bs_0 = \suff{\as_0}$:
\begin{align}
u^\prime &= \pref{\as_0}\suff{\as_0}\overbrace{xy \cdots xy}^{=}{\as_1} = pq\\
v^\prime &= \underbrace{\bs_0}_{\suff{\as_0}}\underbrace{yx \cdots yx}_{=}\bs_1 = qp = \suff{\as_0}\overbrace{xy \cdots xy}^{=}{\as_1} \pref{\as_0} \label{singleton:eq2}
\end{align}
Equating $v^\prime$ on both sides of the \Cref{singleton:eq2}, we get that $xy=yx$ and
\begin{equation}
\bs_1 = {\as_1} \pref{\as_0} \ . \label{sing:eq3}
\end{equation}
Appending the equation $\bs_0 = \suff{\as_0}$ to \Cref{sing:eq3}, we get $\as_1\as_0 = \bs_1\bs_0$ . Hence $(\as_1\as_0,\bs_1\bs_0)$ is conjuagte with $\varepsilon$ as a witness. Since $xy=yx$, we can also deduce that $(u,v)$ is an identical pair with $\varepsilon$ as a witness. Therefore, $\varepsilon$ is a common witness of $\{(u,v), (\as_1\as_0,\bs_1\bs_0)\}$.

\item \label{pcase1c} $\suff{\as_0}$ is a proper prefix of $\bs_0$: Since the total length of block $xy \cdots xy$ is at least 4$|\bs_0|$, it follows that $\bs_0$ must end within the first half of the block of $xy$. Moreover, it should end after an $x$ by Cases \ref{cutc}, \ref{cutd}, \ref{cute} and \ref{same} of \nameref{cutcorr} since there is at least one $yx$ after $\bs_0$ in $v^\prime$. Therefore,
\begin{equation}
\bs_0 = \suff{\as_0}(xy)^{m}x \label{sing:4}
\end{equation}
for some integer $m \geq 0$. Continuing with the analysis, we have:
\begin{align}
u^\prime &= \pref{\as_0}\overbrace{\suff{\as_0}(xy)^{m}x}^{\bs_0}\overbrace{yx \cdots y}^{=}{\as_1} = pq\\
v^\prime &= \bs_0\underbrace{yx \cdots y}_{=}(xy)^{m}x\bs_1 = qp = \overbrace{\suff{\as_0}(xy)^{m}x}^{\bs_0}\overbrace{yx \cdots y}^{=}{\as_1} \pref{\as_0} \label{singleton:eq3}
\end{align}
By equating the sets for $v^\prime$ on both sides of \Cref{singleton:eq3}, we get
\begin{equation}
(xy)^{m}x\bs_1 = {\as_1} \pref{\as_0} \ . \label{sing:5}
\end{equation}
Concatenating \Cref{sing:4} and \Cref{sing:5}, we obtain
\[\as_1\as_0(xy)^{m}x = (xy)^{m}x \bs_1\bs_0 \ .\]
From \Cref{uz=zv}, we get $(xy)^m x$ is an inner witness of $(\as_1\as_0,\bs_1\bs_0)$, and it is also an inner witness for $(u,v)$ using \Cref{samewitness}. Thus, $(xy)^m x$ is a common inner witness of $\{(u,v), (\as_1\as_0,\bs_1\bs_0)\}$.
\end{enumerate}

\medskip
\paragraph*{Case 2: When the cut in $u^\prime$ is within the block of $(u,v) \cdots (u,v)$}
Then $p$ ends within the block of $(u,v)$'s. There are two cases based on whether the cut in $u^\prime$ is within the first half or the second half of the block of $(u,v) \cdots (u,v)$.
\begin{enumerate}[label=(\alph*)]
\item \label{pcase2a} When $p$ ends within the first half of the block of $(u,v)$'s:
\begin{align*}
u^\prime &= \as_0 \overbrace{xy \cdots xy}^{\text{cut region}}\overbrace{xy\cdots xy}^{\geq 2n \text{ times}}\as_1 = pq\\
v^\prime &= \bs_0yx \cdots yxyx\cdots yx\bs_1 =qp
\end{align*}
We compare the prefixes of $q$ in $u^\prime$ and $v^\prime$. Since the length of the remaining half of the block of $xy$'s is still greater than $2n|u| > 2 |\bs_0|$, it follows that $\bs_0$ in $v^\prime$ matches within the block of $xy$'s in $u^\prime$ and there is at least one $xy$ occurring after it. Moreover, it ends after an $x$ by Cases \ref{cutc}, \ref{cutd}, \ref{cute} and \ref{same} of \nameref{cutcorr}, as there is at least one $yx$ in $v^\prime$ after $\bs_0$. Therefore,
\begin{equation}
p\bs_0 = \as_0 (xy)^{m}x \label{sing:eq6}
\end{equation}
for some integer $m \geq 0$.
\begin{align}
u^\prime &= \overbrace{\as_0(xy)^{m}x}^{p\bs_0}\overbrace{yx\cdots y}^{=}\as_1 = pq\\
v^\prime &= \bs_0\underbrace{yx \cdots y}_{=}(xy)^{m}x\bs_1 = qp = \bs_0\overbrace{yx\cdots y}^{=}\as_1p \label{singleton:eq4}
\end{align}
By equating the sets for $v^\prime$ on both sides of \Cref{singleton:eq4}, we get
\begin{align*}
(xy)^{m}x\bs_1 = \as_1p ~ & \implies ~ (xy)^{m}x\bs_1\bs_0 = \as_1p\bs_0 &&\text{(Appending $\bs_0$)}\\
& \implies ~ (xy)^{m}x\bs_1\bs_0 = \as_1\as_0(xy)^{m}x &&\text{(Substituting \Cref{sing:eq6})}
\end{align*}

Therefore we obtain,
\[\as_1\as_0(xy)^{m}x = (xy)^{m}x\bs_1\bs_0 \ .\]
From \Cref{uz=zv}, $(xy)^m x$ is an inner witness of $(\as_1\as_0,\bs_1\bs_0)$, and it is also an inner witness for $(u,v)$ using \Cref{samewitness}. Therefore, $(xy)^m x$ is a common inner witness of $\{(u,v), (\as_1\as_0,\bs_1\bs_0)\}$.
\item \label{pcase2b} When $p$ ends within the second half of the block of $(u,v)$'s:
\begin{align*}
u^\prime &= \as_0 \overbrace{xy \cdots xy}^{\geq 2n \text{ times}}\overbrace{xy\cdots xy}^{\text{cut region}}\as_1 = pq\\
v^\prime &= \bs_0yx \cdots yxyx\cdots yx\bs_1 = qp
\end{align*}
We compare the suffixes of $p$ in $u^\prime$ and $v^\prime$. Since the suffix of $p$ within the block $xy$ is still greater than $2n|u| > 2 |\bs_1|$, it follows that the suffix $\bs_1$ in $v^\prime$ matches within the block of $xy$'s in $u^\prime$ and there is at least one $xy$ occurring before it. Moreover, it starts with a $y$ by Cases \ref{cutc}, \ref{cutd}, \ref{cute} and \ref{same} of \nameref{cutcorr} since there is at least one $yx$ before $\bs_1$. Hence,
\begin{equation}
\bs_1q = (yx)^{m}y\as_1 \label{sing:eq8}
\end{equation}
for some integer $m \geq 0$.
\begin{align}
u^\prime &= \as_0\overbrace{xy\cdots x}^{=}\overbrace{(yx)^{m}y\as_1}^{\bs_1q} = pq\\
v^\prime &= \bs_0(yx)^{m}y\underbrace{xy \cdots x}_{=}\bs_1 = qp = q\as_0\overbrace{xy\cdots x}^{=}\bs_1 \label{singleton:eq5}
\end{align}
By equating $v^\prime$ on both sides of the \Cref{singleton:eq5}, we get
\begin{align*}
\bs_0(yx)^{m}y = q\as_0 ~ & \implies ~ \bs_1\bs_0(yx)^{m}y = \bs_1q\as_0 &&\text{(Concatenating $\bs_1$ on the left side)}\\
& \implies ~ \bs_1\bs_0(yx)^{m}y = (yx)^{m}y\as_1\as_0 &&\text{(Substituting \Cref{sing:eq8})}
\end{align*}
Therefore, we obtain
\[(yx)^{m}y\as_1\as_0 =\bs_1\bs_0(yx)^{m}y \ .\]
From \Cref{uz=zv}, we get $(yx)^m y$ is an outer witness of $(\as_1\as_0,\bs_1\bs_0)$, and it is also an outer witness for $(u,v)$ using \Cref{samewitness}. Therefore, $(yx)^{m}y$ is a common outer witness of $\{(u,v), (\as_1\as_0,\bs_1\bs_0)\}$.
\end{enumerate}

\medskip
\paragraph*{Case 3: When the cut in $u^\prime$ is within $\as_1$} Then $q = \suff{\as_1}$ is a suffix of $\as_1$ and $\as_1=\pref{\as_1}\suff{\as_1}$ for some word $\pref{\as_1}$.
\begin{align*}
u^\prime &= \overbrace{\as_0xy\cdots xy \pref{\as_1}}^{p}\overbrace{\suff{\as_1}}^{q}\\
v^\prime &= \bs_0 yx \cdots yx \bs_1
\end{align*}
This case is symmetric to \emph{Case 1}, where the cut in $u^\prime$ is within $\as_0$.

Thus, the pair $(u',v') \in M$ is conjugate only if there exist a common witness for the set $\{(u,v),(\as_1\as_0,\bs_1\bs_0)\}$.
\end{proof}

Now we extend the above proposition to an arbitrary set $G$.

\begin{lemma}\label{gmonomialcut}
If a set $M= (\as_0,\bs_0)G^*(\as_1,\bs_1)$ is conjugate then one of the following is true:
\begin{enumerate}
\item \label{gmonomialcut:1} $|\calW(G)|=\infty$ and $R(G) \cup \{(\as_1\as_0, \bs_1\bs_0)\})$ has a common witness.
\item \label{gmonomialcut:2} $|\calW(G)|=1$ and there exists a pair $(u,v) \in G$ such that $W(G)=W(G \cup \{(\as_1\as_0,\bs_1\bs_0)\})=W(\{(u,v),(\as_1\as_0,\bs_1\bs_0)\}$.
\end{enumerate}
\end{lemma}

\begin{proof}
Since $M$ is conjugate, $G^*$ is conjugate by \Cref{boundedkleenestar}. Further, $G^*$ is conjugate iff $G$ has a common witness by \Cref{conjugacy}. The set $G$ has two possibilities: it either has a unique common witness or infinitely many common witnesses using \Cref{infmanycwcorr}.

Assume that $G$ has infinitely many common witnesses. From \Cref{infmanycwcorr}, each pair in $G$ have the same primitive root, say $(\rho,\rho^\prime)$. The common witnesses of $G$ are the same as that of the witnesses of $(\rho,\rho^\prime)$ using \Cref{samewitnessext}. Since $M$ is conjugate, $(\as_0,\bs_0){(\rho^n,{\rho^\prime}^n)}^*(\as_1,\bs_1)$ is conjugate for some $n \geq 1$. From \Cref{monomialcut}, there exists a common witness of $\{(\rho^n,{\rho^\prime}^n), (\as_1\as_0,\bs_1\bs_0)\}$. Furthermore, the witness of $(\rho,\rho^\prime)$ is the same as the witness of $(\rho^n,{\rho^\prime}^n)$ by \Cref{samewitness}. Therefore, we can conclude that there exists a common witness of $\{(\rho,\rho^\prime),(\as_1\as_0,\bs_1\bs_0)\}$, and thus of $G \cup \{(\as_1\as_0,\bs_1\bs_0)\}$.

Next we assume that $G$ has a unique common witness. WLOG, assume that $G$ has a unique common inner witness $z$. It suffices to show that $z$ is an inner witness of $(\as_1\as_0,\bs_1\bs_0)$. Every pair in $G^*$ has a common witness with $(\as_1\as_0,\bs_1\bs_0)$ using \Cref{monomialcut}. If there exists a pair $(u,v)$ in $G^*$ such that $\{(u,v), (\as_1\as_0,\bs_1\bs_0)\}$ has infinitely many common witnesses, then both $(u,v)$ and $(\as_1\as_0,\bs_1\bs_0)$ have the same primitive root by \Cref{infmanycwcorr}. Since $z$ is an inner witness of $(u,v)$, $z$ is an inner witness of its primitive root, and thus also an inner witness of $(\as_1\as_0,\bs_1\bs_0)$ using \Cref{samewitness}. Suppose instead that every pair in $G^*$ has a unique common witness with $(\as_1\as_0,\bs_1\bs_0)$.  
Since $G$ has a unique common inner witness $z$, there exists two distinct pairs $(u,v)$ and $(u',v')$ of length greater than $z$ such that $z$ is the unique common inner witness for the set $\{(u,v),(u',v')\}$. These pairs can be expressed as $(u,v) = (zr,rz)$ and $(u',v') = (zr',r'z)$, where $r \neq r'$. Consider the following three pairs formed using $(u,v)$ and $(u',v')$: 
\[(u_1,v_1) = (u'uu,v'vv) = (zr'zrzr,r'zrzrz),\]
\[(u_2,v_2) = (uu'u,vv'v) = (zrzr'zr,rzr'zrz),\] 
\[(u_3,v_3) = (uuu',vvv') = (zrzrzr',rzrzr'z).\] 
For $1 \leq i,j \leq 3$ with $i \neq j$, the following properties hold:
\begin{enumerate}
\item $|(u_i,v_i)| = |(u_j,v_j)|$.
\item $|(\rho_{u_i},\rho_{v_i})| = |(\rho_{u_j},\rho_{v_j})|$. Since $u_i$ and $u_j$ (resp. $v_i$ and $v_j$) are conjugates, their primitive roots are conjugates, and hence are of equal length, i.e., $|\rho_{u_i}| = |\rho_{u_j}|$ (resp.~$|\rho_{v_i}| = |\rho_{v_j}|$).
\item $(u_i,v_i) \neq (u_j,v_j)$. Otherwise, if $u_i= u_j$ (or $v_i=v_j$), it would follow that $r'zr=rzr'$ (or, $r'zrzr=rzrzr'$), implying that $r'$ is a outer witness of $(zr,rz) = (u,v)$. Since $r'$ is also an outer witness of $(u',v') = (zr',r'z)$, it contradicts the assumption that $z$ is a unique common witness of the set $\{(u,v),(u',v')\}$. Hence, $u_i\neq u_j$ and $v_i \neq v_j$.
\item $(\rho_{u_i},\rho_{v_i}) \neq (\rho_{u_j},\rho_{v_j})$. Since $u_i \neq u_j$ (resp. $v_i \neq v_j$) and $|(\rho_{u_i}| = |\rho_{u_j}|$ (resp. $|\rho_{v_i}| = |\rho_{v_j})|)$, we get that $\rho_{u_i} \neq \rho_{u_j}$ (resp. $\rho_{v_i} \neq \rho_{v_j})$.
\end{enumerate}

Let $(x_i,y_i)$ be the cut of $(\rho_{u_i},\rho_{v_i})$ for $1 \leq i \leq 3$. The set $\{(u_1,v_1),(u_2,v_2), (u_3,v_3)\}$ has a unique common witness, otherwise their primitive roots will be same by \Cref{infmanycwcorr} which is a contradiction. Since $z$ is an inner witness of each $(u_i,v_i)$ , $z$ is the unique common inner witnesss of $\{(u_1,v_1),(u_2,v_2), (u_3,v_3)\}$ , and thus of $\{(\rho_{u_1},\rho_{v_1}), (\rho_{u_2},\rho_{v_2}), (\rho_{u_2},\rho_{v_2})\}$ by \Cref{samewitnessext}. Since $(u_1,v_1)^*(u_2,v_2)^*(u_3,v_3)^* \subset G^*$ is conjugate and $|(\rho_{u_1},\rho_{v_1})| = |(\rho_{u_2},\rho_{v_2})| = |(\rho_{u_3},\rho_{v_3})| $, $z=x_1 =x_2 = x_3$ using \Cref{eqlengthprimitivescorr}.

By assumption, each pair $(u_i,v_i)$ has a unique common witness with $(\as_1\as_0,\bs_0\bs_1)$. Thus, at least two of these three pairs, say $(u_i,v_i)$ and $(u_j,v_j)$, either both have unique common inner witnesses with $(\as_1\as_0,\bs_1\bs_0)$, or both have unique common outer witnesses with $(\as_1\as_0,\bs_1\bs_0)$. Assume that the set $\{(u_i,v_i),(\as_1\as_0,\bs_1\bs_0)\}$ and $\{(u_j,v_j),(\as_1\as_0,\bs_1\bs_0)\}$ both has a unique common inner witness, say $z_i$ and $z_j$ respectively. Let $(\alpha,\beta)$ be the primitive root of $(\as_1\as_0,\bs_1\bs_0)$. There exists integers $m_i,m_j,n_i,n_j \in \N$ such that
\begin{equation}
z_i = (x_iy_i)^{m_i}x_i = (\alpha\beta)^{n_i}\alpha
\end{equation}
\begin{equation}
z_j = (x_jy_j)^{m_j}x_j = (\alpha\beta)^{n_j}\alpha
\end{equation}
Suppose both $m_i, m_j \geq 1$, then $x_iy_i$ and $x_jy_j$ are both equal length prefixes of $(\alpha\beta)^*\alpha$, and hence $x_iy_i = x_jy_j$, which contradicts that $\rho_{u_i} \neq \rho_{u_j}$. Hence, we get that either $m_i = 0$ or $m_j = 0$. WLOG, assume that $m_i=0$, i.e., $x_i = (\alpha\beta)^{n_i}\alpha$ and since $z = x_i = x_j$, we get $z \in (\alpha\beta)^*\alpha$, and hence is an inner witness of $(\as_1\as_0,\bs_1\bs_0)$. Infact, $z$ is the unique common inner witness of $(u_i,v_i)$ and $(\as_1\as_0,\bs_1\bs_0)$. Since $z$ (the unique common inner witness of $G$) is also an inner witness of $(\as_1\as_0,\bs_1\bs_0)$, we get that $z$ is a common inner witness of $G \cup \{(\as_1\as_0,\bs_1\bs_0)\}$.

The case where the set $\{(u_i,v_i),(\as_1\as_0,\bs_1\bs_0)\}$ and $\{(u_j,v_j),(\as_1\as_0,\bs_1\bs_0)\}$ both has a unique common outer witness is invalid and leads to a contradaction, since we get $y_i$ and $y_j$ are equal length prefixes of $(\beta\alpha)^*\beta$, and hence $y_i = y_j$. This contradicts $\rho_{u_i} \neq \rho_{u_j}$ since $x_i = x_j$.
\end{proof}

The remaining direction to prove is $(2) \implies (3)$ in \Cref{singlemonomial}. This direction states that if there exists a common witness for both $G$ and $(\as_1\as_0,\bs_1\bs_0)$ for a given set $M = (\as_0,\bs_0)G^*(\as_1,\bs_1)$, then $M$ has a common witness. It is a straightforward corollary of the below lemma.

\begin{lemma}\label{g1monomialwitness}
Let $M = (\as_0,\bs_0)G^*(\as_1,\bs_1)$ be a sumfree set with nonempty redux. If there exists a common witness $z'$ for $G \cup \{(\as_1\as_0,\bs_1\bs_0)\}$, then one of the following cases is true:
\begin{enumerate}[label=(\alph*)]
\item\label{Thm7:1} If $z'$ is a unique common inner witness, then $M$ has a unique common witness $z = [\as_0z',{\bs_0}]_R = [\as_1,z'\bs_1]_L$. Moreover, if $|\as_0z'| \geq |\bs_0|$ or equivalently $|\as_1| \leq |z'\bs_1|$, then $z$ is a common inner witness, otherwise it is a common outer witness.

\item\label{Thm7:2} If $z'$ is a unique common outer witness, then $M$ has a unique common witness $z = [\as_0,\bs_0z']_R = [z'\as_1,\bs_1]_L$. Moreover, if $ |z'\as_1| \geq |\bs_1|$ or equivalently $|\as_0| \leq |\bs_0z'|$, then $z$ is a common outer witness, otherwise it is a common inner witness.

\item\label{Thm7:3} If $G \cup \{(\as_1\as_0,\bs_1\bs_0)\}$ have infinitely many common witnesses, then $M$ is a subset of powers of the primitive root of its redux. Thus, $M$ has infinitely many witnesses.
\end{enumerate}
\end{lemma}

\begin{proof}
There are two cases to consider depending upon if $G \cup \{(\as_1\as_0,\bs_1\bs_0)\}$ has a common inner or outer witness.
\paragraph*{Case $(a)$: When $z'$ is a common inner witness of $G \cup \{(\as_1\as_0,\bs_1\bs_0)\}$} The following equations hold:
\begin{align}
\as_1\as_0 z' &= z' \bs_1\bs_0 \label{g1monomialwitness:eq1}\\
 u z' &= z' v \text{ ~for any pair } (u,v) \in G^*
\end{align}
We claim that $z = [\as_0z',\bs_0]_R = [\as_1,z'\bs_1]_L$ is a common witness for $M$. There are two cases depending upon whether $\bs_0$ is a suffix of $\as_0z'$ or vice-versa in \Cref{g1monomialwitness:eq1}.

\begin{enumerate}[label=(\alph*)]
\item When $\bs_0$ is a suffix of $\as_0z'$ or equivalently, when $\as_1$ is a prefix of $z'\bs_1$. We get $z = \as_0z' {\bs_0}^{-1}= \as_1^{-1}z'\bs_1$, or equivalently $\as_0z' =z\bs_0$ and $\as_1z = z'\bs_1$.
We show that $z$ is a common inner witness for $M$.

For any $(u,v) \in G^*$,
\begin{align*}
\as_0u\as_1 z &= \as_0uz'\bs_1 &&\text{(Substituting $\as_1z = z'\bs_1$)}\\
&=\as_0z'v\bs_1 &&\text{($z'$ is an inner witness of $(u,v)$)}\\
&=z\bs_0v\bs_1 &&\text{(Substituting $\as_0z' =z\bs_0$)}
\end{align*}

\item When $\as_0z'$ is a suffix of $\bs_0$ or equivalently $z'\bs_1$ is a prefix of $\as_1$. We get $z = \bs_0(\as_0z')^{-1} = (z'\bs_1)^{-1}\as_1$, or equivalently $\as_1=z'\bs_1z$ and $z\as_0z' = \bs_0$. We show that $z$ is a common outer witness for $M$.

For any $(u,v) \in G^*$,
\begin{align*}
z\as_0u\as_1 &= z\as_0uz'\bs_1z &&\text{(Substituting $\as_1=z'\bs_1z$)}\\
&= z\as_0z'v\bs_1z &&\text{($z'$ is an inner witness of $(u,v)$)}\\
&= \bs_0v\bs_1z &&\text{(Substituting $z\as_0z' = \bs_0$)}
\end{align*}
\end{enumerate}

\medskip
\paragraph*{Case $(b)$: When $z'$ is a common outer witness of $G$ and $(\as_1\as_0,\bs_1\bs_0)$} Therefore, the following equations hold:
\begin{align}
z' \as_1\as_0 &= \bs_1\bs_0 z' \label{g1monomialwitness:eq2}\\
z' u &= v z' \text{ ~for any pair } (u,v) \in G^*
\end{align}
We claim that $z = [\as_0,\bs_0z']_R = [z'\as_1,\bs_1]_L$ is a witness for $M$. There are two cases depending upon if $\as_0$ is a suffix of $\bs_0z'$ or vice-versa in \Cref{g1monomialwitness:eq2}.

\begin{enumerate}[label=(\alph*)]
\item When $\as_0$ is a suffix of $\bs_0z'$ or equivalently, $\bs_1$ is a prefix of $z'\as_1$. We get $z = \bs_0z' \as_0^{-1} = \bs_1^{-1}z'\as_1$, or equivalently $z\as_0=\bs_0z'$ and $z'\as_1 =\bs_1z$.
We show that $z$ is a common outer witness for $M$.

For any $(u,v) \in G^*$,
\begin{align*}
z\as_0u\as_1 &= \bs_0z'u\as_1 &&\text{(Substituting $z\as_0=\bs_0z'$)}\\
&=\bs_0vz'\as_1 &&\text{($z'$ is an outer witness of $(u,v)$)}\\
&=\bs_0v\bs_1z &&\text{(Substituting $z'\as_1 =\bs_1z$)}
\end{align*}

\item If $\bs_0z'$ is a suffix of $\as_0$ or equivalently, $z'\as_1$ is a prefix of $\bs_1$. Therefore, $z = \as_0(\bs_0z')^{-1} = (z'\as_1)^{-1}\bs_1$, or equivalently $\as_0=z\bs_0z'$ and $z'\as_1z=\bs_1$.
We show that $z$ is a common inner witness for $M$.

For any $(u,v) \in G^*$,
\begin{align*}
\as_0u\as_1z &= z\bs_0z'u\as_1z &&\text{(Substituting $\as_0=z\bs_0z'$)}\\
&=z\bs_0vz'\as_1z &&\text{($z'$ is an outer witness of $(u,v)$)}\\
&=z\bs_0v\bs_1 &&\text{(Substituting $z'\as_1z=\bs_1$)}
\end{align*}
\end{enumerate}
Therefore, if there exists a common witness $z'$ for $G$ and $(\as_1\as_0,\bs_1\bs_0)$, then there also exists a common witness $z$ for $M$.
Furthermore, since the redux remains constant in the word equation that relates the common witness of $M$ and those of $G \cup \{(\as_1\as_0,\bs_1\bs_0)\}$ , each distinct common witness of $M$ corresponds to a distinct common witness of $G \cup \{(\as_1\as_0,\bs_1\bs_0)\}$, and vice-versa. Consequently, if $z'$ is the unique common witness for $G \cup \{(\as_1\as_0,\bs_1\bs_0)\}$ then $z$ is the unique common witness of $M$, and conversely.

Suppose the set $G \cup \{(\as_1\as_0,\bs_1\bs_0)\}$ has infinitely many witnesses.
According to \Cref{infmanycwcorr}, all the pairs in that set have the same primitive root, let us say $(\rho,\rho^\prime)$. Therefore, $G^*(\as_1,\bs_1)(\as_0,\bs_0)$ is a subset of powers of $(\rho,\rho^\prime)$ and is conjugate. Since $M$ is a cyclic shift of $G^*(\as_1,\bs_1)(\as_0,\bs_0)$ and is also conjugate, it is also a subset of powers of a primitive root, let us say $(\rho_m,\rho_m^\prime)$, that is a cyclic shift of $(\rho,\rho^\prime)$. Moreover, $\as_1$ (\textit{resp.}~$\bs_1$) is an inner (\textit{resp.}~outer) witness of $(\rho,\rho_m)$ (\textit{resp.}~$(\rho^\prime,\rho_m^\prime)$). We observe that $(\rho_m,\rho_m^\prime)$ is the primitive root of the redux of $M$. Hence, $M$ is a subset of powers of the primitive root of its redux.
\end{proof}

\subsection{Common Witness of a Sumfree Set}
We prove \Cref{generalmonomialwitness} in this subsection.

%\hl{We prove $(1) \iff (2)$ below, and $(1) \iff (3)$ (by common witness theorem of Kleene closure)}We prove $(1) \implies (2), (3) \implies (1)$ and $(2) \iff (3)$ in \Cref{generalmonomialwitness}. $(3) \implies (1)$ is obvious. We show $(2) \iff (3)$ first.
\begin{lemma}\label{mcommonwitness}
Given a sumfree set $M = (\as_0,\bs_0){G_1}^*(\as_1,\bs_1) \cdots (\as_{k-1},\bs_{k-1}){G_k}^*(\as_k,\bs_k)$. The following are equivalent.
\begin{enumerate}
\item $z$ is a common inner (\textit{resp.}~outer) witness of $M$.
\item $z$ is a common inner (\textit{resp.}~outer) witness of each of its singleton redux.
\end{enumerate}
\end{lemma}

\begin{proof}
Let  $M = (\as_0,\bs_0){G_1}^*(\as_1,\bs_1) \cdots (\as_{k-1},\bs_{k-1}){G_k}^*(\as_k,\bs_k)$ be a sumfree set.  For $i \in \{1, \ldots, k\}$, let $M_i$ denote the singleton redux of $M$ keeping only the Kleene star $G_i^*$, i.e., \[M_i = (\as_0 \cdots \as_{i-1},\bs_0 \cdots \bs_{i-1})G_i^*(\as_i \cdots \as_k,\bs_i \cdots \bs_k)\] The proof of $(1) \implies (2)$ is trivial, i.e., if $z$ is a common witness of $M$, then $z$ is also a common witness of $M_i \subseteq M$.

We prove $(2) \implies (1)$. WLOG, assume that $z$ is a common \emph{inner} witness of each $M_i$'s. If the redux is an empty pair $(\varepsilon,\varepsilon)$, then $M = {G_1}^*{G_2}^*\cdots {G_k}^*$. It is easy to prove that $z$ is a common inner witness of $M$ by showing that for any arbitrary pair $(u_i,v_i) \in G_i^*$, $u_1 u_2 \cdots u_k z = z v_1v_2 \cdots v_k$. Proof by induction on $k$.
\begin{align*}
u_1 \cdots u_{k-1}u_k z &= u_1 \cdots u_{k-1}z v_k &&\text{(Since $u_kz = zv_k$)}\\
&=z v_1 \cdots v_{k-1} v_k &&\text{(By Induction Hypothesis)}
\end{align*}

Now, assume that the redux is nonempty. Since each $M_i$ has a common witness, there exist a common witness, say $z_i$, for $G_i \cup \{(\as_i \cdots \as_k\as_0\cdots \as_{i-1},\bs_i \cdots \bs_k\bs_0\cdots \bs_{i-1})\}$ by \Cref{singlemonomial}. There are 3 possible cases for $z_i$.
\begin{enumerate}[label=(\alph*)]
\item $z_i$ is a unique common inner witness. Therefore, for any pair $(u_i,v_i) \in G_i^*$,
\begin{align}
u_i z_i &= z_i v_i\\
\as_i \cdots \as_k\as_0\cdots \as_{i-1}z_i &= z_i\bs_i \cdots \bs_k\bs_0\cdots \bs_{i-1}\\
z=\as_0 \cdots \as_{i-1}z_i(\bs_0 \cdots \bs_{i-1})^{-1} &= (\as_i \cdots \as_k)^{-1}z_i\bs_i \cdots \bs_k \text{ (By \Cref{g1monomialwitness} \ref{Thm7:1})} \label{z_m i.w}
\end{align}

\item $z_i$ is a unique common outer witness. Therefore, for any pair $(u_i,v_i) \in G_i^*$,
\begin{align}
z_i u_i &= v_i z_i\\
z_i\as_i \cdots \as_k\as_0\cdots \as_{i-1} &= \bs_i \cdots \bs_k\bs_0\cdots \bs_{i-1}z_i\\
z = \as_0\cdots \as_{i-1}(\bs_0 \cdots \bs_{i-1}z_i)^{-1} &= (z_i\as_i \cdots \as_k)^{-1}\bs_i \cdots \bs_k \text{ (By \Cref{g1monomialwitness} \ref{Thm7:2})} \label{z_m o.w}
\end{align}

\item When $G_i \cup \{(\as_i \cdots \as_k\as_0\cdots \as_{i-1},\bs_i \cdots \bs_k\bs_0\cdots \bs_{i-1})\}$ has infinitely many witnesses, by \Cref{infmanycwcorr}, it is a set of powers of the same primitive root say $(\rho_i,\rho_i^\prime)$. Therefore $z_i$ belongs to witnesses of $(\rho_i,\rho_i^\prime)$. From \Cref{g1monomialwitness} \ref{Thm7:3}, the set $M_i$ reduces to a set of powers of the primitive root of the redux $(\as_0 \cdots \as_k,\bs_0 \cdots \bs_k)$, say $(\rho,\rho^\prime)$. Note that $\as_i \cdots \as_k$ (\textit{resp.}~$\as_0 \cdots \as_{i-1}$) is an inner (\textit{resp.}~outer) witness of $(\rho_i,\rho)$. Similarly $\bs_i \cdots \bs_k$ (\textit{resp.}~$\bs_0 \cdots \bs_{i-1}$) is an inner (\textit{resp.}~outer) witness of $(\rho_i^\prime,\rho^\prime)$.
\end{enumerate}

We show that $z$ is a common inner witness of $M$, i.e., for any arbitrary pair $(u_i,v _i) \in G_i^*$ (possibly empty), we prove $\as_0u_1\as_1u_2\as_2 \cdots \as_{k-1}u_k\as_kz =z\bs_0v_1\bs_1v_2\bs_2 \cdots \bs_{k-1}v_k\bs_k$. The proof is by induction on the number of singleton reduxes, $0 \leq i \leq k$.

\medskip
\paragraph*{Base Case}
When $i=0$, it is vacuously true since $z$ is a witness of the redux.

\medskip
\paragraph*{Inductive Case}
Assume for induction that it is true for the first $i-1$ singleton reduxes, i.e., \[\as_0u_1\as_1u_2\as_2 \cdots u_{i-1}\as_{i-1} \cdots \as_{k}z =z\bs_0v_1\bs_1v_2\bs_2 \cdots v_{i-1}\bs_{i-1}\cdots\bs_{k} \ .\]
We prove it for the first $i$ singleton reduxes, i.e., we show \[\as_0u_1\as_1u_2\as_2 \cdots u_{i-1}\as_{i-1}u_{i}\as_{i} \cdots \as_{k}z =z\bs_0v_1\bs_1v_2\bs_2 \cdots v_{i-1}\bs_{i-1}v_{i}\bs_{i}\cdots\bs_{k} \ .\] There are 3 possible cases for the common witness $z_i$ of $G_i \cup \{(\as_i \cdots \as_k\as_0\cdots \as_{i-1},\bs_i \cdots \bs_k\bs_0\cdots \bs_{i-1})\}$.
\begin{enumerate}
\item When $z_i$ is a unique common inner witness. From \Cref{z_m i.w}, $z = (\as_i \cdots \as_k)^{-1}z_i\bs_i \cdots \bs_k$.
\begin{align*}
\as_0&u_1\as_1 \cdots u_{i-1}\as_{i-1}u_i \as_i \as_{i+1} \cdots \as_k z\\
&= \as_0 u_1\as_1 \cdots u_{i-1}\as_{i-1}u_iz_i\bs_i \cdots \bs_k &&\text{ (Subs. $z$)}\\
&= \as_0 u_1\as_1 \cdots u_{i-1}\as_{i-1}z_iv_i\bs_i \cdots \bs_k &&\text{ ($u_iz_i = z_iv_i$)}\\
&= \as_0 u_1\as_1 \cdots u_{i-1}\as_{i-1}\as_i \cdots \as_kz(\bs_i \cdots \bs_k)^{-1}v_i \bs_i \cdots \bs_k &&\text{ (Subs. $z_i$)}\\
&= z\bs_0v_1\bs_1 \cdots v_{i-1}\bs_{i-1}\bs_i \cdots \bs_k (\bs_i \cdots \bs_k)^{-1}v_i \bs_i \cdots \bs_k &&\text{ (Inductive Hypothesis)}\\
&= z\bs_0v_1\bs_1 \cdots v_{i-1}\bs_{i-1}v_i \bs_i \cdots \bs_k &&\text{ (Simplifying)}
\end{align*}

\item When $z_i$ is a unique outer witness. By \Cref{z_m o.w}, $z = (z_i\as_i \cdots \as_k)^{-1}\bs_i \cdots \bs_k$.
\begin{align*}
\as_0 &u_1\as_1 \cdots u_{i-1}\as_{i-1}u_i \as_i \as_{i+1} \cdots \as_k z\\
&= \as_0 u_1\as_1 \cdots u_{i-1}\as_{i-1}u_i \as_i \as_{i+1} \cdots \as_k(z_i\as_i \cdots \as_k)^{-1}\bs_i \cdots \bs_k &&\text{(Subs. $z$)}\\
&= \as_0 u_1\as_1 \cdots u_{i-1}\as_{i-1}u_i z_i^{-1}\bs_i \cdots \bs_k &&\text{(Simplifying)}\\
&= \as_0 u_1\as_1 \cdots u_{i-1}\as_{i-1}{z_i}^{-1}v_i\bs_i \cdots \bs_k &&\text{($u_i = {z_i}^{-1}v_iz_i$)}\\
&= \as_0 u_1\as_1 \cdots u_{i-1}\as_{i-1}\as_i \cdots \as_kz(\bs_i \cdots \bs_k)^{-1}v_i \bs_i \cdots \bs_k &&\text{(Subs. $z_i^{-1}$)}\\
&= z\bs_0v_1\bs_1 \cdots v_{i-1}\bs_{i-1}\bs_i \cdots \bs_k (\bs_i \cdots \bs_k)^{-1}v_i \bs_i \cdots \bs_k &&\text{(Inductive Hypothesis)}\\
&= z\bs_0v_1\bs_1 \cdots v_{i-1}\bs_{i-1}v_i \bs_i \cdots \bs_k &&\text{(Simplifying)}
\end{align*}

\item When $z_i$ is a witness of the primitive root $(\rho_i,\rho_i^\prime)$ of $G_i \cup \{(\as_i \cdots \as_k\as_0\cdots \as_{i-1},\bs_i \cdots \bs_k\bs_0\cdots \bs_{i-1})\}$ (The case where $G_i \cup \{(\as_i \cdots \as_k\as_0\cdots \as_{i-1},\bs_i \cdots \bs_k\bs_0\cdots \bs_{i-1})\}$ have infinitely many witnesses). Here $(u_i,v_i)$ is some $m^\mathit{th}$ power of $(\rho_i,\rho_i^\prime)$. Since $z$ is a witness of a singleton redux, it is also a witness of the redux $(\as_0 \cdots \as_k,\bs_0 \cdots \bs_k)$, and hence a witness of its primitive root $(\rho,\rho^\prime)$. We also know that $\as_i \cdots \as_k$ is an inner witness of $(\rho_i,\rho)$ and $\bs_i \cdots \bs_k$ is an inner witness of $(\rho_i^\prime,\rho^\prime)$. 
\begin{align*}
\as_0 &u_1\as_1 \cdots u_{i-1}\as_{i-1}u_i \as_i \as_{i+1} \cdots \as_k z\\
&= \as_0 u_1\as_1 \cdots u_{i-1}\as_{i-1}(\rho_i)^{m} \as_i \as_{i+1} \cdots \as_k z &&\text{(Subs. $u_i$)}\\
&= \as_0 u_1\as_1 \cdots u_{i-1}\as_{i-1} \as_i \as_{i+1} \cdots \as_k(\rho)^{m} z &&\text{($\as_i \cdots \as_k$ is an i.w. of $(\rho_i,\rho)$)}\\
&= \as_0 u_1\as_1 \cdots u_{i-1}\as_{i-1} \as_i \as_{i+1} \cdots \as_kz (\rho^\prime)^{m} &&\text{($z$ is a witness of $(\rho,\rho^\prime)$)}\\
&= \as_0 u_1\as_1 \cdots u_{i-1}\as_{i-1} \as_i \as_{i+1} \cdots \as_kz (\bs_i \cdots \bs_k)^{-1}({\rho_i}^\prime)^{m}\bs_i \cdots \bs_k &&\text{($\bs_i \cdots \bs_k$ is an i.w. of $(\rho_i^\prime,\rho^\prime)$)}\\
&= \as_0 u_1\as_1 \cdots u_{i-1}\as_{i-1} \as_i \as_{i+1} \cdots \as_kz (\bs_i \cdots \bs_k)^{-1}v_i\bs_i \cdots \bs_k &&\text{(Subs. $v_i = (\rho_i^\prime)^{m}$)}\\
&= z\bs_0v_1\bs_1 \cdots v_{i-1}\bs_{i-1}\bs_i \cdots \bs_k (\bs_i \cdots \bs_k)^{-1}v_i \bs_i \cdots \bs_k &&\text{(Inductive Hypothesis)}\\
&= z\bs_0v_1\bs_1 \cdots v_{i-1}\bs_{i-1}v_i \bs_i \cdots \bs_k &&\text{(Simplifying)}
\end{align*}
\end{enumerate}
Thus $z$ is a common inner witness of $M$.
\end{proof}

\section{Conjugacy of Monoid Closures}
\label{sec:conjugateMonoid}

The entire section is devoted to the proof of \Cref{thm:conjugateMonoidClosure}.

The proof of $(\Longleftarrow)$ direction of \Cref{thm:conjugateMonoidClosure} is straightforward. For Item (1), if $M$ has a common witness, then there exist a witness for every pair of words produced by $M$ and hence $M$ is conjugate. For Item (2), since each $G_i^*$ is conjugate for $i \in [1,k]$ and $\as_0 \cdots \as_k \sim \bs_0 \cdots \bs_k$, we get that $M$ produces pairs of words of equal length. Moreover, each pair of words in $M$ is a cyclic shift of $v\rho^i$ for some $i \geq 0$, and hence it is conjugate.

    We prove the $(\Longrightarrow)$ direction in \Cref{thm:conjugateMonoidClosure} in the rest of the section.
    The case when $M$ is of the form $(\as_0,\bs_0)G_1^*(\as_1,\bs_1)$  holds already from \Cref{singlemonomial}. We now assume that $k \geq 2$. 
    
    Since $M$ is conjugate, each $G_i^*$ for $i \in [1,k]$ is conjugate by \Cref{boundedconj}. Using \Cref{conjugacy}, we get that each $G_i$ has a common witness. From \Cref{infmanycwcorr}, either $|\calW(G_i^*)|=\infty$ or $|\calW(G_i^*)|=1$.

    First, we show that when the redux of $M$ is empty, $M$ has a common witness. If the redux of $M$ is empty, then each singleton redux of $M$ is of the form $G_i^*, i \in [1,k]$ and $M= G_1^* G_2^* \cdots G_k^*$.  We define $P= \bigcup_{i=1}^{k}P_i$ where 
    \begin{itemize}
        \item $P_i = \{(\rho_i^{n_i},\rho_i'^{n_i})\}$ if $|\calW(G_i^*)|=\infty$, where $\{(\rho_i,\rho_i')\}=\roots{G_i^*}$ and $n_i>0$ is such that $(\rho_i^{n_i},\rho_i'^{n_i}) \in G_i^*$, and 
        \item $P_i= \{(u_{1},v_{1}), (u_{2},v_{2})\}$ if $|\calW(G_i^*)|=1$, where $(u_1,v_1),(u_2,v_2)\in G_i^*$ are such that $\calW(\{(u_1,v_1),(u_2,v_2)\}) = \calW(G_i^*)$. 
    \end{itemize}

    Observe that  
    \begin{align}
        \prod_{i=1}^{k}\left(\prod_{p \in P_i}p^*\right ) \subseteq M
    \end{align}
    is conjugate. By \Cref{commonwitnessfinite},  the set $P$ has a common witness $z$.  Therefore 
    $z$ is a common witness for each of the $G_i^*$ by construction and \Cref{samewitness}. Therefore, $z$ is a common witness for each of the singleton redux namely $G_i^*$ for $i\in [1,k]$. By virtue of \Cref{mcommonwitness}, $M$ has a common witness. Thus if redux of $M$ is empty $M$ has a common witness.

    Henceforth we assume that the redux of $M$ is nonempty. Let
    \begin{align}
        C=|\as_0\as_1 \cdots \as_k| = |\bs_0\bs_1 \cdots \bs_k|.
    \end{align}
    
    \begin{definition}[Canonical Pair]\label{def:diffpumpedcanonical}
    A pair $(u',v') \in M$ is \emph{canonical} if it is of the form:
    \begin{align}(u, v) &= (\as_0,\bs_0)(u_1,v_1)^{j_1} (\as_1,\bs_1)(u_2,v_2)^{j_2} \cdots (u_k,v_k)^{j_k}(\as_k,\bs_k) 
    \end{align}
where $(u_i,v_i) \in G_i^*$ and $j_i \in \N$, for $i\in[1,k]$, satisfies the following conditions.
    \begin{enumerate} 
    \item Let $i\in [1,k]$. Then,
        \begin{enumerate}
    \item $(u_i,v_i) = (\rho_i^n, \rho_i'^n)$   if $|W(G_i^*)|=\infty$, where $\roots{G_i} = \{(\rho_i, \rho_i')\}$ and $n \geq 1$ is such that $(\rho_i^n, \rho_i'^n) \in G_i^*$, and 

        \item $(u_i,v_i) = (u,v) \in G_i^*$ if $|W(G_i^*)|=1$, where  $(u,v) \in G_i^*$ is such that 
        \begin{align}
           W(\{(u,v) , (\as_i \cdots \as_k \as_1 \cdots \as_{i-1}, \bs_i \cdots \bs_k \bs_1 \cdots \bs_{i-1})\})&=W(G_i^*)
        \end{align}
        guaranteed by  \Cref{gmonomialcut}.
    \end{enumerate}

        \item If $|W(G_i^*)| =1$ for some $i\in [1,k]$, then  $\rho_{u_j} \not\sim \rho_{u_l}$ for some $j,l \in [1,k]$.
    \end{enumerate}    
    \end{definition}

    We next establish the existence of canonical pairs. 
\begin{lemma}\label{lemma:item2existDiffpump}
Let $\rho$ be a primitive word. If $|W(G_i^*)| = 1$,  then there exists a pair $(u_i, v_i) \in G_i^*$ such that $\rho_{u_i} \not\sim \rho$ and $W(\{(u_i,v_i), (\as_i \cdots \as_k \as_0 \cdots \as_{i-1}, \bs_i \cdots \bs_k \bs_0 \cdots \bs_{i-1})\})=W(G_i^*)$.
\end{lemma}
\begin{proof}

Assume that $G_i^*$ has a unique common witness and let $\rho$ be an arbitrary primitive word. 

Let $p=(\as_i \cdots \as_k \as_0 \cdots \as_{i-1},\; \bs_i \cdots \bs_k \bs_0 \cdots \bs_{i-1})$.
Since $|W(G_i^*)|=1$, there exist at least two pairs $(u_1,v_1), (u_2,v_2) \in G_i^*$ such that $\rho_{u_1} \neq \rho_{u_2}$.
Hence, by \Cref{prop:comb}, we can construct three pairs in $G_i^*$ that have pairwise distinct primitive roots. Therefore, there exist two pairs $(u_1,v_1), (u_2,v_2) \in G_i^*$ that have primitive roots different from that of $p$. Therefore any two pairs from the set $\{(u_1,v_1), (u_2,v_2), p\}$ has a unique common witness by \Cref{1-infinity}. Moreover,
by \Cref{gmonomialcut}, 
\begin{align}
W(\{(u_1,v_1), p\})=W(
\{(u_2,v_2),p\})=W(G_i^*).
\end{align}  If either $\rho_{u_1} \not\sim \rho$ or $\rho_{u_2} \not\sim \rho$, then we choose the corresponding pair as $(u_i,v_i)$, and hence the claim follows. 
Otherwise, assume that $\rho_{u_1} \sim \rho\sim \rho_{u_2}$. 

Next consider the words $(u_3,v_3) = (u_1u_2, v_1v_2)$ and $(u_4,v_4) = (u_3u_1, v_3v_1)$.
By \Cref{prop:nontriv-concat2}, we deduce that $\rho_{u_3} \not \sim \rho$ and 
$\rho_{u_4} \not \sim \rho$.
If $\rho_{u_3}$ differs from the primitive root of $\fst(p)$,
then $W(\{(u_3,v_3),p\} = W(G_i^*)$ and $\rho_{u_3} \not \sim \rho$. Hence the word $(u_i,v_i) = (u_3,v_3)$ satisfies the claim. 
Otherwise, $\rho_{u_3}$ is the primitive root of $\fst(p)$. By \Cref{prop:comb}, $\rho_{u_3} \not \in \{\rho_{u_1},\rho_{u_2}\}$ since $\rho_{u_1}\neq \rho_{u_2}$.  
By applying \Cref{prop:comb} again, we get $\rho_{u_4} \neq \rho_{u_3}$. Hence, the primitive roots of $u_4$ and $\fst(p)$ differ. Thus, $(u_i,v_i) = (u_4,v_4)$ satisfies the claim.
\end{proof}
\begin{corollary}
    If $M$ is conjugate then there exists a canonical pair in $M$.
\end{corollary}
\begin{proof}
      We construct a canonical pair $(u,v)$ in $M$ by choosing, for each $i \in [1,k]$, a pair $(u_i,v_i) \in G_i^*$ satisfying Item~$1$ of \Cref{def:diffpumpedcanonical}. Item~$1.(a)$ holds trivially, while Item~$1.(b)$ follows from \Cref{gmonomialcut}. It remains to verify Item~$2$ in \Cref{def:diffpumpedcanonical}. If $|W(G_i^*)|=\infty$ for every $i \in [1,k]$, then Item~$2$ is vacuously satisfied. Otherwise, suppose that there exists an index $i \in [1,k]$ such that $|W(G_i^*)|=1$. 
      If there exists $j\neq i \in [1,k]$such that $\rho_{u_i} \not\sim \rho_{u_j}$, then we are done. Otherwise, using 
\Cref{lemma:item2existDiffpump} we choose a pair $(u_i',v_i')\in G_i^*$ that satisfies $\rho_{u_i'} \not\sim \rho_{u_i}$ and   Item~$1.(b)$.
\end{proof}
        
    \begin{definition}[Differentially-Pumped Pair]\label{def:diffpumped}
    
A \emph{Differentially-Pumped Pair} is a canonical pair of the form
    \begin{align}(u, v) &= (\as_0,\bs_0)(u_1,v_1)^{j_1} (\as_1,\bs_1)(u_2,v_2)^{j_2} \cdots (u_k,v_k)^{j_k}(\as_k,\bs_k) 
    \end{align}
    where $j_1 > 2k \cdot \max(C,\ell)$ where 
$\ell = 2\cdot \max\{|u_i| \mid i \in [1,k]\}$, and $j_2,\ldots, j_k$ obeys

\begin{equation}\label{eq:diffpump}
        2 \cdot (|u_1^{j_1}| + |u_2^{j_2}| + \cdots + |u_{i-1}^{j_{i-1}}| + C) < |u_i^{j_i}|.
    \end{equation}
   \end{definition}
    
     Observe that $\ell > |u_i| + |u_j| - \mathit{gcd}(|u_i|,|u_j|) > 0$ (Fine~\&~Wilf index between $u_i$ and $u_j$) for $i,j \in [1,k]$.

   Since $M$ is conjugate, a differentialy-pumped pair $(u,v) \in M$ is conjugate and has a cut, say $(p,q)$. Towards proving \Cref{thm:conjugateMonoidClosure}, we do a case analysis on the cut depending whether the cut $(p,q)$ is short or long (defined below).

\begin{definition}[Short \& Long Cuts]
Let $(u, v) \in M$ be the differentially pumped pair from \Cref{def:diffpumped}. Let $(p,q)$ be a cut of $(u,v)$. We say $p$ is \emph{short} if the prefix $p\bs_0$ of $u$ is a prefix of $\as_0u_1^{j_1-1}$ as shown below.

\begin{align*}
  u &= \overbrace{\as_0\,u_1 \cdots }^{p\bs_0} \overbrace{\cdots u_1 \, \as_1\, u_2\cdots u_2 \, \as_2 \cdots u_k \cdots u_k \,\as_k}^{\bs_0^{-1}q}\\
v &= \underbrace{\bs_0\,v_1 \cdots \cdots v_1 \, \bs_1 \,v_2 \cdots v_2 \, \bs_2 \, \cdots v_k \cdots}_{q} \underbrace{\cdots v_k \, \bs_k}_{p}
\end{align*}

Similarly, $q$ is \emph{short} if the prefix $q\as_0$ of $v$ is a prefix of $\bs_0v_1^{j_1-1}$.
The cut $(p,q)$ is \emph{short} if either $p$ or $q$ is short. When $p,q$ or the cut $(p,q)$ is not short, they are \emph{long}.
    
\end{definition}

   \begin{definition}[Blocks of a Canonical Pair]

   Given a differentially-pumped pair from Definition,
    $B_1, B_2, \ldots, B_k$ represent the factors of $u_1^{j_1}$s, $u_2^{j_2}$, $\ldots u_k^{j_k}$ respectively in $u$. Likewise  $B_1^\prime, B_2^\prime, \ldots, B_k^\prime$ represent the factors $v_1^{j_1}$, $v_2^{j_2}$, $\ldots v_k^{j_k}$ respectively in $v$.

    \begin{align*}
  u &= \as_0\,\overbrace{u_1 \cdots u_1}^{j_1 \text{ times }}\as_1\overbrace{u_2 \cdots u_2}^{j_2 \text{ times }}\as_2 \cdots \overbrace{u_k \cdots u_k}^{j_k \text{ times }}\as_k &&= \as_0\,\overbrace{x_1y_1 \cdots x_1y_1}^{B_1} \, \as_1\, \overbrace{x_2y_2 \cdots x_2y_2}^{B_2} \, \as_2 \cdots \overbrace{x_ky_k \cdots x_ky_k}^{B_k} \,\as_k\\
v &= \bs_0\,v_1 \cdots v_1 \, \bs_1 \,v_2 \cdots v_2 \, \bs_2 \, \cdots v_k \cdots v_k \, \bs_k &&= \bs_0\,y_1x_1 \cdots y_1x_1 \, \bs_1 \,y_2x_2 \cdots y_2x_2 \, \bs_2 \, \cdots y_kx_k \cdots y_kx_k \, \bs_k
\end{align*}
   \end{definition}
 Also, let $(x_i,y_i)$ denote a cut of $(\rho_{u_i},\rho_{v_i})$ for each $i \in [1,k]$.

\medskip
\noindent

In the next two subsections, we prove that whenever a differentially-pumped pair has a short cut, $M$ has a common witness (See \Cref{prop:short}). Otherwise, if it has a long cut then we show that there exist words $v, \rho$ such that $\fst(G_i^*) \subsetsim \rho^*$ for all $i \in [1,k]$, and $\fst(M) \sim \snd(M) \subsetsim v\rho^*$ where either $v= \epsilon$ or $v \sim \as_0 \cdots \as_k$ (See \Cref{prop:long}). As a consequence, \Cref{thm:conjugateMonoidClosure} follows.

\subsection{Short cut for Differentially-pumped pair}

In this subsection, we prove that whenever a differentially-pumped pair has a short cut, then $M$ has a common witness.
\begin{proposition}\label{prop:short}
    Assume that $M$ is conjugate and let $(u,v) \in M$ be a differentially-pumped pair from \Cref{def:diffpumped}. If $(u,v)$ has a short cut then $M$ has a common witness.
\end{proposition}

First, we prove some preparatory lemmas.
     
\begin{lemma}\label{claimblock}
Assume that $(u,v)$ from \Cref{def:diffpumped} has a short cut $(p,q)$ such that $p$ is short. Then the following conditions hold.
\begin{enumerate}
    \item Let $w_1 = [\as_0,p\bs_0]_L$. If $\as_0$ is a prefix of $p\bs_0$, then 
$w_1$ is an inner witness of $(u_1,v_1)$, and $B_1=w_1q_1$ and $B_1'=q_1w_1$ for some word $q_1$. Otherwise, $p\bs_0$ is a prefix of $\as_0$ and  $w_1$ is an outer witness of $(u_1,v_1)$, and $B_1 = q_1w_1$ and $B_1' = w_1q_1$ for some word $q_1$.
\item For $i\in [2,k]$, 
\begin{enumerate}
    \item if $w_{i-1}$ is an inner witness of $(u_{i-1},v_{i-1})$, then  $w_i = [\as_{i-1},w_{i-1}\bs_{i-1}]_L$ satisfies the following:
           If $\as_{i-1}$ is a prefix of $w_{i-1}\bs_{i-1}$, then $w_i$ is an inner witness of $(u_i,v_i)$, and
$B_i = w_iq_i$ and $B_i' = q_i w_i$ for some word $q_i$.
              Otherwise, $w_{i-1}\bs_{i-1}$ is a prefix of $\as_{i-1}$ and $w_i$ is an outer witness of $(u_i,v_i)$, and  $B_i = q_iw_i$ and $B_i' = w_i q_i$ for some word $q_i$.

              \item  If $w_{i-1}$ is an outer witness of $(u_{i-1},v_{i-1})$, then $w_i = [w_{i-1}\as_{i-1},\bs_{i-1}]_L$ satisfies the following: 
          If $w_{i-1}\as_{i-1}$ is a prefix of $\bs_{i-1}$, then $w_i$ is an inner witness of $(u_i,v_i)$, and $B_i = w_i q_i$ and $B_i' = q_i w_i$ for some word $q_i$.
           Otherwise, $\bs_{i-1}$ is a prefix of $w_{i-1}\as_{i-1}$ and $w_i$ is an outer witness of $(u_i,v_i)$, and   $B_i = q_iw_i$ and $B_i' = w_i q_i$ for some word $q_i$.
\end{enumerate}
\item  If $w_k$ is an inner witness of $(u_k,v_k)$, then $w_1=[\as_{k}\as_0,w_{k}\bs_{k}\bs_0]_L$. Moreover, if $w_1$ is an inner witness $(u_1,v_1)$, then $w_1 = (\as_{k}\as_0)^{-1}w_{k}\bs_{k}\bs_0$, and $w_1 = (w_{k}\bs_{k}\bs_0)^{-1}\as_{k}\as_0$ otherwise. If $w_k$ is an outer witness of $(u_k,v_k)$, then $w_1 = [w_k\as_{k}\as_0,\bs_{k}\bs_0]_L$. Moreover, if $w_1$ is an inner witness $(u_1,v_1)$, then $w_1 = (w_k\as_{k}\as_0)^{-1}\bs_{k}\bs_0$, and $w_1 = (\bs_{k}\bs_0)^{-1}w_k\as_{k}\as_0$ otherwise.
\end{enumerate}

\end{lemma}     

\begin{proof}
 \noindent 1.\quad    
     Let $w_1 = [\as_0,p\bs_0]_L$. There are two cases: either $\as_0$ is a prefix of $p\bs_0$, or $p\bs_0$ is a prefix of $\as_0$. In either case, we show that $w_1$ is a witness of $(u_1,v_1)$.

    \begin{itemize}
        \item \emph{$\as_0$ is a prefix of $p\bs_0$}: We show that $w_1$ is an inner witness of $(u_1,v_1)$, and there exists a word $q_1$ such that $B_1 = w_1q_1$ and $B_1' = q_1w_1$.

          Since $\as_0$ is a prefix of $p\bs_0$ and $p$ is short, $w_1 = \as_0^{-1}p\bs_0$ is a prefix of $B_1(u_1)^{-1}$ in $u$. We match  $(p \bs_0)^{-1}u$ with $\bs_0^{-1}vp^{-1}$. Since $\bs_0^{-1}vp^{-1}$  starts with $y_1x_1$ and since there is at least one factor $x_1y_1$ following $p\bs_0$, either $w_1 \in (x_1y_1)^*x_1$ (by Case~\ref{distinct} of \nameref{cutcorr}) or $w_1 \in (x_1y_1)^*$ and $x_1y_1 = y_1x_1$ (by \ref{same} of \nameref{cutcorr}). In the former case, $w_1$ is an inner witness of $(x_1y_1,y_1x_1) = (\rho_{u_1},\rho_{v_1})$ and hence is also an inner witness of $(u_1,v_1)$ by \Cref{samewitness}. Similarly, when $w_1 \in (x_1y_1)^*$ and $x_1y_1 = y_1x_1$, we get $w_1$ is an inner witness of $(x_1y_1,y_1x_1)$ as well as $(u_1,v_1)$. 

          Let $q_1 = w_1^{-1}B_1$. Clearly, $q_1$ is of the form $(y_1x_1)^*y_1$ or $(x_1y_1=y_1x_1)^*$.
          In both cases,
          after matching $q_1$ in block $B_1'$ of $v$, we observe that a factor equal to $w_1$ appears in the suffix of the block of $B_1'$, i.e., $B_1' = q_1w_1$ (as shown below). 
 \begin{align*}
  u &= \overbrace{\as_0\,\overbrace{x_1y_1 \cdots x_1}^{w_1}}^{p\bs_0} \overbrace{\cdots x_1y_1}^{q_1} \, \as_1\, x_2y_2 \cdots x_2y_2 \, \as_2 \cdots x_ky_k \cdots x_ky_k \,\as_k\\
v &= \bs_0\,\underbrace{y_1x_1 \cdots y_1}_{q_1} \underbrace{x_1 \cdots y_1x_1}_{w_1} \, \bs_1 \,y_2x_2 \cdots y_2x_2 \, \bs_2 \, \cdots y_kx_k \cdots \cdots y_kx_k \, \bs_k
\end{align*}
          \item  \emph{When $p\bs_0$ is a prefix of $\as_0$}: We show that $w_1$ is an outer witness of $(u_1,v_1)$, and there exists a word $q_1$ such that $B_1 = q_1w_1$ and $B_1' = w_1q_1$.

           Since $p\bs_0$ is a prefix of $\as_0$ and $|\as_0| \leq C < j_1$ , $|w_1| = |(p\bs_0)^{-1}\as_0|<|v_1^{j_1}|$. Therefore $w_1$ is a prefix of $B_1'(v_1)^{-1}$ in $v$. On matching $(p \bs_0)^{-1}u$ and $\bs_0^{-1}vp^{-1}$, since the block $B_1$ in $u$ starts with $x_1y_1$ and since there is at least one factor $y_1x_1$ in $B_1'$ after matching $w_1$, either $w_1 \in (y_1x_1)^*y_1$ (by Case~\ref{distinct} of \nameref{cutcorr}) or $w_1 \in (y_1x_1)^*$ and $x_1y_1 = y_1x_1$ (by \ref{same} of \nameref{cutcorr}).
           In the former case, $w_1$ is an outer witness of $(x_1y_1,y_1x_1) = (\rho_{u_1},\rho_{v_1})$ and hence is also an outer witness of $(u_1,v_1)$ by \Cref{samewitness}. When $w_1 \in (y_1x_1)^*$, we get $x_1y_1 = y_1x_1$ and hence $w_1$ is an outer witness of $(x_1y_1,y_1x_1)$ as well as $(u_1,v_1)$. 

          Let $q_1= (w_1)^{-1}B_1'$. As in the case above, $q_1 \in (x_1y_1)^*x_1$ or $q_1 \in (x_1y_1=y_1x_1)^*$. In both cases, after matching $q_1$ in block $B_1$ of $u$, a factor equal to $w_1$ appears in the suffix of the block of $B_1$, i.e., $B_1 = q_1w_1$ (as shown below).

          \begin{align*}
  u &= \overbrace{\as_0}^{p\bs_0w_1}\overbrace{x_1y_1 \cdots x_1}^{q_1}\overbrace{y_1\cdots x_1y_1}^{w_1} \, \as_1\, x_2y_2 \cdots x_2y_2 \, \as_2 \cdots x_ky_k \cdots x_ky_k \,\as_k\\
v &= \bs_0\,\,\,\underbrace{y_1x_1 \cdots y_1}_{w_1} \underbrace{x_1 \cdots y_1x_1}_{q_1} \, \bs_1 \,y_2x_2 \cdots y_2x_2 \, \bs_2 \, \cdots y_kx_k \cdots \cdots y_kx_k \, \bs_k
\end{align*}
      \end{itemize}

\noindent 2.\quad Proof by induction on $i\in [2,k]$. 
          The base case and the inductive step are identical. Hence we treat only the inductive step. 
          
          WLOG assume that $w_{i-1}$ is an inner witness of $(u_{i-1},v_{i-1})$. The case where $w_{i-1}$ is an outer witness of $(u_{i-1},v_{i-1})$ is symmetric. In this case, $B_{i-1} = w_{i-1}q_{i-1}$ and $B_{i-1}' = q_{i-1}w_{i-1}$ (as shown below). Moreover, $|w_{i-1}| \leq |u_{i-1}^{j_{i-1}}|$.

          \begin{align*}
  u &= \cdots \as_{i-2}\overbrace{x_{i-1}y_{i-1} \cdots x_{i-1}}^{w_{i-1}} \overbrace{\cdots x_{i-1}y_{i-1}}^{q_{i-1}} \, \as_{i-1}\, x_iy_i \cdots x_iy_i \cdots \\
v &= \cdots \bs_{i-2}\underbrace{y_{i-1}x_{i-1} \cdots }_{q_{i-1}} \underbrace{x_{i-1} \cdots y_{i-1}x_{i-1}}_{w_{i-1}} \, \bs_{i-1} \,y_ix_i \cdots y_ix_i \cdots
\end{align*}

On matching further, there are two cases for $w_i = [\as_{i-1},w_{i-1}\bs_{i-1}]_L$ as follows.
\begin{itemize}
    \item \emph{When $\as_{i-1}$ is a prefix of $w_{i-1}\bs_{i-1}$}: Here, $w_i = (\as_{i-1})^{-1}w_{i-1}\bs_{i-1}$ begins at the block $B_i$ in $u$ as depicted in \Cref{fig:cas1.a}. Since $|w_i| \leq |w_{i-1}\bs_{i-1}| \leq |u_{i-1}^{j_{i-1}}| + C < |u_i^{j_i}|/2$, $w_i$ ends within block $B_i$ and there exists at least one $x_iy_i$ in $B_i$ immediately following $w_i$. Let $B_i = w_i q_i$ for some word $q_i$. On further matching $q$ in $u$ and $v$, since the block $B_i'$ in $v$ starts with $y_ix_i$ and since $q_i$ consists at least one factor $x_1y_1$, either $w_i \in (x_iy_i)^*x_i$ (by Case~\ref{distinct} of \nameref{cutcorr}) or $w_i \in (x_iy_i)^*$ and $x_iy_i=y_ix_i$ (by \ref{same} of \nameref{cutcorr}). In the former case, $w_i$ is an inner witness of $(x_iy_i,y_ix_i) = (\rho_{u_i},\rho_{v_i})$ and hence also an inner witness of $(u_i,v_i)$ by \Cref{samewitness}. When $w_i \in (x_iy_i)^*$, we get $x_iy_i = y_ix_i$ and hence $w_i$ is an inner witness of $(x_iy_i,y_ix_i)$ as well as $(u_i,v_i)$.

          Since $B_i = w_iq_i$, after matching $q_i$ in block $B_i'$ of $v$, we observe that a factor equal to $w_i$ appears in the suffix of the block of $B_i'$, i.e., $B_i' = q_iw_i$ (as shown below). 

\begin{figure}[h]
          \begin{align*}
  u &= \cdots \as_{i-2}\overbrace{x_{i-1}y_{i-1} \cdots x_{i-1}}^{w_{i-1}} \overbrace{\cdots x_{i-1}y_{i-1}}^{q_{i-1}} \, \as_{i-1}\, \overbrace{x_iy_i \cdots }^{w_i}\overbrace{\cdots x_iy_i}^{q_i}\, \as_{i} \cdots \\
v &= \cdots \bs_{i-2}\underbrace{y_{i-1}x_{i-1} \cdots }_{q_{i-1}} \underbrace{x_{i-1} \cdots y_{i-1}x_{i-1}}_{w_{i-1}} \, \bs_{i-1} \underbrace{y_ix_i \cdots}_{q_i} \underbrace{\cdots y_ix_i}_{w_i}\, \bs_{i} \cdots
\end{align*}
\caption{$w_i$ when $w_{i-1}$ is an inner witness and $\as_{i-1}$ is a prefix of $w_{i-1}\bs_{i-1}$}
\label{fig:cas1.a}
\end{figure}

    \item When $w_{i-1}\bs_{i-1}$ is a prefix of $\as_{i-1}$: Here, $w_i = (w_{i-1}\bs_{i-1})^{-1}\as_{i-1}$ begins at the block $B_i'$ in $v$. Since $|w_i| \leq |\as_{i-1}| \leq C < |u_i^{j_i}|/2$, $w_i$ ends within block $B_i'$ and there exists at least one $y_ix_i$ in $B_i'$ immediately following $w_i$. Let $B_i' = w_i q_i$ for some word $q_i$. On further matching $q$ in $u$ and $v$, since the block $B_i$ in $u$ starts with $x_iy_i$ and since $q_i$ contains at least one $y_ix_i$, either $w_i \in (y_ix_i)^*y_i$ (by Case~\ref{distinct} of \nameref{cutcorr}) or $w_i \in (y_ix_i)^*$  and $x_iy_i = y_ix_i$ (by \ref{same} of \nameref{cutcorr}). In the former case, $w_i$ is an outer witness of $(x_iy_i,y_ix_i) = (\rho_{u_i},\rho_{v_i})$ and hence also an outer witness of $(u_i,v_i)$ by \Cref{samewitness}. When $w_i \in (y_ix_i)^*$, we get $x_iy_i = y_ix_i$ and hence $w_i$ is an outer witness of $(x_iy_i,y_ix_i)$ as well as $(u_i,v_i)$.

         Since $B_i' = w_iq_i$, after matching $q_i$ in block $B_i$ of $u$, a factor equal to $w_i$ appears in the suffix of the block of $B_i$, i.e., $B_i = q_iw_i$ (as shown below).

\begin{figure}[h]
           \begin{align*}
  u &= \cdots \as_{i-2}\overbrace{x_{i-1}y_{i-1} \cdots x_{i-1}}^{w_{i-1}} \overbrace{\cdots x_{i-1}y_{i-1}}^{q_{i-1}} \, \as_{i-1}\, \overbrace{x_iy_i \cdots }^{q_i}\overbrace{\cdots x_iy_i}^{w_i}\, \as_{i} \cdots \\
v &= \cdots \bs_{i-2}\underbrace{y_{i-1}x_{i-1} \cdots }_{q_{i-1}} \underbrace{x_{i-1} \cdots y_{i-1}x_{i-1}}_{w_{i-1}} \, \bs_{i-1} \underbrace{y_ix_i \cdots}_{w_i} \underbrace{\cdots y_ix_i}_{q_i}\, \bs_{i} \cdots
\end{align*}
\caption{$w_i$ when $w_{i-1}$ is an inner witness and $w_{i-1}\bs_{i-1}$ is a prefix of $\as_{i-1}$}
\label{fig:cas1.b}
\end{figure}
\end{itemize}
Thus, the claim holds for $i \in \{2, \ldots, k\}$.\\
\newline
\noindent 3.\quad
          We prove the case when both $w_1$ and $w_k$ are inner witnesses of $(u_1,v_1)$ and $(u_k,v_k)$ respectively (as shown below). All other cases are symmetric.

\begin{align*}
  u &= \overbrace{\as_0\,\overbrace{x_1y_1 \cdots x_1}^{w_1}}^{p\bs_0} \overbrace{\cdots x_1y_1}^{q_1} \, \as_1\, x_2y_2 \cdots x_2y_2 \, \as_2 \cdots \overbrace{x_ky_k \cdots x_k}^{w_k}\overbrace{\cdots x_ky_k}^{q_k} \,\as_k\\
v &= \bs_0\,\underbrace{y_1x_1 \cdots y_1}_{q_1} \underbrace{x_1 \cdots y_1x_1}_{w_1} \, \bs_1 \,y_2x_2 \cdots y_2x_2 \, \bs_2 \, \cdots \underbrace{y_kx_k \cdots}_{q_k}\underbrace{x_k \cdots y_kx_k}_{w_k} \, \bs_k
\end{align*}

We get that
\begin{align*}
    w_k\bs_k &= \as_kp &&\text{ (By matching $u$ and $v$)}\\
    w_k\bs_k\bs_0 &= \as_kp\bs_0 &&\text{ (Appending $\bs_0$)}\\
     &= \as_k\as_0w_1. &&\text{ (Subs. $p\bs_0 = \as_0w_1$)} \qedhere
\end{align*}
          
      \end{proof}

\begin{lemma}\label{claimwi}
       For each $i\in [1,k]$, let $w_i$ be the witness of $(u_i,v_i)$ given by \Cref{claimblock}. Then for each $j\in [1,k]$ such that $j>i$, the following holds.
       \begin{enumerate}
           \item If $w_{i}$ is an inner witness of $(u_{i},v_{i})$, then $w_j = [\as_i \cdots \as_{j-1}, w_i\bs_i \cdots \bs_{j-1}]_L$. More precisely, if $w_j$ is an inner witness of $(u_j,v_j)$, then $w_j = (\as_i \cdots \as_{j-1})^{-1}w_i\bs_i \cdots \bs_{j-1}$, otherwise $w_j = (w_i\bs_i \cdots \bs_{j-1})^{-1}\as_i \cdots \as_{j-1}$. 
           \item If $w_{i}$ is an outer witness of $(u_{i},v_{i})$, then $w_j = [w_i\as_i \cdots \as_{j-1}, \bs_i \cdots \bs_{j-1}]_L$. More precisely, if $w_j$ is an inner witness of $(u_j,v_j)$, then $w_j = (w_i\as_i \cdots \as_{j-1})^{-1}\bs_i \cdots \bs_{j-1}$, otherwise $w_j = (\bs_i \cdots \bs_{j-1})^{-1}w_i\as_i \cdots \as_{j-1}$.
       \end{enumerate}

\end{lemma}
\begin{proof}
 Proof by induction on $j-i$. Base case is when $j = i+1$ and the claim holds by \Cref{claimblock}. WLOG assume $w_i$ is an inner witness of $(u_i,v_i)$ and the case 1.~holds for $j$ with $i < j < k$, and we show that it also holds for $j+1$. There are two cases:
        \begin{itemize}
            \item \emph{When $w_j$ is an inner witness of $(u_j,v_j)$}: $w_j = (\as_i \cdots \as_{j-1})^{-1}w_i\bs_i \cdots \bs_{j-1}$ (by inductive hypothesis). From \Cref{claimblock}, we get that $w_{j+1}$ is a witness of $(u_{j+1},v_{j+1})$ and
            \begin{align*}
            w_{j+1} &= [\as_j,w_j\bs_j]_L && \text{(By \Cref{claimblock})}\\
            &=[\as_j,(\as_i \cdots \as_{j-1})^{-1}w_i\bs_i \cdots \bs_{j-1}\bs_j]_L && \text{(Subs. $w_j$)}\\
            &=[\as_i \cdots \as_{j-1}\as_j,w_i\bs_i \cdots \bs_{j-1}\bs_j]_L &&\text{(Since $[u,v]_L = [wu,wv]_L$ for all $w$)}
        \end{align*}
        If $w_{j+1}$ is an inner witness of $(u_{j+1},v_{j+1})$, then 
        \begin{align*}
            w_{j+1}&= (\as_j)^{-1}w_j\bs_j\\
            &= (\as_j)^{-1}(\as_i \cdots \as_{j-1})^{-1}w_i\bs_i \cdots \bs_{j-1}\bs_j &&\text{ (Subs. $w_j$ by I.H.)}\\
            &=(\as_i \cdots \as_{j-1}\as_j)^{-1}w_i\bs_i \cdots \bs_{j-1}\bs_j
        \end{align*}
        Similarly, if $w_{j+1}$ is an outer witness of $(u_{j+1},v_{j+1})$, then 
        \begin{align*}
            w_{j+1}&= (w_j\bs_j)^{-1}\as_j\\
            &= ((\as_i \cdots \as_{j-1})^{-1}w_i\bs_i \cdots \bs_{j-1}\bs_j)^{-1}\as_j &&\text{ (Subs. $w_j$ by I.H.)}\\
            &=(w_i\bs_i \cdots \bs_{j-1}\bs_j)^{-1}\as_i \cdots \as_{j-1}\as_j
        \end{align*}
    
         \item \emph{When $w_j$ is an outer witness}: $w_j = (w_i\bs_i \cdots \bs_{j-1})^{-1}\as_i \cdots \as_{j-1}$ (by inductive hypothesis). From \Cref{claimblock}, we get that $w_{j+1}$ is a witness of $(u_{j+1},v_{j+1})$ and
            \begin{align*}
            w_{j+1} &= [w_j\as_j,\bs_j]_L && \text{(By \Cref{claimblock})}\\
            &=[(w_i\bs_i \cdots \bs_{j-1})^{-1}(\as_i \cdots \as_{j-1})\as_j,\bs_j]_L && \text{(Subs. $w_j$)}\\
            &=[\as_i \cdots \as_{j-1}\as_j,w_i\bs_i \cdots \bs_{j-1}\bs_j]_L &&\text{($[u,v]_L = [wu,wv]_L$ for all $w$)}
        \end{align*}
         If $w_{j+1}$ is an inner witness of $(u_{j+1},v_{j+1})$, then 
        \begin{align*}
            w_{j+1}&= (w_j\as_j)^{-1}\bs_j\\
            &= ((w_i\bs_i \cdots \bs_{j-1})^{-1}\as_i \cdots \as_{j-1}\as_j)^{-1}\bs_j &&\text{ (Subs. $w_j$ by I.H.)}\\
            &=(\as_i \cdots \as_{j-1}\as_j)^{-1}w_i\bs_i \cdots \bs_{j-1}\bs_j
        \end{align*}
        Similarly, if $w_{j+1}$ is an outer witness of $(u_{j+1},v_{j+1})$, we can show that $w_{j+1} = (w_i\bs_i \cdots \bs_{j-1}\bs_j)^{-1}\as_i \cdots \as_{j-1}\as_j$.
        \end{itemize}
        The case where $w_i$ is an outer witness of $(u_i,v_i)$ is symmetric.
        \end{proof}

       \begin{lemma}For each $i\in [1,k]$, let $w_i$ be the witness of $(u_i,v_i)$ given by \Cref{claimblock}. 
           For each $i \in [1,k]$, if $w_i$ is an inner (resp.~outer) witness of $(u_i,v_i)$, then $w_i$ is also an inner (resp.~outer) witness of $(\as_i \cdots \as_k\as_0 \cdots \as_{i-1},\bs_i \cdots \bs_k\bs_0 \cdots \bs_{i-1}).$
               \end{lemma}

       \begin{proof}
           
       WLOG assume that $w_i$ is an inner witness of $(u_i,v_i)$. The other case is symmetric. Using \Cref{claimwi}, 
       \[w_k = [\as_i \cdots \as_{k-1},w_i\bs_i \cdots \bs_{k-1}]_L,\]
       \begin{equation}\label{eq:cl6:1}
         w_i = (\as_1 \cdots \as_{i-1})^{-1}w_1\bs_1\cdots \bs_{i-1}  
       \end{equation}

       We perform a case analysis on whether $w_1$ and $w_k$ are inner/outer witnesses. Assume that $w_1$ is as inner witness of $(u_1,v_1)$ and $w_k$ is an outer witness of $(u_k,v_k)$. Hence,
       \begin{equation}\label{eq:cl6:2}
           w_k = (w_i\bs_i \cdots \bs_{k-1})^{-1}\as_i \cdots \as_{k-1}
       \end{equation}
      Using Item 3 of \Cref{claimblock}, we get that
      \begin{equation}\label{eq:cl6:3}
          w_1 = (w_k\as_k\as_0)^{-1}\bs_k\bs_0.
      \end{equation}
      Therefore,
      \begin{align*}
          w_i &=(\as_1 \cdots \as_{i-1})^{-1}w_1\bs_1\cdots \bs_{i-1} &&\text{ (Using \Cref{eq:cl6:1}}\\
          &=(\as_1 \cdots \as_{i-1})^{-1}(w_k\as_k\as_0)^{-1}\bs_k\bs_0\bs_1\cdots \bs_{i-1} &&\text{ (Subs. \Cref{eq:cl6:3}}\\
          &=(w_k\as_k\as_0\as_1 \cdots \as_{i-1})^{-1}\bs_k\bs_0\bs_1\cdots \bs_{i-1}\\
          &=((w_i\bs_i \cdots \bs_{k-1})^{-1}\as_i \cdots \as_{k-1}\as_k\as_0\as_1 \cdots \as_{i-1})^{-1}\bs_k\bs_0\bs_1\cdots \bs_{i-1} &&\text{ (Subs. \Cref{eq:cl6:2}}\\
          &=(\as_i \cdots \as_{k-1}\as_k\as_0\as_1 \cdots \as_{i-1})^{-1}w_i\bs_i \cdots \bs_{k-1}\bs_k\bs_0\bs_1\cdots \bs_{i-1} 
      \end{align*}
      Hence $w_i$ is also an inner witness of $(\as_i \cdots \as_k\as_0 \cdots \as_{i-1},\bs_i \cdots \bs_k\bs_0 \cdots \bs_{i-1})$.
      All other cases for $w_1$ and $w_k$ are analogous.

\end{proof}
 
  \begin{lemma}\label{claim2Kleene}
     Assume that $M$ is conjugate and a differentially-pumped pair in $M$ has a short cut. If any two singleton reduxes of $M$ have unique common witnesses, then they are identical. Moreover, either both are inner, or both are outer witnesses.
  \end{lemma}
  \begin{proof}
      Let 
      \begin{align}
      M_i &= (\as_0 \cdots \as_{i-1}, \bs_0 \cdots \bs_{i-1})G_i^*(\as_i \cdots \as_k,\bs_i \cdots \bs_k)\\
      M_j &= (\as_0 \cdots \as_{j-1}, \bs_0 \cdots \bs_{j-1})G_j^*(\as_j \cdots \as_k,\bs_j \cdots \bs_k)
      \end{align}
      be two singleton reduxes with unique common witnesses $c_i$ and $c_j$ respectively. Using \Cref{singlemonomial}, we get that the set $G_i \cup \{(\as_i \cdots \as_k\as_0 \cdots \as_{i-1}, \bs_i \cdots \bs_k\as_0 \cdots \bs_{i-1})\}$ has a unique common witness, say $z_i$. Moreover, by construction, $z_i$ is also the unique common witness of the pairs $(u_i,v_i)$ used in the canonical pair from \Cref{def:diffpumped} and $(\as_i \cdots \as_k\as_0 \cdots \as_{i-1}, \bs_i \cdots \bs_k\bs_0 \cdots \bs_{i-1})$. Analogously, let $z_j$ be the unique common witness of  the set $G_j \cup \{(\as_j \cdots \as_k\as_0 \cdots \as_{j-1}, \bs_j \cdots \bs_k\as_0 \cdots \bs_{j-1})\}$ and  $\{(u_j,v_j),(\as_j \cdots \as_k\as_0 \cdots \as_{j-1}, \bs_j \cdots \bs_k\as_0 \cdots \bs_{j-1})\}$ where $(u_j,v_j)$ is from \Cref{def:diffpumped}.  
     From \Cref{claimwi}, $z_i=w_i$ and $z_j =w_j$. WLOG assume that $i < j$. There are four cases depending on whether $w_i$ and $w_j$ is an inner or outer witness.
      \begin{itemize}
                \item \emph{When $w_i$ and $w_j$ are inner witnesses}: By \Cref{claimwi}, $w_j = (\as_i \cdots \as_{j-1})^{-1}w_i\bs_i \cdots \bs_{j-1}$. By \Cref{g1monomialwitness}, $c_i = [\as_0 \cdots \as_{i-1}w_i,\bs_0 \cdots \bs_{i-1}]_R$ and $c_j = [\as_0 \cdots \as_{j-1}w_j,\bs_0 \cdots \bs_{j-1}]_R$. Assume that $c_j$ is a common inner witness. We show that $c_i=c_j$ and it is also an inner witness. 
          \begin{align*}
              c_j &= [\as_0 \cdots \as_{j-1}w_j,\bs_0 \cdots \bs_{j-1}]_R\\
              &=\as_0 \cdots \as_{j-1}w_j(\bs_0 \cdots \bs_{j-1})^{-1} &&\text{(By \Cref{g1monomialwitness})}\\
              &= \as_0 \cdots \as_{j-1}(\as_i \cdots \as_{j-1})^{-1}w_i\bs_i \cdots \bs_{j-1}(\bs_0 \cdots \bs_{j-1})^{-1} && \text{ (Subs. $w_j$) }\\
                &= \as_0 \cdots \as_{i-1}w_i(\bs_0 \cdots \bs_{i-1})^{-1} && \text{ (Simplifying) }\\
              &= [\as_0 \cdots \as_{i-1}w_i,\bs_0 \cdots \bs_{i-1}]_R\\
              &= c_i
          \end{align*}
          We also deduce that $c_i$ is also a common inner witness by \Cref{g1monomialwitness}. Similarly, we can show that when $c_j$ is a common outer witness of $M_j$, then $c_i=c_j$ is also a common outer witness of $M_i$.
          \item \emph{When $w_i$ and $w_j$ is a unique common outer witness}: the case the similar to above.
          \item \emph{When $w_i$ is a unique common inner witness and $w_j$ is a unique common outer witness}: By \Cref{claimwi}, $w_j = (w_i\bs_i \cdots \bs_{j-1})^{-1}\as_i \cdots \as_{j-1}$ . By \Cref{g1monomialwitness}, $c_i = [\as_0 \cdots \as_{i-1}w_i,\bs_0 \cdots \bs_{i-1}]_R$ and $c_j = [\as_0 \cdots \as_{j-1},\bs_0 \cdots \bs_{j-1}w_j]_R$. Assume that $c_j$ is a common inner witness. We show that $c_i=c_j$ and it is also an inner witness. 
          \begin{align*}
              c_j &= [\as_0 \cdots \as_{j-1},\bs_0 \cdots \bs_{j-1}w_j]_R\\
              &=\as_0 \cdots \as_{j-1}(\bs_0 \cdots \bs_{j-1}w_j)^{-1} &&\text{(By \Cref{g1monomialwitness})}\\
              &=\as_0 \cdots \as_{j-1}(\bs_0 \cdots \bs_{j-1}(w_i\bs_i \cdots \bs_{j-1})^{-1}\as_i \cdots \as_{j-1})^{-1} && \text{ (Subs. $w_j$) }\\
              &=\as_0 \cdots \as_{j-1}(\bs_0 \cdots \bs_{i-1}(w_i)^{-1}\as_i \cdots \as_{j-1})^{-1} && \text{ (Simplifying) }\\
              &=\as_0 \cdots \as_{i-1}(\bs_0 \cdots \bs_{i-1}(w_i)^{-1})^{-1}\\
              &=\as_0 \cdots \as_{i-1}w_i(\bs_0 \cdots \bs_{i-1})^{-1}\\
              &=[\as_0 \cdots \as_{i-1}w_i,\bs_0 \cdots \bs_{i-1}]_R\\
              &=c_i
          \end{align*}
           We also deduce that $c_i$ is also a common inner witness by \Cref{g1monomialwitness}. Similarly, we can show that when $c_j$ is a common outer witness of $M_j$, $c_i=c_j$ is also a common outer witness of $M_i$.
          \item \emph{When $w_i$ is a unique common outer witness and $w_j$ is a unique common inner witness}: the case is similar to above. \qedhere
       \end{itemize}
  \end{proof}
  \begin{proof}[Proof of \Cref{prop:short}]

    Now we prove that $M$ has a common witness. 
    Let $M_i$ for $i\in [1,k]$ denote the singleton redux of $M$ obtained from $M$ by substituting all but $G_i^*$ by $(\varepsilon,\varepsilon)$.
  Since $M$ is conjugate, each of its singleton redux, $M_i$ for $i \in [1,k]$ is also conjugate. By \Cref{singlemonomial}, each $M_i$ has a common witness.

    Since the redux of $M$ is nonempty, if a singleton redux has infinitely many common witnesses, then it is a subset of powers of the primitive root of the redux by \Cref{singlemonomial}\ref{Thm7:3}. Hence, any witness for the primitive root of the redux is a witness for that singleton redux of $M$. 
    If all singleton reduxes have infinitely many witnesses, they are all subsets of powers of the 
primitive root. Hence, all of them share a common witness, and by \Cref{mcommonwitness} $M$ has a common witness.

Otherwise assume that there exists a singleton redux with a unique common witness. 
By \Cref{claim2Kleene}, if any two singleton reduxes have unique witnesses, then they are the same. Thus, all singleton reduxes with unique witnesses share a unique witness; 
Let it be $z$. Since $z$ is also a witness for the redux. By \Cref{samewitness}, it is a witness for any power of the primitive root of the redux. Therefore, $z$ is also a witness for all singleton reduxes with infinitely many witnesses by \Cref{singlemonomial}\ref{Thm7:3}. Hence, $z$ is a common witness of each singleton reduxes of $M$. By virtue of \Cref{mcommonwitness}, $M$ has a common witness when $p$ is short.

\end{proof}

\subsection{Long cut for differentially-pumped pair}

       In this subsection we prove the following 
       
       \begin{proposition}\label{prop:long} Assume that $M$ is conjugate. 
         Whenever a differentially-pumped pair of $M$ has a long cut, then
         \[M \subseteq (\as_0,\bs_0)(\rho_1,\rho_1')^* (\as_{1},\bs_{1})
           \cdots 
           (\rho_k,\rho_k')^*(\as_k,\bs_k)\]
           such that $\rho_i\sim \rho_i' \sim \rho_j \sim \rho_j'$ for all $i,j\in [1,k]$. 
           Let $\bar{\as} = \as_1, \ldots, \as_{k-1}, \as_k\as_0$, $\bar{\bs} = \bs_1, \ldots, \bs_{k-1}, \bs_k\bs_0$. Let $\bar{\varrho} $ and $\bar{\varrho}'$ denote the sequence of pairs 
\begin{align}
    \bar{\varrho} &=(\rho_1,\rho_2), \ldots, (\rho_{k-1},\rho_k), (\rho_k,\rho_1)\\
    \bar{\varrho}' &=(\rho'_1,\rho'_2), \ldots, (\rho'_{k-1},\rho'_k), (\rho'_k,\rho'_1).
\end{align}
 Then, one of the following holds.    
         \begin{itemize}
             \item 
             For each $i\in [1,k]$, $(\bar{\alpha})_i$ and $(\bar{\beta})_i$are inner witnesses of $(\bar{\varrho})_i$ and $(\bar{\varrho}')_i$ respectively, and we have  \[\fst(M) \sim \snd(M)  \subsetsim \rho_i^*  \sim  \rho_i'^*   \] for each for each $i\in [1,k]$.
             \item There is a unique $i\in [1,k]$ such that 
             $(\bar{\alpha})_i$ is not an inner witness of $(\bar{\varrho})_i$ and there is a unique $j\in [1,k]$ such that 
        $(\bar{\beta})_j$ is not an inner witness of $(\bar{\varrho}')_i$, and 
           we have  
           \begin{align*} \fst(M) &\sim \snd(M)  \subsetsim \as_i\dots \as_k\as_0
             \dots a_{i-1}\rho_i^*  \sim \bs_j\dots \bs_k\bs_0
             \dots \bs_{j-1}{\rho_j'}^*.
             \end{align*}
         \end{itemize} 
       \end{proposition}
        
       Towards this, we first show that each $G_i^*$, for $i \in [1,k]$,  have infinitely many witnesses and hence $G_i^* \subseteq (\rho_i,\rho_i')^*$ for some conjugate primitive pair $(\rho_i,\rho_i')$ for $i\in [1,k]$. Moreover, for all $i,j \in [1,k]$, $\rho_i \sim \rho_i' \sim \rho_j \sim \rho_j'$. 
       \begin{lemma}\label{claim:rhoconjugate}
            Assume that a differentially-pumped pair $(u,v)$ from \Cref{def:diffpumped} has a long cut. Then the primitive roots of pairs $(u_i,v_i)$ for $i \in [1,k]$ in $(u,v)$ are conjugate, i.e., 
       
       \begin{equation}\label{eq:rhoconjugate}
       \rho_{u_1} \sim \rho_{v_1} \sim \rho_{u_2} \sim \rho_{v_2} \cdots \sim \rho_{u_k} \sim \rho_{v_k}.
       \end{equation}
       \end{lemma}

      \begin{proof}
       There are three cases for a long cut $(p,q)$ in $(u,v)$:
       \begin{enumerate}
           \item The cut $p$ in $u$ ends in the second half of the block $B_i$ for $1 \leq i < k$.
           \item The cut $p$ in $u$ ends after the block $B_{i-1}$ and before the second half of the block $B_i$ for $1 < i \leq k$.
           \item The cut $p$ in $u$ ends in the second half of the block $B_k$.
       \end{enumerate}
       Note that whenever the cut $p$ in $u$ ends in the first half of block $B_1$, or ends after the block $B_k$, either $p$ or $q$ is short. Since $p$ and $q$ are long, $|p\bs_0| > |\as_0u_1^{j_1-1}| \geq|u_1^{j_1-1}|   \geq 2k \cdot \max(C,\ell)$ (since $j_1 > 2k \cdot \max(C,\ell)$). Therefore
       \begin{align}\label{eq:|p|}
            |p| > 2k \cdot \max(C,\ell) - |\bs_0|.
       \end{align}

       For $i \in [1,k]$, let $B_i^f$ and $B_i^s$  represent the first and second halves of the block $B_i$, respectively.
       Similarly, let $B_i'^f$ and $B_i'^s$  represent the first and second halves of the block $B_i'$, respectively.
       Clearly, $|B_i^f| = |B_i'^f|$ and $|B_i^s| = |B_i'^s|$. Since $|B_1| > 2k \cdot \max(C,\ell)$ $|B_1^f|$, we have $|B_1^s| \geq k \cdot max(C,\ell)$. Similarly, from \Cref{eq:diffpump}, we get that for each $i \in [2,k]$,
       \begin{equation}
           |B_1| + |B_2| + \cdots |B_{i-1}| + C \leq |B_i^f| = |B_i^s| 
        \end{equation}
       \begin{align*}
  u' &= \as_0\,\overbrace{u_1 \cdots }^{B_1^f} \overbrace{\cdots  u_1}^{B_1^s}\as_1\overbrace{u_2 \cdots}^{B_2^f}\overbrace{\cdots u_2}^{B_2^s}\as_2 \cdots \overbrace{u_k \cdots}^{B_k^f}\overbrace{\cdots u_k}^{B_k^s}\as_k \\
v' &= \bs_0\,\underbrace{v_1 \cdots }_{B_1'^f} \underbrace{\cdots  v_1}_{B_1'^s} \bs_1 \underbrace{v_2 \cdots}_{B_2'^f}\underbrace{\cdots v_2}_{B_2'^s} \bs_2 \, \cdots \underbrace{v_k \cdots}_{B_k'^f}\underbrace{\cdots v_k}_{B_k'^s} \bs_k 
\end{align*}
       We now show that, in all the above cases for long cut, \Cref{eq:rhoconjugate} holds.

       \begin{itemize}
           
           \item \emph{When the cut $p$ in $u$ ends in  $B_i^s$ for $1 \leq i < k$}: 

           \begin{align*}
  u &= \overbrace{\as_0 \cdots B_i^f B_i^s}^{p} \as_i \cdots B_{k}\as_k\\
v &= \bs_0 \cdots {B_i'^f}{B_i'^s} \bs_i \cdots {B_k'^f}\underbrace{B_k'^s\bs_k}_{p}
\end{align*}

Since $|p| < |\as_0B_1 \cdots \as_{i-1}B_i| < |B_1| + \cdots + |B_{i}| + C < |B_k'^s|$, $p$ in $v$ is matched/contained in $B_k'^s\bs_k$.  The length of each block $B_j$ for $j \in [1,i]$ in $u$ is at least $\ell$. Therefore, by virtue of \Cref{commonfactor}, $\rho_{u_j} \sim \rho_{v_k}$. Hence, for all $j \in [1,i]$,
\begin{equation}\label{eqcase1:1}
    \rho_{v_j} \sim \rho_{u_j} \sim \rho_{v_k} \sim \rho_{u_k}.
\end{equation}
Now, on matching the suffixes of $q$ in $u$ and $v$, we get that a prefix  of $B_k$ of length $|p(\bs_k)^{-1}| - |\as_k|$ matches with suffix of $B_{k-1}'\bs_{k-1}$. Hence, there exists a common factor between blocks $B_k$ and $B_{k-1}'$ of length atleast 
 \begin{align}
     |p(\bs_k)^{-1}| - |\as_k| - |\bs_{k-1}| &= |p| - |\bs_k|-|\as_k| - |\bs_{k-1}| \\
     &\geq 2k \cdot \max(C,\ell) - |\bs_0|- |\bs_k|-|\as_k| - |\bs_{k-1}| && (\text{By \Cref{eq:|p|}})\\
 &\geq 2k \cdot \max(C,\ell) - 2C\\ 
     &\geq 2(k-1) \cdot \max(C,\ell) \\
     &> \ell.
 \end{align} By \Cref{commonfactor}, we get that $\rho_{v_k} \sim \rho_{u_k} \sim \rho_{v_{k-1}} \sim \rho_{u_{k-1}}$. Continuing inductively, for all $j \in [i+1,k]$, we get that there exists a common factor between blocks $B_j$ and $B_{j-1}'$ of length atleast 
 \begin{align}
     |p(\bs_k)^{-1}| - \sum_{n \in [j, \ldots, k]} |\as_n| - \sum_{n \in [j-1, \ldots, k-1]}|\bs_{n}| &=|p| -|\bs_k|-  \sum_{n \in [j, \ldots, k]} |\as_n| - \sum_{n \in [j-1, k-1]}|\bs_{n}| \\
     &\geq   2k \cdot \max(C,\ell) - 2C \\
&\geq 2(k-1) \cdot \max(C,\ell) \\
&> \ell.
 \end{align}
 Hence, for $j \in [i+1,k]$, by \Cref{commonfactor}, we deduce that
\begin{equation}\label{eq:case1:2}
    \rho_{u_k} \sim \rho_{v_{k-1}} \sim \rho_{u_{k-1}} \cdots \sim \rho_{v_{i+1}} \sim \rho_{u_{i+1}}.
\end{equation}
By combining Equations~\ref{eqcase1:1} and \ref{eq:case1:2}, we get that \Cref{eq:rhoconjugate} holds.

  \item \emph{When the cut $p$ in $u$ ends after the block $B_{i-1}$ and before the block $B_i^s$ for some $1 < i \leq k$}:

       \begin{align*}
  u &= \overbrace{\as_0 \cdots{B_{i-1}}\as_{i-1}{B_i^f}}^{p}\overbrace{{B_i^s} \as_i\cdots \as_k}^{q} \\
v &= \bs_0 \cdots {B_{i-1}'} \bs_{i-1} {B_i'^f}{B_i'^s} \bs_i \cdots \bs_k
\end{align*} 

In this case, while comparing the prefixes of $q$ in $u$ and $v$, since $|B_i^f|,|B_i^s| > |B_1| + \cdots |B_{i-1}| + C = |B_1'| + \cdots |B_{i-1}'| + C$, the prefix $\bs_0 B_1' \cdots B_{i-1}'\bs_{i-1}$ of $q$ in $v$ is matched within the block $B_i$. Since length of each block $B_j'$ for $j \in [1,i-1]$  is at least $\ell$, by virtue of \Cref{commonfactor}, $\rho_{v_j} \sim \rho_{u_i}$. Therefore, for all $j \in [1,i-1]$
\begin{equation}\label{eqcase3:1}
    \rho_{u_j} \sim \rho_{v_j} \sim \rho_{v_i} \sim \rho_{u_i}.
\end{equation}

On further matching $q$ in $v$ and $u$ , there exists a suffix of the block $B_i'$ of length at least 
\begin{align}
    |\bs_0 B_1' \cdots B_{i-1}'\bs_{i-1}B_i'| - |\as_{i-1}B_i| &= |\bs_0 B_1' \cdots B_{i-1}'\bs_{i-1}| - |\as_{i-1}| 
    \\
    &\geq |B_1'| -  |\as_{i-1}| &&(\text{Since $i
    \geq 2$})\\
    &\geq 2k\max(C,\ell) - |\as_{i-1}| 
\end{align} to match with $\as_iB_{i+1}$. Hence, $B_i'$ and $B_{i+1}$ share a common factor of length atleast 
\begin{align}
    2k\max(C,\ell) - |\as_{i-1}| - |\as_{i}| &\geq 2k\max(C,\ell) - C \\
    &\geq 2(k-1)\max(C,\ell)\\ 
    &> \ell.
\end{align}  From \Cref{commonfactor}, we get that $\rho_{u_i} \sim \rho_{v_i} \sim \rho_{u_{i+1}} \sim \rho_{v_{i+1}}$.

Continuing inductively, for all $j \in [i, k-1]$, we get that blocks $B_j'$ and $B_{j+1}$ share a common factor of length atleast 
$$2k \max(C,\ell) - \sum_{n \in [i-1, j]}|\as_n| - \sum_{n \in [i, j-1]}|\bs_{n}| \geq 2(k-1) \cdot \max(C,\ell) > \ell.$$

From \Cref{commonfactor}, we get that $\rho_{v_j} \sim \rho_{u_{j+1}}$ for $j \in [i, k-1]$. By transitivity of conjugacy, we get
\begin{equation}\label{eqcase3:2}
\rho_{u_i} \sim \rho_{v_i} \sim \rho_{v_{i+1}} \sim \rho_{v_{i+1}} \sim \cdots \sim \rho_{u_k} \sim \rho_{v_k}.
\end{equation}
Combining Equations~\ref{eqcase3:1} and \ref{eqcase3:2}, we get that \Cref{eq:rhoconjugate} holds.

           \item \emph{The cut $p$ in $u$ ends within  $B_k^s$}: 
           \begin{align*}
  u &= \overbrace{\as_0\,{B_1^f} {B_1^s}\as_1{B_2^f}{B_2^s}\as_2 \cdots {B_k^f}{B_k^s}}^{p}\overbrace{\as_k^{\phantom{f}}}^{q} \\
v &= \bs_0\,{B_1'^f} {B_1'^s} \bs_1 {B_2'^f}{B_2'^s} \bs_2 \, \cdots{B_k'^f}{B_k'^s} \bs_k
\end{align*}

Since $q$ is not short,  $|q\as_0| > |\bs_0v_1^{j_1-1}| \geq |v_1^{j_1-1}| \geq 2k\max(C,\ell)$ (since $j_1 > 2k \cdot \max(C,\ell))$. Also, by assumption, $|q| \leq |B_k'^s\as_k|$. While matching the $q$ in $u$ and $v$, there exists a maximal $i \in [1,k-1]$, such that $\bs_0B_1' \cdots B_i'^f$ is a proper prefix of $q$ in $v$. Since  $|B_j'^f|$ for $j \in [1,i]$ is atleast $\geq k\max(C,\ell) > \ell$, and $|B_j'^f|$ is matched within the block $B_k$, by virtue of \Cref{commonfactor}, $\rho_{v_j} \sim \rho_{u_k}$. Therefore, for all $j \in [1,i]$
\begin{equation}\label{eqcase4:1}
    \rho_{u_j} \sim \rho_{v_j} \sim \rho_{v_k} \sim \rho_{u_k}.
\end{equation}

If $i = k-1$, then we are done. Otherwise, assume that $i \leq k-2$. Now, matching the prefix of $p$ in $u$ and $v$, we get that a suffix of block $B_{i+1}$ of length atleast 
\begin{align}
    |\as_0 \cdots B_i\as_iB_{i+1}| - |q^{-1}(\bs_0B_1' \cdots B_{i+1}'\bs_{i+1})| 
    &= |q| - (|\as_0 \cdots \as_i| - |\bs_0 \cdots \bs_{i+1}|) \\
    &\geq |q|-C \\
    &\geq 2k\max(C,\ell) \\
    &> \ell 
\end{align}
 matches with block $B_{i+2}'$, and hence $\rho_{u_{i+1}} \sim \rho_{v_{i+2}}$ by \Cref{commonfactor}.  
Continuing inductively, we get for $j \in [(i+1),(k-1)]$, block $B_{j}$ and $B_{j+1}'$ shares a common factor of length atleast $2(k-1)\max(C,\ell) > \ell$, and by \Cref{commonfactor},
\begin{equation}\label{eqcase4:2}
    \rho_{v_j} \sim \rho_{u_{j}} \sim \rho_{v_{j+1}} \sim \rho_{u_{j+1}}.
\end{equation}
Combining Equations~\ref{eqcase4:1} and \ref{eqcase4:2}, we get that \Cref{eq:rhoconjugate} holds.
      \qedhere

       \end{itemize}
       \end{proof}
As a consequence of \Cref{claim:rhoconjugate}, we get the following.

\begin{corollary}\label{cor:Mform}
    Assume that $M$ is conjugate. Whenever a differentially-pumped pair in $M$ has a long cut, each Kleene closure $G_i^*$ of $M$, for $i \in [1,k]$, has infinitely many common witnesses. Moreover,
    \begin{equation}\label{eq:Mform}
           M \subseteq (\as_0,\bs_0)(\rho_1,\rho_1')^* (\as_{1},\bs_{1})
           \cdots 
           (\rho_k,\rho_k')^*(\as_k,\bs_k).
       \end{equation}
for some $\rho_1 \sim \rho_1' \sim \cdots \rho_k \sim \rho_k'$.
\end{corollary}
\begin{proof}
     By virtue of Item 3 of \Cref{def:diffpumpedcanonical}, if there exists $G_i^*$, for some $i\in[1,k]$, with a unique common witness, then every differentially pumped pair $(u,v)$ constructed satisfies $\rho_{u_j} \not \sim \rho_{u_{\ell}}$ for some $j,\ell \in [1,k], j \neq \ell$. Therefore, if a a differentially-pumped pair in $M$ has a long cut, each Kleene closure $G_i^*$ has infinitely many witnesses, and therefore $G_i^* \subseteq (\rho_i,\rho_i')^*$ for some primitive pair $(\rho_i,\rho_i')$. Since by  \Cref{claim:rhoconjugate}, $\rho_i \sim \rho_j$ for all $i,j\in[1,k]$ the claim follows.
\end{proof}

       Therefore, in the long cut case, 
       \begin{equation}\label{eq:Mform}
           M \subseteq (\as_0,\bs_0)(\rho_1,\rho_1')^* \cdots (\as_{k-1},\bs_{k-1})(\rho_k,\rho_k')^*(\as_k,\bs_k).
       \end{equation}
such that $\rho_1 \sim \rho_1' \sim \cdots \rho_k \sim \rho_k'$. Hence, $ |\rho_1| = \cdots = |\rho_k| = |\rho'_1|  = \cdots = |\rho'_k|$, and let $d$ denote this common length where $d= |\rho_1| = \cdots = |\rho_k|$. 
We use the following shorthand for the rest of the section.
\begin{definition}
Let $\bar{\as} = \as_1, \ldots, \as_{k-1}, \as_k\as_0$ and $\bar{\bs} = \bs_1, \ldots, \bs_{k-1}, \bs_k\bs_0$. Let $\bar{\varrho} $ and $\bar{\varrho}'$ denote the sequence of pairs 
\begin{align}
    \bar{\varrho} &=(\rho_1,\rho_2), \ldots, (\rho_{k-1},\rho_k), (\rho_k,\rho_1)\\
    \bar{\varrho}' &=(\rho'_1,\rho'_2), \ldots, (\rho'_{k-1},\rho'_k), (\rho'_k,\rho'_1)
\end{align}

\end{definition}

\begin{definition}[$d$-near]
    Assume that $(u,v) \in M$ is conjugate with a cut $(x,y)$. We say $\as \in \bar{\as}$ and $\bs \in \bar{\bs}$ overlap  if some position of $\as$ matches (recall \Cref{def:matchcut}) with some position of $\bs$ w.r.t.~the cut $(x,y)$. Note that if $\as$ and $\bs$ do not overlap, then the nearest positions according to the matching are either the first position of $\as$ and the last position of $\bs$ or vice versa. We say $\as$ and $\bs$ are \emph{$d$-near} w.r.t.~the cut $(x,y)$ for some $d \in \N$ if either they overlap or there are at most $d$-many positions between their nearest positions. 
\end{definition}

For a sequence of words $\bar{s} = s_1, s_2, \cdots, s_n$, we denote by $(\bar{s})_i$, for $i \in [1,n]$, as the $i$-th word $s_i$ in the sequence.

  \begin{lemma}\label{claim:reduxform}
       Assume that a differentially-pumped pair $(u,v)$ from \Cref{def:diffpumped} has a long cut $(p,q)$, and let $d = |\rho_1| = \cdots = |\rho_k|$. Then precisely one of the following is true.
       \begin{itemize}
           \item

           None of the pairs $((\bar{\as})_i,(\bar{\bs})_j)$ for $i, j\in[1,k]$ are $d$-near. In this case, the word $(\bar{\as})_i$ is an inner witness of the pair $(\bar{\varrho})_i$, and the word $(\bar{\bs})_i$ is an inner witness of $(\bar{\varrho}')_i$ for each $i\in [1,k]$.
           \item 
 There exist a unique pair $((\bar{\as})_i,(\bar{\bs})_j)$ for some $i\neq j\in[1,k]$ that are $d$-near w.r.t.~the cut $(p,q)$. 
Moreover, $(\bar{\as})_l$  is an inner witness of the pair $(\bar{\varrho})_l$ for each $l\in [1,k]\setminus \{i\}$, and $(\bar{\bs})_l$  is an inner witness of the pair $(\bar{\varrho}')_l$ for each $l \in [1,k]\setminus \{j\}$.
\end{itemize}
       \end{lemma}

       \begin{proof} 
        In the definition of $(u,v)$, recall that $\ell = 2\cdot \max\{|u_i| \mid i \in [1,k]\} \geq 2\cdot d$.

    Firstly, we prove that in the long cut $(p,q)$ of $(u,v)$ there is at most one pair $(\bar{\as})_i$ and $(\bar{\bs})_j$ for some $i,j \in [1,k]$ that are $d$-near. 

Assume there exist $i,{i'}, j, {j'} \in [1,k]$ such that $(\bar{\as})_i$ and $(\bar{\bs})_j$, as well as, $(\bar{\as})_{i'}$ and $(\bar{\bs})_{j'}$ are $d$-near. By definition of the pair $(u,v)$, we have $i=i'$ if and only if $j=j'$ since a word in $\bar{\as}$ can be $d$-near with at most one word in $\bar{\bs}$ and vice-versa. Hence, assume for the sake of contradiction that $i \neq i'$ and $j \neq j'$. 

We claim that $i \neq j$ and $i' \neq j'$. For the sake of contradiction assume that $i=j$, i.e., $(\bar{\as})_i$ and $(\bar{\bs})_i$ are $d$-near. We do a case analysis. When $i=k$, i.e., when $(\bar{\as})_k = \as_k\as_0$ and $(\bar{\bs})_k = \bs_k\bs_0$ are $d$-near then either $p$ or $q$ is short since $\min(|p|,|q|) \leq \max(|\as_k\as_0|, |\bs_k\bs_0|) + d$. Therefore, 
\begin{align}
    \min(|p\bs_0|,|q\as_0|) &\leq 2C +d \\
    &< 2C + \ell\\
    &<2k\cdot \max(C,\ell) &&\text{(Since $k \geq 2$)}
\end{align}

\noindent contradicting the assumption that the cut is long. Next assume that $i \neq k$. In this case, note that $(\bar{\as})_i = \as_i$ and $(\bar{\bs})_i = \bs_i$. Since $\as_i$ and $\bs_i$ are $d$-near, matching the suffix $\as_iB_{i+1}\as_{i+1} \cdots \as_{k-1}B_k\as_k$ with $\bs_iB'_{i+1}\bs_{i+1} \cdots \bs_{k-1}B'_k\bs_k$, we get a suffix of length $d + \sum_{l=i}^{k}|\as_l| - \sum_{l=i}^{k}|\bs_l| \leq d+ C$ either in $u$ or $v$ to be matched with the prefix of the other. Thus, $\min(|p|,|q|) \leq C + d$, and hence $\min(|p\bs_0|,|q\as_0|) \leq 2C + d < 2k\cdot \max(C,\ell)$. Hence, either $p$ or $q$ is short, contradicting the assumption that the cut is long.  Therefore, we conclude that $i \neq j$. By a similar argument $i' \neq j'$. 

WLOG we can assume that $i <i'$. We have two cases: either $j<j'$ or $j'<j$. Assume $j <j'$. Since $(\bar{\as})_i$ and $(\bar{\bs})_j$, as well as, $(\bar{\as})_{i'}$ and $(\bar{\bs})_{j'}$ are $d$-near, by matching the infix of $u$ and $v$ defined by $(\bar{\as})_i, (\bar{\as})_{i'}$ and $(\bar{\bs})_j, (\bar{\bs})_{j'}$, namely  $(\bar{\as})_iB_{i+1} \cdots B_{i'}(\bar{\as})_{i'}$ and $(\bar{\bs})_jB_{j+1}' \cdots B'_{j'}(\bar{\bs})_{j'}$, we note  that 
\begin{equation}\label{ijeq}
    \left | \, |(\bar{\as})_iB_{i+1} \cdots B_{i'}(\bar{\as})_{i'}| - |(\bar{\bs})_jB'_{j+1} \cdots B'_{j'}(\bar{\bs})_{j'}| \,  \right | \leq C + 2\cdot d
\end{equation}
If $i' > j'$, then $|B_{i'}| > 2 \cdot (|B'_1| + \cdots + |B'_{j'}| + C)$ by definition of pair $(u,v)$. Hence, 
\begin{align}
    \left | \, |(\bar{\as})_iB_{i+1} \cdots B_{i'}(\bar{\as})_{i'}| - |(\bar{\bs})_jB'_{j+1} \cdots B'_{j'}(\bar{\bs})_{j'}| \, \right | \geq 1/2\cdot |B_{i'}| > k\cdot \max(C,\ell) > C+2 \cdot d.
\end{align}
Therefore,  contradicting \Cref{ijeq}. Similarly, the case when $j' > i'$ also contradicts \Cref{ijeq}. Thus, the case of $j < j'$ leads to a contradiction. 

Next assume that $j' <j$. Observe that $(\bar{\as})_iB_{i+1} \cdots B_{i'}(\bar{\as})_{i'}$ and $(\bar{\bs})_jB'_{j+1} \cdots B'_k(\bar{\bs})_k B'_1 \cdots B'_{j'}(\bar{\bs})_{j'}$ (or equivalently the factors $(\bar{\as})_{i'}B_{i'+1} \cdots B_k(\bar{\as})_k B_1 \cdots B_{i}(\bar{\as})_{i}$ and $(\bar{\bs})_{j'}B'_{j'+1} \cdots B'_{j}(\bar{\bs})_{j}$) are the factors defined by $(\bar{\as})_i, (\bar{\as})_{i'}$ in $u$ and $(\bar{\bs})_j, (\bar{\bs})_{j'}$ in $v$). Let $d'=\left | \, |(\bar{\as})_iB_{i+1} \cdots B_{i'}(\bar{\as})_{i'}| - |(\bar{\bs})_jB'_{j+1} \cdots B'_k(\bar{\bs})_k B'_1 \cdots B'_{j'}(\bar{\bs})_{j'}|\, \right |$. Since $(\bar{\as})_i$ and $(\bar{\bs})_j$, as well as, $(\bar{\as})_{i'}$ and $(\bar{\bs})_{j'}$ are $d$-near,
\begin{equation}\label{i'jeq}
    d' \leq C + 2\cdot d.
\end{equation}
If $i'\neq k$, then $|B'_{k}| > 2 \cdot (|B_1| + \cdots + |B_{i'}| + C)$. Hence, $ d' \geq 1/2\cdot |B'_{k}| > C + 2\cdot d$. Therefore,  contradicting \Cref{i'jeq}. 
The case when $i' = k$ is symmetric to the above case by considering the factors $(\bar{\as})_k B_1 \cdots B_{i}(\bar{\as})_{i}$ and $(\bar{\bs})_{j'}B'_{j'+1} \cdots B'_{j}(\bar{\bs})_{j}$ defined by $(\bar{\as})_i, (\bar{\as})_{i'}$ in $u$ and $(\bar{\bs})_j, (\bar{\bs})_{j'}$ in $v$.

Therefore, there does not exist $i,i',j,j' \in [1,k]$, $(i,j) \neq (i',j')$, such that $(\bar{\as})_i$ and $(\bar{\bs})_j$, as well as, $(\bar{\as})_i'$ and $(\bar{\bs})_j'$ are $d$-near, i.e., there exists at most one pair $(i,j)$ such that $(\bar{\as})_i$ and $(\bar{\bs})_j$ are $d$-near. Moreover, $i \neq j$. 

For some $i \in [1,k]$, assume that $(\bar{\as})_i \in \bar{\as}$ is not $d$-near with any word in $\bar{\bs}$. This implies that in the matching of $u$ and $v$ w.r.t.~the cut $(p,q)$, the factor $\rho_i (\bar{\as})_i \rho_{i+1}$ (with $\rho_{i+1} = \rho_1$ when $i=k$) is matched within some block of $\rho_j'^*$ for some $j \in [1,k]$. Since $\rho_i \sim \rho_{i+1} \sim \rho_j'$, it follows from \Cref{prop:iwredux} that $(\bar{\as})_i$ is an inner witness of $(\bar{\rho})_i = (\rho_i,\rho_{i+1})$. Similarly, for any $i \in [1,k]$, if $(\bar{\bs})_i \in \bar{\bs}$ is not $d$-near with any word in $\bar{\as}$, then $(\bar{\bs})_i$ is an inner witness of $(\bar{\rho'})_i = (\rho'_i,\rho'_{i+1})$ (with $\rho'_{i+1} = \rho'_1$ when $i=k$). Hence, the claim is proved.
       \end{proof}

\begin{lemma}\label{interval:associative}
  Let $1\leq i < j \leq k$. Assume that   $(\bar{\as})_\ell$ is an inner witness of $(\bar{\rho})_\ell = (\rho_{\ell},\rho_{\ell+1})$ (with $\rho_{\ell+1}= \rho_1$ when $\ell=k$), for all $\ell \in [i,j]$. 
    Then, 
   $(\bar{\as})_i \cdots (\bar{\as})_j$ is an inner witness of $(\fst((\bar{\rho})_i),\snd((\bar{\rho})_j)$, and  
         moreover, 
        \begin{align}\rho_i^*(\bar{\as})_i \rho_{i+1}^*(\bar{\as})_{i+1} \cdots (\bar{\as})_{j}\rho^*_{j+1} &= \rho_i^*(\bar{\as})_i \cdots (\bar{\as})_j\\ &= (\bar{\as})_i \cdots (\bar{\as})_j \rho_{j+1}^*.
        \end{align}

Further, if $1 \leq j < i \leq k$ and $(\bar{\as})_\ell$ is an inner witness of $(\bar{\rho})_\ell = (\rho_{\ell},\rho_{\ell+1})$ (with $\rho_{\ell+1}= \rho_1$ when $\ell=k$), for all $\ell \in [i,k] \cup [1, j]$, then
 $(\bar{\as})_i \cdots (\bar{\as})_k(\bar{\as})_1\cdots (\bar{\as})_j$ is an inner witness of $(\fst((\bar{\rho})_i),\snd((\bar{\rho})_j)$, and  moreover, 
 \begin{align}\rho_i^*(\bar{\as})_i \rho_{i+1}\cdots  (\bar{\as})_{k}\rho^*_{1}(\bar{\as})_1 \rho_{2}\cdots  (\bar{\as})_{j}\rho^*_{j+1} &= \rho_i^*(\bar{\as})_i \cdots (\bar{\as})_k(\bar{\as})_{1} \cdots (\bar{\as})_j\\
 &= (\bar{\as})_i \cdots (\bar{\as})_k(\bar{\as})_{1} \cdots (\bar{\as})_j \rho_{j+1}^*.
 \end{align}

\end{lemma}
\begin{proof} 
Since by assumption  $\ \rho_\ell^{n}(\bar{\as})_\ell = (\bar{\as})_\ell\rho_{\ell+1}^{n}$ for all $\ell \in [i,j]$ and $\forall n\geq 0$, we have
    \begin{align}\rho_i^*(\bar{\as})_i \rho_{i+1}^*(\bar{\as})_{i+1}\cdots  (\bar{\as})_{j-1}\rho^*_{j}(\bar{\as})_{j}\rho^*_{j+1} 
&=\rho_i^*(\bar{\as})_i \rho_{i+1}^*(\bar{\as})_{i+1}\cdots(\bar{\as})_{j-1}\rho^*_{j}(\bar{\as})_{j}\\
&=\rho_i^*(\bar{\as})_i \rho_{i+1}^*(\bar{\as})_{i+1}\cdots \rho^*_{j-1} (\bar{\as})_{j-1}(\bar{\as})_{j}\\  &=\rho_i^*(\bar{\as})_i \cdots (\bar{\as})_{j} && (\text{By induction on the interval $[i,j]$})
        \end{align}
By an analogous argument, 
\begin{align}\rho_i^*(\bar{\as})_i \rho_{i+1}^*(\bar{\as})_{i+1}\cdots  (\bar{\as})_{j-1}\rho^*_{j}(\bar{\as})_{j}\rho^*_{j+1}  &=(\bar{\as})_i \cdots (\bar{\as})_{j}\rho_{j+1}^*. 
\end{align}
From the above, we infer that $(\bar{\as})_i \cdots (\bar{\as})_{j}$
 is an inner witness of the pair $(\rho_i, \rho_{j+1})=(\fst((\bar{\rho})_i),\snd((\bar{\rho})_j)$.

 Next assume that $1 \leq j \leq i \leq k$ and $(\bar{\as})_\ell$ is an inner witness of $(\bar{\rho})_\ell = (\rho_{\ell},\rho_{\ell+1})$ (with $\rho_{\ell+1}= \rho_1$ when $\ell=k$), for all $\ell \in [i,k] \cup [1, j]$. By the above observation, 
\begin{align}
\rho_i^*(\bar{\as})_i \rho_{i+1}\cdots  (\bar{\as})_{k}\rho^*_{1}  &=(\bar{\as})_i \cdots (\bar{\as})_{k}\rho_{1}^*\\
&=\rho_{i}^*(\bar{\as})_i \cdots (\bar{\as})_{k}\label{eq:97}
\end{align}
 \begin{align}
\rho_1^*(\bar{\as})_1 \rho_{2}\cdots  (\bar{\as})_{j}\rho^*_{j+1}  &=(\bar{\as})_1 \cdots (\bar{\as})_{j}\rho_{j+1}^*\\
&=\rho_{1}^*(\bar{\as})_1 \cdots (\bar{\as})_{j} \label{eq:99}
\end{align}
From the above equations, we get that
\begin{align}
    \rho_i^*(\bar{\as})_i \rho_{i+1}\cdots  (\bar{\as})_{k}\rho^*_{1}(\bar{\as})_1 \rho_{2}\cdots  (\bar{\as})_{j}\rho^*_{j+1} &= \rho_i^*(\bar{\as})_i \rho_{i+1}\cdots  (\bar{\as})_{k}\rho^*_{1}(\bar{\as})_1 \cdots (\bar{\as})_{j} &&\text{(Using \Cref{eq:99})}\\
    &= \rho_{i}^*(\bar{\as})_i \cdots (\bar{\as})_{k}(\bar{\as})_1 \cdots (\bar{\as})_{j}  &&\text{(Using \Cref{eq:97})}\\
    &=(\bar{\as})_i \cdots (\bar{\as})_{k}(\bar{\as})_1 \cdots (\bar{\as})_{j}\rho_{j+1}^* &&\text{(Similarly)}
\end{align}

From the above, we infer that $(\bar{\as})_i \cdots (\bar{\as})_{k}(\bar{\as})_1 \cdots (\bar{\as})_{j}$
 is an inner-witness of the pair $(\rho_i, \rho_{j+1})=(\fst((\bar{\rho})_i),\snd((\bar{\rho})_j)$.
 \end{proof}

\begin{proof}[Proof of \Cref{prop:long}]

From \Cref{cor:Mform}, we get that $M \subseteq (\as_0,\bs_0)(\rho_1,\rho_1')^* (\as_{1},\bs_{1})
           \cdots 
           (\rho_k,\rho_k')^*(\as_k,\bs_k)$
for some $\rho_1 \sim \rho_1' \sim \cdots \rho_k \sim \rho_k'$. 
By virtue of \Cref{claim:reduxform}, there are two cases.

\medskip
\noindent
\emph{Case 1: When the word $(\bar{\as})_i$ is an inner witness of $(\bar{\rho})_i$ for each $i \in [1,k]$.} Notice that by \Cref{interval:associative}, $(\bar{\as})_1 \cdots (\bar{\as})_{k}$ is an inner witness of $(\rho_1,\rho_1)$, and hence $(\bar{\as})_1 \cdots (\bar{\as})_{k} \in \rho_1^*$ . We show that $\fst(M)$ is a cyclic shift of $\rho_1^*$ as follows.  
\begin{align}
    \fst(M) &\subseteq \as_0\rho_1^*\as_1\rho^*_2\as_2 \cdots \rho^*_k\as_k\\
    &\sim \rho_1^*\as_1\rho^*_2\as_2 \cdots \rho^*_k\as_k\as_0\\
    &= \rho_1^*(\bar{\as})_1\rho^*_2(\bar{\as})_2 \cdots \rho^*_k(\bar{\as})_k\\
    &= \rho_1^*(\bar{\as})_1(\bar{\as})_2 \cdots(\bar{\as})_k &&\text{(By \Cref{interval:associative})}\\
    &\subseteq \rho_1^*. &&\text{(Since $(\bar{\as})_1 \cdots (\bar{\as})_{k}$ is an i.w of $(\rho_1,\rho_1)$) }
\end{align}
By a similar argument, we show that $\snd(M) \subsetsim \rho_1'^*$. Therefore, since $M$ is conjugate, $M \subsetsim (\rho_1,\rho_1')^*$ and hence has a common witness. Moreover, since $\rho_1' \sim \rho_1$, $\fst(M)$ and $\snd(M)$ are cyclic shifts of $\rho_1^*$.

\medskip
\noindent
\emph{Case 2: When there exists a unique pair $(i,j)$, $i\neq j\in[1,k]$, such that for all $l\in [1,k]\setminus \{i\}$, $(\bar{\as})_l$  is an inner witness of the pair $(\bar{\varrho})_l$, and for all $l \in [1,k]\setminus \{j\}$, $(\bar{\bs})_l$ is an inner witness of the pair $(\bar{\varrho}')_l$.}
\begin{align}
    \fst(M) &\subseteq \as_0\rho_1^*\as_1\rho^*_2\as_2 \cdots \rho^*_k\as_k\\
    &\sim \rho_1^*\as_1\rho^*_2\as_2 \cdots \rho^*_k\as_k\as_0\\
    &= \rho_1^*(\bar{\as})_1\rho^*_2(\bar{\as})_2 \cdots \rho^*_k(\bar{\as})_k\\
    &\sim (\bar{\as})_i\rho_{i+1}^* \cdots (\bar{\as})_k \rho_{1}^* (\bar{\as})_1 \cdots (\bar{\as})_{i-1}\rho^*_i\\
    &= (\bar{\as})_i \cdots (\bar{\as})_k(\bar{\as})_1 \cdots (\bar{\as})_{i-1}\rho^*_i &&\text{(By \Cref{interval:associative})}
\end{align}

 Therefore, we have  $\fst(M) \sim \snd(M) \subsetsim  (\bar{\as})_i \cdots (\bar{\as})_k(\bar{\as})_1 \cdots (\bar{\as})_{i-1}(\rho_i)^*$. 
 By an analogous argument,  $\fst(M) \sim \snd(M) \subsetsim  (\bar{\bs})_j \cdots (\bar{\bs})_k(\bar{\bs})_1 \cdots (\bar{\bs})_{j-1}(\rho_j')^*$, and the claim follows.

\end{proof}

% !TEX root = main.tex
\section{Computing Witness of a Sumfree Expression}\label{Subsec:computingwitness}

In this section, we give a decision procedure to compute the common witness of a sumfree expression, if it exists. A sumfree expression can have no common witness, a unique common witness, or infinitely many common witnesses. Thus, the set of common witnesses (abbreviated as the \emph{witness set}) is either empty, singleton, or infinite. Whenever there are infinitely many common witnesses for an expression, the witnesses are the same as those of its primitive root (\Cref{infmanycwcorr}). In that case we compute the primitive root as their finite representation.

The following proposition shows that there is a bound to the size of the unique common witness of two conjugate pairs if it exists, which aids in computing the common witness of two pairs in linear time.
\begin{proposition}\label{witlengthG}
If two conjugate primitive pairs $(u_1,v_1)$ and $(u_2,v_2)$ have a unique common witness $z$, then $|z| \leq 2 \cdot \max(|u_1|,|u_2|)$.
\end{proposition}

\begin{proof}
Let $z$ be a common inner witness. Therefore, $z = (x_1y_1)^{n_1}x_1 = (x_2y_2)^{n_2}x_2$ for some $n_1,n_2 \geq 0$ where $(x_i,y_i)$ is a cut of the primitive pair $(u_i,v_i)$ for $i \in \{1,2\}$. We claim either $n_1$ or $n_2$ is less than $2$. Suppose not, i.e., $n_1 \geq 2$ and $n_2 \geq 2$. Thus $u_1^\omega$ and $u_2^\omega$ share a common prefix of length at least $|u_1| + |u_2|$. From \Cref{finewilf}, they have the same primitive root. It implies that $x_1y_1 = x_2y_2$ since $u_1$ and $u_2$ are primitive words. Since $(x_1y_1)^{n_1}x_1 = (x_2y_2)^{n_2}x_2$, $x_1y_1 = x_2y_2$ and $|x_1|,|x_2| < |x_1y_1|$, we obtain $n_1 = n_2$, and hence, $x_1 = x_2$. This implies $y_1 = y_2$ and hence, $y_1x_1 = y_2x_2$. Both $(u_1,v_1)$ and $(u_2,v_2)$ are the same word; thus, they have infinitely many common witnesses, that is a contradiction. Hence $|z| \leq 2 \cdot \mathit{max}\{|u_1|,|u_2|\}$.

The case for common outer witness is symmetric.
\end{proof}
The above proposition holds true for any two conjugate pairs (not necessarily primitive) by \Cref{samewitnessext}.

\begin{proposition}\label{existswitness}
The witness set of two conjugate pairs can be computed in linear time.
\end{proposition}

\begin{proof}
Since the common witness for a set and its primitive root are the same (\Cref{samewitnessext}), we compute the witness set for the primitive roots of the given two conjugate pairs. We can compute the primitive roots in linear time by \Cref{computeroot}.
Let $G = \{(u_1,v_1), (u_2,v_2)\}$ be the set of primitive roots of the given conjugate pairs and let $(x_i,y_i)$ be the cut of $(u_i,v_i)$ for $i \in \{1,2\}$. The cut of the primitive pair $(u_i,v_i)$ can be computed in linear time by finding the first nontrivial occurrence of $v_i$ in $u_iu_i$.

One of the following possibilities holds true for $G$: it has no common witness, a unique common witness, or infinitely many common witnesses. The following algorithm outlines the computation of the witness set of $G$:
\begin{enumerate}
\item Check if the primitive pairs are identical, i.e., verify if $u_1 = u_2$ and $v_1 = v_2$. If yes, then $G$ has infinitely many common witnesses by \Cref{infmanycwcorr}. The witness is finitely represented by the primitive pair $(u_1,v_1)$. This step takes linear time w.r.t.~the length of the primitive pairs.
\item If the pairs are not identical, then check if $G$ has a unique common witness using \Cref{witlengthG} as follows: WLOG assume that $|u_1| > |u_2|$. According to \Cref{witlengthG}, if a unique common witness exists for $G$, its length is at most $2 \cdot \max(|u_1|,|u_2|) = 2 \cdot |u_1|$. Thus, it suffices to check whether $(x_1y_1)^\omega$ and $(x_2y_2)^\omega$ share equal prefixes of length $|x_1|$ or $|x_1y_1x_1|$, that also end in $x_2$. If it is satisfied, then $G$ has a unique common witness. This step can be performed in linear time w.r.t.~the length of the primitive pairs.
\item If none of the above holds, then $G$ has no common witness.
\end{enumerate}
The overall complexity of the algorithm is linear w.r.t~the length of the primitive roots of the given pairs.
\end{proof}

The witness set of a sumfree expression is equal to the intersection of witness sets of each of its singleton reduxes. So first, we show how to compute the witness set of a sumfree expression with only one Kleene star, in effect the witness set of a singleton redux. Using this procedure, we show how to compute the witness set of a general sumfree expression.
\begin{lemma}\label{computeMwit}
Let $M= (\as_0,\bs_0)E^*(\as_1,\bs_1)$ be a sumfree expression. Given the witness set of $E$, we can compute the witness set of $M$ in time $\mathcal{O}(m + n)$ where $m$ is the size of the expression $M$, and $n$ is the size of the witness of $E$.
\end{lemma}
\begin{proof}
If redux of $M$ is empty, then the witness set of $M$ is same as the witness set of $E$ by \Cref{conjugacy}. Now assume that $M$ has a nonempty redux. From \Cref{singlemonomial}, $M$ has a common witness iff $E \cup \{(\as_1\as_0,\bs_1\bs_0)\}$ has a common witness. The common witness of $M$ is computed from the common witness of $E$ and $(\as_1\as_0,\bs_1\bs_0)$.

The idea is to check if a common witness exists for $E \cup \{(\as_1\as_0,\bs_1\bs_0)\}$ using \Cref{gmonomialcut}. If it exists, using that we compute the common witness for $M$.
There are two possibilities for common witnesses of $E$:
\begin{enumerate}
\item\label{Mwit:1} $E$ has a unique common inner (\textit{resp.}~outer) witness $z$. By \Cref{gmonomialcut:2} of \Cref{gmonomialcut}, it suffices to check if $z$ is a common inner (\textit{resp.}~outer) witness of $(\as_1\as_0,\bs_1\bs_0)$. This can be checked in $\mathcal{O}(m+n)$ time using \Cref{uz=zv}. If so, $z$ is the common witness of $E \cup \{(\as_1\as_0,\bs_1\bs_0)\}$. Now compute the common witness of $M$ using \Cref{singlemonomial}\ref{Thm7:1} (\textit{resp.}~\Cref{singlemonomial}\ref{Thm7:2}). This can be computed in $\mathcal{O}(m+n)$ time. Otherwise, $E \cup \{(\as_1\as_0,\bs_1\bs_0)\}$ has no common witness and hence, $M$ has no common witness by \Cref{singlemonomial}.

\item $E$ has infinitely many common witnesses. In this case, the witnesses of $E$ are the same as that of its primitive root, say $(\rho,\rho^\prime)$. From \Cref{gmonomialcut:1} of \Cref{gmonomialcut}, it suffices to check if $(\rho,\rho^\prime)$ and $(\as_1\as_0,\bs_1\bs_0)$ has a common witness. For this, first check if the primitive root of $(\as_1\as_0,\bs_1\bs_0)$ is equal to $(\rho,\rho^\prime)$. This step takes time $\mathcal{O}(m +n)$. We have two cases:
\begin{enumerate}[label=(\alph*)]
\item If $(\as_1\as_0,\bs_1\bs_0)$ have same primitive root as that of $E$, then $E$ and $(\as_1\as_0,\bs_1\bs_0)$ have infinitely many common witnesses by \Cref{infmanycwcorr}. In this case, $M$ is a set of powers of the primitive root of the redux by \Cref{singlemonomial}\ref{Thm7:3}. Thus $M$ has infinitely many witnesses. Compute the primitive root of its redux using \Cref{computeroot}. This step takes $\mathcal{O}(m)$ time.

\item Otherwise, compute the unique common witness of $(\rho,\rho^\prime)$ and $(\as_1\as_0,\bs_1\bs_0)$ if it exists using \Cref{existswitness}. If so, we are back to \emph{Case 1}; else $M$ has no common witness. This step takes $\mathcal{O}(m+n)$ time.
\end{enumerate}
\end{enumerate}
\end{proof}

Using the above algorithm, we compute the common witness of a general sumfree expression as follows.
\begin{lemma}\label{computeEwit}
Let $M$ be a sumfree expression. Given the witness set of each Kleene star in $M$, we can compute the witness set of $M$ in time $\mathcal{O}(m \cdot (m+n))$ where $m$ is the size of the expression and $n$ is the maximum size among the given witnesses.
\end{lemma}

\begin{proof}

From \Cref{mcommonwitness}, the witness set of $M$ is the intersection of the witness sets of its singleton reduxes. Check if each singleton redux of $M$ has a common witness and compute it using \Cref{computeMwit}. This step takes $\mathcal{O}(m \cdot (m+n))$. If there is a singleton redux with no common witness, then $M$ has no common witness by \Cref{generalmonomialwitness}. 

Assume that $M$ has a nonempty redux. The algorithm is as follows. Firstly, check if the redux of $M$ is conjugate using \Cref{isconjugate} (in time $\mathcal{O}(m)$). If yes, then compute the primitive root of the redux, say $(\rho_m,\rho_m^\prime)$, using \Cref{computeroot} (in time $\mathcal{O}(m)$). Otherwise, $M$ has no common witness. There are two cases depending on whether all singleton redux has infinitely many witnesses.
\begin{enumerate}[label=(\alph*)]
\item If all the singleton reduxes have infinitely many witnesses, then $M$ is a set of powers of the primitive root of the redux $(\rho_m,\rho_m^\prime)$ by \Cref{singlemonomial}\ref{Thm7:3}. Thus, $M$ has infinitely many common witnesses.

\item If there exists a singleton redux with a unique common witness, say $z$, then for all other singleton reduxes of $M$ with a unique witness $z^\prime$, check if $z = z^\prime$ (for all other singleton reduxes $z$ is already a witness by virtue of being a witness of the redux of $M$). If so, $z$ is the unique common witness of $M$; otherwise $M$ has no common witness.
\end{enumerate}
This takes $\mathcal{O}(mn + m)$ time.
Consider the case where the redux of $M$ is empty. Let $M$ be of the form $E_1^*E_2^* \cdots E_k^*$.  We iterate over $i \in \{1,\ldots, k-1\}$.
\begin{enumerate}[label=(\alph*)]
\item If both $E_i$ and $E_{i+1}$ has infinitely many witnesses, then the witness set of $E_i^*$ and $E_{i+1}^*$ is finitely represented by their primitive roots. Check if the primitive roots are the same. If yes, continue the procedure for $i=i+1$. Otherwise, check and compute the common witness, say $z$, of the primitive roots if exists using  \Cref{existswitness} (in time $\mathcal{O}(n)$). If not, then $M$ has no common witness. Otherwise, check if $z$ is a common witness of all singleton reduxes $E_{i+2}^*, \ldots, E_k^*$ (This can be done in $\mathcal{O}(n)$ for each singleton redux --- If a singleton redux has a unique witness then check if it is equal to $z$. Otherwise, if it has infinitely many witness, then it is represented by a primitive root. It suffices to check if $z$ is a witness of that primitive root).
\item If either one of $E_i$ and $E_{i+1}$ has a unique witness, say $z$, then check if $z$ is a common witness of all singleton reduxes $E_{i+2}^*, \ldots, E_k^*$ . 
\end{enumerate}
This step takes $\mathcal{O}(mn + n)$ time. Overall, it takes $\mathcal{O}(m \cdot (m+n))$.
\end{proof}

\subparagraph*{Computation of the Witness Set:} Given a sumfree expression $M$, we compute its witness set bottom-up. We start from the innermost Kleene star. It is a pair of words $(u,v)$. First, we check if $(u,v)$ is conjugate using \Cref{isconjugate}. If yes, then there are infinitely many common witnesses for $(u,v)^*$, namely the witnesses of its primitive root, otherwise $M$ has no witness. This step can be done in a time linear in the length of $(u,v)$. Now we recursively use \Cref{computeEwit} to compute the common witness of the expression under the Kleene star in each level. If there is no common witness for any level of Kleene star expression, then $M$ is not conjugate.

To find out the complexity of the decision procedure, it suffices to estimate the maximum length of a witness involved in the computation.

\subparagraph*{Length of the Witness of a Sumfree Expression:} We claim that if a sumfree expression $M$ is conjugate, then there exists a witness of length linear in size of $M$. If $M$ has infinitely many witnesses, from \Cref{infmanycwcorr}, $M$ is a set of powers of a primitive root. Therefore, there exists a witness of length that is less than that of the length of the primitive root. Next, suppose $M$ has a unique common witness. In that case, there exists a subexpression $E_i^*$ such that $E_i^*$ has a unique common witness, and all Kleene star appearing in $E_i$ has infinitely many witnesses. Thus, all of them have a common witness of length at most $|E_i|$. Therefore, there is a singleton redux $M_i$ of $E_i^*$ that has a unique witness $z_i$. The size of $z_i$ is linear in $M_i$ and the size of the witnesses of subexpressions of $E_i$. Both are upper bounded by size of $M$. Furthermore, the common witnesses for all subsequent levels is unique (if it exists) and its length is bounded by $|M|$.

\subparagraph*{Complexity of the Algorithm:} Since the size of the common witness of $M$ is linear in $|M|$, by \Cref{computeEwit}, the overall complexity of computing a common witness of a sumfree expression is $\mathcal{O}(h \cdot m^2)$ where $h$ is the \emph{star height} of $M$ and $m$ is the length of the expression.

\section{Deciding Conjugacy of a Sumfree Expression}\label{sec:decidingconj}
Given a sumfree expression $E=(\as_0,\bs_0)E_1^*(\as_1,\bs_1)\dots E_k^*(\as_k,\bs_k)$, $k\geq 1$, the following algorithms decides  whether $E$ is conjugate.
\begin{enumerate}
    \item Check whether $E$ has a common witness using the algorithm described in Section 7. If there exists a common witness, then output YES.
    \item Otherwise, verify the following conditions. If any of them fails, output NO.
    \begin{enumerate}
    \item $\as_0\dots \as_k \sim \bs_0\dots \bs_k$.
        \item For $i\in [1,k]$, each $E_i^*$ is conjugate and moreover it has a primitive root $(\rho_i,\rho_i')$.
    \item $\rho_i \sim \rho_j$ for all $i,j \in [1,k]$.
    \item Let $\bar{\alpha}=(\as_1,\ldots, \as_{k-1},\as_k\as_0)$ and $\bar{\varrho} =(\rho_1,\rho_2), \ldots, (\rho_{k-1},\rho_k), (\rho_k,\rho_1)$. 
 Then, there is at most one element $i\in[1,k]$ such that $(\bar{\alpha})_i$ is \emph{not} an inner witness of $(\bar{\varrho})_i$. 
     \item Let $\bar{\beta}=(\bs_1,\ldots, \bs_{k-1},\bs_k\as_0)$ and $\bar{\varrho}' =(\rho'_1,\rho'_2), \ldots, (\rho'_{k-1},\rho'_k), (\rho'_k,\rho'_1)$. 
     Then there is at most one element $j\in[1,k]$ such that $(\bar{\beta})_j$ is \emph{not} an inner witness of $(\bar{\varrho}')_j$.
    \end{enumerate} 
    \item If $i$ and $j$ do not exist in steps (d) and (e), then output YES.
    \item Otherwise, output YES if the below expression is conjugate.
    \begin{align}
        \left(\begin{array}{c}
             \as_i\dots \as_k\as_0\dots \as_{i-1} \\
              \bs_j\dots \bs_k\bs_0\dots \bs_{j-1} 
\end{array}\right)\left(\begin{array}{c}
            \rho_i \\
             \rho_j' 
        \end{array}\right)^*
    \end{align}
    Note that the above expression is conjugate iff the set $\{(\as_i\dots \as_k\as_0\dots \as_{i-1},\bs_j\dots \bs_k\bs_0\dots \bs_{j-1}), (\rho_i,\rho_j')\}$ has a common witness.
    \item In all other cases, Output NO.
 \end{enumerate}
 The correctness of the above algorithm follows from \Cref{thm:conjugateMonoidClosure} and \Cref{prop:long}. In the above algorithm, Steps $1$, $2.(b)$, and $4$ amounts to computing the common witness of a sumfree expression, and therefore is in polynomial time. Now, Steps $2.(a)$ and $2.(c)$ is checking conjugacy of two words, which is in linear time. Steps $2.(d)$ and $2.(e)$ are in linear time. Therefore, the algorithm runs in polynomial time.

\section{Conclusion}
\label{Sec:conclusion}
It is shown that the conjugacy problem of a rational relation is decidable. The current decision procedure proceeds through the analysis of rational expressions. In its essence, it is analogous to the boundedness checking of distance automata using factorisation trees \cite{Colcombet21}, though explicit use of factorisation trees are avoided  using sumfree rational expressions instead. An obvious question is the existence of an automata-theoretic proof. Factorisation forests remain the primary tool to settle boundedness questions on automata and by that standard the proof approach taken in this paper is natural and quite possibly the most intuitive.

Computing a witness of a given sumfree expression, if one exists, can be done in polynomial time. However, converting a rational expression into a sum of sumfree expressions may result in an exponential blowup. Thus, the algorithm presented in the paper is of exponential time. It remains to find the precise complexity of this problem.

It is natural to look at the conjugacy problem of more general classes, for instance functions definable by a deterministic two-way transducers (regular functions \cite{Engelfriet}), or by two-way pebble automata (polyregular functions \cite{bojanczyk2022transducers}).

\appendix
\section{}
\label{app1}

\begin{lemma}\label{sumfreere}
Every rational expression $E$ can be converted to one in sumfree normal form $E'$ in exponential time. Moreover, $|E'| \leq 2^{2 \cdot |E|}$.
\end{lemma}
\begin{proof}
Let $E$ be a rational expression over the monoid $\bfM$. 
We assume that the rational expression $E$ is given as a tree $e$. We take the size of $e$, denoted as $|e|$, to be the number of nodes in the tree. We inductively define a tree $e’$ that has the same language as the sumfree normal form of the expression $E$ and furthermore, as shown in \Cref{snftree}, it is in the shape of a right-comb with the internal nodes of the spine labelled with $+$’s (and the leaf of the spine is labelled with $\emptyset$) and the pendant left subtrees attached to the spine are sumfree. We call $e'$ as the SNF tree of $e$.

\begin{figure}[t]

\centering
\begin{tikzpicture} 
	\node {$+$}
                   child {node{$e_1$}}
		child {
                           node {$+$}
                           child{node{$e_2$}}
                           child {
                                    node {$+$}
                                    child{node{$e_3$}}
                                    child{node {$\emptyset$}}
                                    }};
\end{tikzpicture}
\caption{SNF tree for the SNF expression $E_1 + E_2 + E_3$.}

\label{snftree}
\end{figure}

We obtain an equivalent sumfree normal form expression and its expression tree $e'$ by induction on the structure of $E$. We prove the following invariant along with the construction of $e’$.

\begin{claim}
$|e’| \leq 2^{2|e|}$
\end{claim}

The following definition is used in the analysis below. Let $N(e’)$ denote the number of summands in $e’$, i.e., $N(e’)$ is the number of nodes in the spine of the comb, or equivalently, 1 more than the number of nodes labelled with $‘+’$ in $e’$. Hence $N(e’) \leq |e’|$.

\paragraph*{Base Case} When $E$ is $\emptyset$ or $m\in M$, then $E$ is already sumfree. The tree $e$ corresponds to a tree with a single node. We take $e’$ to be the tree with 3 nodes in SNF with the left subtree of the root being $e$. Hence $|e’| = 3$ and the claim holds.

\paragraph*{Inductive Case} Assume that $G$ and $H$ are rational expressions with expression trees $g$ and $h$ respectively. Let $g’$ and $h’$ denote their SNF trees. By induction hypothesis, $G \equiv \alpha_1 + \cdots + \alpha_k$ and $F \equiv \beta_1 + \cdots + \beta_\ell$ such that $\alpha_i, \beta_j$, $1\leq i \leq k$, $1 \leq j \leq \ell$, are sumfree expressions. Also, $|g’| \leq 2^{2|g|}$ and $
|h’| \leq 2^{2|h|}$.

\begin{enumerate}
\item
If $E = G + H$, then by substituting for $G$ and $H$, we get an equivalent expression of the desired form. This step takes constant time. To obtain $e’$ we replace the leaf of the spine of $g’$ with the root of $h’$. Clearly,
\begin{align*}
|e’| &= |g’| + |h’| - 1\\
&\leq 2^{2|g|} + 2^{2|h|} - 1\\
&\leq 2^{2(|g|+|h|+1)}\\
&= 2^{2|e|}.
\end{align*}

\item
If $E = G\cdot H$, then by substituting $G$ and $H$ we get
$E \equiv (\alpha_1 + \cdots + \alpha_k) \cdot (\beta_1 + \cdots + \beta_\ell)$. Distributing the monoid operation over the union, we get $E \equiv (\alpha_1\beta_1 + \cdots + \alpha_1\beta_\ell) + \cdots + (\alpha_k\beta_1 + \cdots + \alpha_k\beta_\ell) $, that is in the required form. This step takes time quadratic in the maximum among the length of the SNF expressions $G$ and $H$.

Assume there are $p$-many (resp. $q$-many) pendant subtrees attached to the spine of $g’$ (resp. $h’$). The tree $e’$ is a right-comb with $pq$-many pendant subtrees where each subtree is obtained by the pairwise concatenation of pendant subtrees from $g’$ and $h’$ respectively. Clearly, $N(e’) = N(g’)N(h’) - 1$.
\begin{align*}
|e’| &\leq N(g’)N(h’) + |g’|N(h’) + |h’|N(g’) + N(g’)N(h’)\\
&\leq 4|g’||h’|\\
&\leq 4 \cdot 2^{2|g|}2^{2|h|}\\
&\leq 2^{2(|g| + |h| + 1)}\\
&= 2^{2|e|}.
\end{align*}

\item
Finally, if $E = G^*$, then by repeatedly applying the rational identity $(X+Y)^* = (X^* Y^*)^*$, where $X,Y$ are rational expressions, we get $E = G^* \equiv (\alpha_1 + \alpha_2 + \cdots + \alpha_k)^* = (\alpha_1^*\alpha_2^* \cdots \alpha_k^*)^*$. This step takes linear time w.r.t.~the length of the SNF epression $G$.

We obtain the tree $g’$ corresponding to $g$, and construct a new tree $h$ from $g’$ as follows.
\begin{itemize}
\item Add an intermediate node labelled with $*$  between each pendant subtree and the spine of $g’$.
\item Replace each $+$ labelled nodes in the spine with concatenation.
\item Replace $\emptyset$ in the leaf of spine with epsilon.
\item Add a new root node labelled with $*$.

\end{itemize}
Now, $e’$ is obtained by attaching $h$ as the left subtree of a right-comb in the desired form. Clearly $N(e’) = 1$.
\begin{align*}
|e’| &\leq |g’|+N(g’)+3\\
&\leq |g’|+|g’|+3\\
&\leq 2 \cdot 2^{2|g|}+3\\
&= 2^{2|g|+1}+3\\
&\leq 2 \cdot 2^{2|g|+1} &&\text{(Since $|g| \geq 1, 2^{2|g|+1} \geq 8$)}\\
&= 2^{2(|g|+1)}\\
&= 2^{2|e|}
\end{align*}
\end{enumerate}

Hence proved that the upper bound on the size of the SNF expression is exponential in the size of the given expression.

Each step of constructing an SNF expression takes polynomial time in the length of its constituent SNF expressions. Therefore, any rational expression can be converted to an equivalent sumfree normal form in exponential time.

\end{proof}

\bibliographystyle{elsarticle-num} 

\bibliography{reference}

\end{document}